\documentclass[a4paper,UKenglish]{CSML}

\pdfoutput=1

\usepackage{lastpage}

\def\dOi{13(4:16)2017}
\lmcsheading%
{\dOi}
{1--\pageref{LastPage}}
{}
{}
{May\phantom.~31, 2016}
{Nov.~28, 2017}
{}

\usepackage[hidelinks]{hyperref}
\usepackage[safe]{tipa}
\usepackage{amssymb}
\usepackage{stmaryrd}
\usepackage{xspace}
\usepackage{amssymb}
\usepackage{proof}
\usepackage{mathpartir}
\usepackage{framed}
\usepackage{alltt}
\usepackage{url}
\usepackage{mathdots}
\usepackage[all]{xy}
\usepackage{pifont} 

\usepackage{thmtools, thm-restate}

\usepackage{float}
\floatstyle{boxed} 
\restylefloat{figure}

\newcommand{\tvi}[0]{\vrule height 12pt depth 5pt width 0pt}
\newcommand{\FF}[0]{\mathcal{F}}
\newcommand{\FFs}[1]{\mathcal{F}_{#1}}
\newcommand{\FFgen}[1]{\overline{\FF}}
\newcommand{\XX}[0]{\mathcal{X}}
\newcommand{\RR}[0]{\mathcal{R}}
\newcommand{\RRM}[0]{{\mathcal{R}_{\mathrm M}}}
\newcommand{\SO}[0]{\mathcal{S}} 
\newcommand{\sst}[0]{s} 

\newcommand{\cmark}{\ding{51}}%
\newcommand{\xmark}{\ding{55}}%
\newcommand{\ie}[0]{\textit{i.e.}\xspace}
\newcommand{\eg}[0]{\textit{e.g.}\xspace}
\newcommand{\wrt}[0]{\textit{w.r.t.}\xspace}

\newcommand{\ande}[0]{\texttt{and}}
\newcommand{\ore}[0]{\texttt{or}}
\newcommand{\note}[0]{\texttt{not}}
\newcommand{\impe}[0]{\texttt{imp}}
\newcommand{\atom}[1]{\texttt{#1}}
\newcommand{\atoms}[0]{\texttt{at}}
\newcommand{\atomc}[1]{\atoms(#1)}
\newcommand{\dnf}[0]{\texttt{DNF}}
\newcommand{\dnfr}[0]{\texttt{DNFr}}
\newcommand{\eqdef}[0]{=}
\newcommand{\ra}[0]{\rightarrowtriangle}
\newcommand{\sarrow}[0]{\,\mapsto\,}
\newcommand{\vpn}[1]{\varphi_{#1}}
\newcommand{\phiseq}[0]{\vpn{\seqsym}}
\newcommand{\phich}[0]{\vpn{\gets\!\!\!\!+}}%
\newcommand{\vpg}[0]{\vpn{ }}
\newcommand{\vpS}[0]{\vpn{S}}
\newcommand{\vpo}[0]{\vpn{S_1}}
\newcommand{\vpt}[0]{\vpn{S_2}}
\newcommand{\psif}[1]{\varphi_{f_{#1}}}
\newcommand{\psiF}[0]{\varphi_{f}}

\newcommand{\dummy}[0]{{\bot(\any)}}
\newcommand{\bott}[1]{\bot}

\newcommand\Vtextvisiblespace[1][.5em]{%
  \mbox{\kern.06em\vrule height.3ex}%
  \vbox{\hrule width#1}%
  \hbox{\vrule height.3ex}}
\newcommand{\arty}[0]{ar}
\newcommand{\any}[0]{\rule{5pt}{0.6pt}\kern.08em}
\newcommand{\ap}[0]{{\mbox{\Large!}}}
\newcommand{\apt}[1]{{\mbox{\Large!}^{#1}}\!}
\newcommand{\at}[0]{\mathbin{@}}

\newcommand{\eqto}[0]{=}
\newcommand{\sigsep}{,\;}
\newcommand{\sigd}{\,:\,}
\newcommand{\prosep}{{\times}}
\newcommand{\ctxPlusZero}[0]{\Gamma_{pz}}

\newcommand{\tryShort}[0]{\Try(\seq{\Plus(\Zero,x)\ra x}{X})}
\newcommand{\subj}[0]{\Plus(\Zero,\Plus(\Zero,\Succ(\Zero)))}
\newcommand{\FFmeta}{\FF_\rawappl}
\newcommand{\TFmeta}{\mathcal{T}(\FFmeta)}
\newcommand{\TFXmeta}{\mathcal{T}(\FFmeta, \XX)}
\newcommand{\trmeta}[1]{\mbox{\textopencorner} #1 \mbox{\textcorner}}
\newcommand{\rawappl}{\mathtt{appl}}
\newcommand{\appl}[2]{\mathtt{\rawappl}({#1},{#2})}
\newcommand{\cons}[2]{{#1}::{#2}}
\newcommand{\conss}[0]{::}
\newcommand{\nil}[0]{\mathtt{nil}}
\newcommand{\symb}[1]{\trmeta{#1}}
\newcommand{\rawbotlist}{\bot_{list}}
\newcommand{\botlist}[1]{\rawbotlist(#1)}
\newcommand{\rawrconcat}{\mathtt{rconcat}}
\newcommand{\rawrev}{\mathtt{rev}}
\newcommand{\rawappend}{\mathtt{append}}
\newcommand{\rawpropag}{\mathtt{propag}}
\newcommand{\rconcat}[2]{\rawrconcat({#1},{#2})}
\newcommand{\reverse}[1]{\rawrev({#1})}
\newcommand{\append}[2]{\rawappend({#1},{#2})}
\newcommand{\propag}[1]{\rawpropag({#1})}
\newcommand{\f}[0]{{f}}
\newcommand{\head}[0]{{h}}
\newcommand{\tail}[0]{{q}}
\newcommand{\args}[0]{{args}}

\newcommand{\revtried}[0]{{r\_tried}}
\newcommand{\revdone}[0]{{r\_done}}
\newcommand{\todo}[0]{{todo}}
\newcommand{\llist}{list}
\newcommand{\rawTRm}{\mathbb{T}_{\mathcal M}}
\newcommand{\TRm}[1]{\rawTRm({#1})}
\newcommand{\TRmg}[1]{\mathbb{B}_{\mathcal M}({#1})}
\newcommand{\TF}[0]{\mathcal{T}({\mathcal{F}})}
\newcommand{\TFs}[1]{\mathcal{T}_{#1}(\mathcal{F})}
\newcommand{\TFsf}[2]{\mathcal{T}_{#1}({#2})}
\newcommand{\TFX}{\mathcal{T(F,X)}}
\newcommand{\TFXs}[1]{\mathcal{T}_{#1}\mathcal{(F,X)}}
\newcommand{\XS}{\mathcal{X_S}}
\newcommand{\PPos}[0]{\mathcal{P}os}
\newcommand{\stt}[2]{\ensuremath{#1_{|#2}}}
\newcommand{\rmp}[3]{\ensuremath{{#1\left[#3\right]_{#2}}}}
\newcommand{\var}[1]{\mathcal{V}ar\left({#1}\right)}

\newcommand{\elan}{\textsc{Elan}}
\newcommand{\tom}{\textsc{Tom}}
\newcommand{\aprove}{\textsc{AProVE}}
\newcommand{\TTT}{{TTT2}}
\newcommand{\java}{{Java}}
\newcommand{\C}{\texttt{C}}
\newcommand{\maude}{{Maude}}
\newcommand{\stratego}{{Stratego}}
\newcommand{\strategyanalyser}{\textsc{StrategyAnalyser}}
\newcommand{\step}[2]{{#1}\sred{#2}}

\newcommand{\stratapp}[2]{\ensuremath{ #1 \circ #2}}
\newcommand{\stratappctx}[3]{\ensuremath{ #1 \vdash #2 \circ #3}}
\newcommand{\sred}[0]{\Longrightarrow}

\newcommand{\applysubs}[2]{#1(#2)}
\newcommand{\multieval}[1]{\mathop{\trs_{#1}}}
\newcommand{\TR}[2]{\mathbb{T}_{#1}(#2)}
\newcommand{\TRr}[1]{\mathbb{T}(#1)}
\newcommand{\TRrND}[1]{\mathbb{T}_{nd}(#1)}
\newcommand{\TRctx}[2]{\mathbb{T}({#2})}
\newcommand{\TRctxND}[2]{\mathbb{T}_{nd}({#2})}
\newcommand{\TRctxS}[2]{\mathbb{T}_{\SO}({#2})}
\newcommand{\TRg}[2]{\mathbb{B}({#2})}
\newcommand{\TRgS}[2]{\mathbb{B}_{\SO}({#2})}
\newcommand{\TRgg}[1]{\mathbb{B}({#1})}
\newcommand{\TRgGG}[0]{\TRg{\Gamma}{\Gamma}}
\newcommand{\emptyctx}[0]{\oslash}
\newcommand{\TRlists}[0]{\mathbb{L}}
\newcommand{\id}[0]{\mathit{Identity}}
\newcommand{\fail}[0]{\mathit{Fail}}
\newcommand{\dotsym}[0]{\mathbin{.}}
\newcommand{\seqsymbase}[0]{;} 
\newcommand{\seqsym}[0]{\mathbin{\seqsymbase}} 
\newcommand{\seq}[2]{{#1;#2}}
\newcommand{\seqsub}[2]{{#1\seqsymbase#2}}
\newcommand{\choicesym}[0]{\mathbin{\gets\!\!\!\!\!\!+}} 
\newcommand{\choice}[2]{#1 \choicesym #2}
\newcommand{\choicesub}[2]{#1 \gets\!\!\!\!+ #2}
\newcommand{\all}[0]{\mathit{All}}
\newcommand{\one}[0]{\mathit{One}}
\newcommand{\Try}[0]{\mathit{Try}}
\newcommand{\Repeat}[0]{\mathit{Repeat}}
\newcommand{\OnceBottomUp}[0]{\mathit{OnceBottomUp}}
\newcommand{\OBU}[0]{\mathit{obu}}
\newcommand{\BottomUp}[0]{\mathit{BottomUp}}
\newcommand{\OnceTopDown}[0]{\mathit{OnceTopDown}}
\newcommand{\TopDown}[0]{\mathit{TopDown}}
\newcommand{\Innermost}[0]{\mathit{Innermost}}
\newcommand{\rgf}[0]{\ensuremath{\mathsf{gfx}}}
\newcommand{\rf}[0]{\ensuremath{\mathsf{hx}}}
\newcommand{\dist}[0]{\ensuremath{\mathsf{dist}}}
\newcommand{\factorial}[0]{\ensuremath{\mathsf{fact}}}
\newcommand{\tgf}[0]{t_{\text{gf}}}
\newcommand{\im}[0]{\ensuremath{\mathsf{innermost}}}

\newcommand{\reps}[0]{\ensuremath{\mathsf{repeat}}}

\newcommand{\obu}[0]{\ensuremath{\mathsf{obu}}}
\newcommand{\bu}[0]{\ensuremath{\mathsf{bu}}}
\newcommand{\td}[0]{\ensuremath{\mathsf{td}}}
\newcommand{\tdsos}[0]{\ensuremath{\mathsf{tdStopOnSuccess}}}
\newcommand{\rbu}[0]{\ensuremath{\mathsf{rbu}}}

\newcommand{\bup}[0]{\bu(\propagate)}
\newcommand{\buptwo}[0]{\bu(\propagate{2})}
\newcommand{\propagate}[0]{\ensuremath{\mathsf{propagate}}}
\newcommand{\compile}[0]{\ensuremath{\mathsf{compile}}}
\newcommand{\rename}[0]{\ensuremath{\mathsf{rename}}}

\newcommand{\failres}{\texttt{Fail}\xspace}
\newcommand{\combs}[1]{${#1}$}
\newcommand{\Zero}[0]{\texttt{Z}}
\newcommand{\Succ}[0]{\texttt{S}}
\newcommand{\Plus}[0]{\texttt{+}}
\newcommand{\Mult}[0]{\texttt{*}}
\newcommand{\varS}[1]{\mathcal{FV}ar_{\mathcal X}\left({#1}\right)}
\newcommand{\domS}[1]{\mathcal{D}om_{\mathcal X}\left({#1}\right)}

\newcommand{\ints}[0]{\texttt{Nat}}
\newcommand{\bool}[0]{\texttt{Bool}}
\newcommand{\boolbi}[0]{\mathbb{B}}
\newcommand{\odd}[0]{\texttt{odd}}
\newcommand{\even}[0]{\texttt{even}}
\newcommand{\false}[0]{\texttt{false}}
\newcommand{\true}[0]{\texttt{true}}
\newcommand{\falsebi}[0]{\textrm{ff}}
\newcommand{\truebi}[0]{\textrm{tt}}
\newcommand{\trs}{\longrightarrow\!\!\!\!\!\rightarrow}
\newcommand{\trsredstar}[3]{{#2}\stackrel{#1}{\trs}{#3}}
\newcommand{\trsred}[3]{#1\bullet{#2}\trs{#3}}
\newcommand{\trsrednl}[3]{#1 \\\tag*{$\bullet{#2} \trs{#3}$}}
\newcommand{\trsrednll}[3]{#1 \bullet{#2} \\ \tag*{$\trs{#3}$}}

\newcommand{\supzero}[0]{+}
\newcommand{\Xx}[0]{x}
\newcommand{\Yy}[0]{y}
\newcommand{\Zz}[0]{z}
\newcommand{\Xxs}[2]{\Xx^{#2}_{#1}}
\newcommand{\Yys}[2]{\Yy^{#2}_{#1}}
\newcommand{\varX}[0]{X}
\newcommand{\vare}[0]{\texttt{var}}
\newcommand{\fresh}[0]{\texttt{fresh}}

\begin{document}

\title[(Meta-)encodings of programmable strategies into term rewriting
systems]
{Faithful (meta-)encodings of programmable strategies into term rewriting systems\rsuper*}
\author[H.~Cirstea]{Horatiu Cirstea}
\author[S.~Lenglet]{Sergue{\"i} Lenglet}
\author[P.-E.~Moreau]{Pierre-Etienne Moreau}
\address{Universit{\'e} de Lorraine, CNRS, LORIA, UMR 7503,
F-54506 Vand{\oe}uvre-l{\`e}s-Nancy, France}
\email{\{firstname.lastname@loria.fr\}}



\titlecomment{{\lsuper*}This is an extension of ``A faithful encoding of programmable strategies into term rewriting systems'' presented at RTA 2015}

\maketitle 
\begin{abstract}
  Rewriting is a formalism widely used in computer science and
  mathematical logic. When using rewriting as a programming or
  modeling paradigm, the rewrite rules describe the transformations
  one wants to operate and rewriting strategies are used to control
  their application. The operational semantics of these strategies are
  generally accepted and approaches for analyzing the termination of
  specific strategies have been studied. We propose in this paper a
  generic encoding of classic control and traversal strategies used in
  rewrite based languages such as {\maude}, {\stratego} and {\tom}
  into a plain term rewriting system.
  The encoding is proven sound and complete and, as a direct
  consequence, established termination methods used for term rewriting
  systems can be applied to analyze the termination of strategy
  controlled term rewriting systems.
  We show that the encoding of strategies into term rewriting systems
  can be easily adapted to handle many-sorted signatures and we use a
  meta-level representation of terms to reduce the size of the
  encodings.
  The corresponding implementation in {\tom} generates term rewriting
  systems compatible with the syntax of termination tools such as
  {\aprove} and {\TTT}, tools which turned out to be very effective in
  (dis)proving the termination of the generated term rewriting
  systems. The approach can also be seen as a generic strategy
  compiler which can be integrated into languages providing pattern
  matching primitives; experiments in {\tom} show that applying our
  encoding leads to performances comparable to the native {\tom}
  strategies.
\end{abstract}

\section{Introduction}
\label{se:intro}
Term rewriting is a very powerful tool used in theoretical studies as
well as for practical implementations. It is used, for example, in
order to describe the meaning of programming
languages~\cite{meseguer2004,RosuS10}, and also to describe by
inference rules a logic~\cite{GirardLafontTaylor89}, a theorem
prover~\cite{JouannaudKirchnerSIAM86}, or a constraint
solver~\cite{JouannaudKirchner-rob91}.  Besides its use for specifying
and implementing such formalisms, it is also used as an underlying
computation mechanism in systems like Mathematica~\cite{Wolfram2003},
{\maude}~\cite{Maude2:03}, or
{\tom}~\cite{MoreauRV-2003,BallandBKMR-RTA2007}, where rewrite rules
are objects of the language: they can be defined by the user and
manipulated by other constructs.

Rewrite rules, the core concept in term rewriting, consist of a
pattern that describes a schematic situation and the transformation
that should be applied in that particular case.
For example, the following set of rewrite rules specify how to
transform a Boolean expression into its corresponding disjunctive
normal form (DNF):
\[
\begin{array}{l}
 \ande(\ore(x, y), z) \ra \ore(\ande(x, z), \ande(y, z)) \\
 \ande(z, \ore(x, y)) \ra \ore(\ande(z, x), \ande(z, y))
\end{array}
\]
The left-hand side of each rule, \ie\ its pattern, expresses a
potentially infinite number of expressions on which the rule can be
applied; these expressions are obtained by instantiating the variables
$x,y,z$ with arbitrary Boolean expressions.  The application of the
rewrite rule is decided using a (matching) algorithm which only
depends on the pattern and on the term it is applied on, and not on
the context in which it is applied. Given a term, a rewrite rule can
be thus potentially applied on any of its sub-terms. The expression
$\ande(\ore(\atom{P},\atom{Q}),\ande(\ore(\atom{P},\atom{Q}),\atom{R}))$
where $\atom{P}$, $\atom{Q}$, $\atom{R}$ are atoms (constants) of the
language, is transformed either into
$\ore(\ande(\atom{P},\ande(\ore(\atom{P},\atom{Q}),\atom{R})),\ande(\atom{Q},\ande(\ore(\atom{P},\atom{Q}),\atom{R})))$,
if we applied the first rule on the whole term, or into
$\ande(\ore(\atom{P},\atom{Q}),\ore(\ande(\atom{P},\atom{R}),\ande(\atom{Q},\atom{R})))$
if we apply it on the $\ande$ sub-term. The latter term can be further
reduced in two different ways depending on whether we apply the first
or the second rule.
Such sets of rewrite rules, called \emph{term rewriting
  systems}~(TRS), are thus very convenient for describing
schematically the transformations one wants to operate.

In many situations, the application of a set of rewrite rules to a
term eventually leads to the same final result independently on the
way the rules are applied, and in such cases we say that the rewrite
rules are \emph{confluent} and \emph{terminating}. This is for example
the case for the above two rules.
When using rewriting as a programming or modeling paradigm, it is
nevertheless common to consider TRS that are non-confluent or
non-terminating. For example, if we want to compute not only DNF but
also conjunctive normal forms the following two rules could be added
to the previous TRS:
\[
\begin{array}{l}
  \ore(\ande(x, y), z) \ra \ande(\ore(x, z), \ore(y, z)) \\
  \ore(z, \ande(x, y)) \ra \ande(\ore(z, x), \ore(z, y))
\end{array}
\]
The resulting TRS is clearly non-terminating since we can have an
infinite number of rule applications:
$\ore(\ande(\atom{P},\atom{Q}),\atom{R})$ rewrites into
$\ande(\ore(\atom{P},\atom{R}),\ore(\atom{Q},\atom{R})),$ which in
turn rewrites into
$\ore(\ande(\atom{P},\ore(\atom{Q},\atom{R})),\ande(\atom{R},\ore(\atom{Q},\atom{R}))),$
and we can go on like this forever.
The problem clearly comes here from the interference between the two
sets of rules and there are several solutions to tackle this
problem~\cite{VISSER2005831}.

In a programming context, we can imagine that these rules are
separated in different modules; this solves the termination problem
but we could end up with a multitude of small modules if the problem
occurs for several sets of rules.
Another more general solution is
functionalization~\cite{VISSER2005831} which consists in introducing
some new \textit{function symbols} and use them to implicitly guide
the application of the rewrite rules modified accordingly. If we use
now the symbol $\atoms{}$ to represent atoms, we can compute the DNF
of a Boolean expression~$e$ by applying the following TRS on
$\dnf(e)$:
\[
\begin{array}{l}
  \dnf(\atomc{x}) \ra \atomc{x} \\
  \dnf(\ande(x, y)) \ra \dnfr(\ande(\dnf(x), \dnf(y))) \\
  \dnf(\ore(x, y)) \ra \ore(\dnf(x), \dnf(y)) \\
  \dnfr(\ande(\ore(x, y), z)) \ra \ore(\dnf(\ande(x, z)), \dnf(\ande(y, z))) \\
  \dnfr(\ande(z, \ore(x, y))) \ra \ore(\dnf(\ande(z, x)), \dnf(\ande(z, y))) \\
  \dnfr(\ande(\atomc{x}, \atomc{y})) \ra   \ande(\atomc{x}, \atomc{y})
\end{array}
\]
The rules for $\dnf$ simulate the innermost normalization strategy by
recursively traversing terms while the rules for $\dnfr$ are used to
encode the distribution.
However, for each new transformation,
like, \eg, computing the conjunctive normal forms,
new symbols should be introduced and new rules for these symbols
should be defined.  Moreover, if we extend our grammar of Boolean
expressions with new symbols, like for example negation $\note$ and implication
$\impe$, then new rewrite rules should be added for traversing the
corresponding terms or for defining the base cases.
Thus, this approach not only leads to a TRS with considerably more
rules than the original one, but they are also more difficult to
understand and to reuse.

The functionalization approach encodes an order on the application of rules
and gives an explicit specification on how to traverse symbols. 
\emph{Rewriting strategies} allow one to specify which rules should be
applied and where in the term without having to write explicitly in
the rewrite rules how to perform this control.  Rule-based languages
like {\elan}~\cite{ElanBKKMR-wrla98}, {\maude},
{\stratego}~\cite{Vis01.rta}, or {\tom} provide elementary rewriting
strategies and allow the programmer to define its own strategies
thanks to a specific strategy language. These user-defined strategies
are usually called \emph{programmable
  strategies}~\cite{VISSER2005831}. All these languages thus clearly
distinguish between what we want to transform (the data structures),
how to transform them (the rewrite rules), and how to control the
application of these transformations (the programmable strategies).

Rewriting strategies may or may not change the semantics of the TRS
they are controlling.  If the TRS is terminating and confluent,
imposing a reduction strategy would always lead to the same result,
but the number of alternative reductions to get to this result is
possibly smaller than for the uncontrolled TRS.  Imposing an order on
the (application of) the rules may reduce the number of alternative
reductions even further.
For example, if we extend the TRS computing the DNF with some simplification
rules for Boolean expressions, we could give a higher priority to the
application of some specific rules like, \eg $\ande(x,\false)\ra\false$,
to obtain shorter reductions.
When the underlying TRS is non-confluent or non-terminating, the
rewriting strategy could change its original semantics. The TRS
consisting of the first four rules in this section is non-terminating,
but when controlled by a strategy which applies in an innermost way
only the first two rules, the termination is retrieved.
Rewriting strategies can be also used to specify some default rules
which are applied only when the others are not applicable.  For
example, if we replace the last rule in the TRS defining $\dnf$ with
the more general rule $\dnfr(\ande(x,y))\ra\ande({x},{y})$ the TRS
would become non-confluent, but we can recover confluence by using a
strategy which prioritizes the application of the other, meaningful
rules, and gives the lowest priority to the last one.

Similarly to plain TRS ({\ie} TRS without strategy), it is interesting
to guarantee that a strategy-controlled TRS enjoys properties such as
confluence and termination. Confluence holds as long as the rewrite
rules are deterministic ({\ie} the corresponding matching algorithm
exhibits at most one solution) and all strategy operators are
deterministic (or a deterministic implementation is
provided). Confluence is clearly lost when considering
non-deterministic strategies. Termination is more delicate and the
normalization under some specific strategy is usually guaranteed by
imposing (sufficient) conditions on the rewrite rules on which the
strategy is applied.
Such conditions have been proposed for the
innermost~\cite{Arts2000133,GieslM02,ThiemannM08,GnaedigK09},
outermost~\cite{EndrullisH09,Thiemann09,GnaedigK09,RaffelsieperZ09},
context-sensitive~\cite{GieslM04,AlarconGL10,GnaedigK09}, or
lazy~\cite{GieslRSST11} reduction strategies. Termination under
programmable strategies has been studied for \elan~\cite{FissoreGK03}
and \stratego~\cite{KaiserL09,LammelTK13}. In~\cite{FissoreGK03}, the
authors give sufficient conditions for the termination of certain
programmable strategies relying only on the rewrite rules involved in
the strategy.  Since the proposed criterion is applicable only to a
specific subclass of programmable strategies and relies on a
coarse-grained description of strategies, the approach cannot be used
to prove termination for many terminating strategies.
In~\cite{KaiserL09,LammelTK13}, the termination of some traversal
strategies (such as top-down, bottom-up, innermost) is proven,
assuming the rewrite rules are measure decreasing for a notion of
measure that combines the depth and the number of occurrences of a
specific constructor in a term.

\medskip
\noindent \emph{Contributions.} 
In this paper we describe a general approach consisting in translating
programmable strategies into plain TRS.  The interest of this encoding
that we show sound and complete is twofold.
First, termination analysis
techniques~\cite{Arts2000133,Hirokawa2005172,Giesl06mechanizingand}
and corresponding tools like \aprove~\cite{terminationDP-aprove2011}
and \TTT~\cite{ttt-2009} that have been successfully used for checking
the termination of plain TRS can be used to verify termination in
presence of rewriting strategies.  Confluence can be analyzed in a
similar way.
Second, the translation can be seen as a generic strategy compiler and
thus can be used as a portable implementation of strategies which
could be easily integrated in any language accepting a term
representation for the  objects it manipulates and providing rewrite rules or
at least pattern matching primitives.

This translation was introduced in~\cite{CirsteaLM15}. We propose here
two additional translations which are intended to improve the
efficiency and expressiveness of the approach.  The number of rewrite
rules in the original encoding strongly depends on the considered
signature; we introduce a meta-level representation of terms allowing
to abstract over the signature leading thus to encodings whose number
of rules is independent of the signature.
We also consider many-sorted signatures and we show how the original
translation handling only mono-sorted signatures can be adapted to
generate well-sorted encodings executable in languages offering this
feature.
In summary, given a strategy $S$, we present an unsorted translation,
a many-sorted one, as well as one working at meta-level, and we show
that they produce faithful, {\ie} sound and complete, encodings of the
original strategy. 

The translations have been implemented in {\tom} and generate TRS which
could be fed into {\TTT}/{\aprove} for termination analysis or executed
efficiently by {\tom}.

The paper is organized as follows. The next section introduces the
notions of rewriting system and rewriting
strategy. Section~\ref{se:encodings} presents the translation of
rewriting strategies into rewriting systems (introduced
in~\cite{CirsteaLM15}), and in Sections~\ref{se:metaencodings}
and~\ref{se:typing}, we describe the translations for meta-level terms
and sorted terms, respectively.  In Section~\ref{se:implementation},
we give some implementation details and present experimental results.
In the appendix, we provide detailed proofs of the properties stated
in the paper and, in particular, of the simulation theorem
from~\cite{CirsteaLM15}.

\section{Strategic rewriting}
\label{se:semantics}
This section briefly recalls some basic notions of rewriting used in this paper;
see \cite{BaaderN98,Terese2002} for more details on first order terms and term
rewriting systems, and \cite{visser-icfp98,BallandMR-SPE2012} for details on rewriting
strategies and their implementation in rewrite based languages.

\subsection{Term rewriting systems}
A \emph{signature} $\Sigma$ consists of a finite set~$\FF$ of symbols
together with a function $\arty$ which associates to any symbol $f$
its \emph{arity}.  We write $\FF^n$ for the subset of symbols of arity
$n$, and $\FF^+$ for the symbols of arity $n>0$.  Symbols in $\FF^0$
are called \emph{constants}.
Given a countable set $\XX$ of \emph{variable} symbols, the set of
\emph{first-order terms} \emph{$\TFX$} is the smallest set
containing~$\XX$ and such that $f(t_1,\ldots,t_n)$ is in $\TFX$
whenever $f\in\FF^n$ and $t_i\in\TFX$ for $i\in [1,n]$.
We write $\var t$ for the set of variables occuring in $t\in\TFX$.  If
$\var t$ is empty, $t$ is called a \emph{ground} term; $\TF$ denotes
the set of all ground terms. A \emph{linear} term is a term where
every variable occurs at most once. A \emph{substitution} $\sigma$ is
a mapping from $\XX$ to $\TFX$ which is the identity except over a
finite set of variables (its \emph{domain}). A substitution extends as
expected to a mapping from $\TFX$ to $\TFX$.

A {\em position} of a term $t$ is a finite sequence of positive
integers describing the path from the root of $t$ to the root of the
sub-term at that position. We write $\varepsilon$ for the empty
sequence, which represents the root of $t$, $\PPos(t)$ for the set of
positions of $t$, and $\stt t \omega$ for the sub-term of $t$ at
position $\omega\in\PPos(t)$. Finally, $\rmp t \omega s$ is the term
$t$ with the sub-term at position $\omega$ replaced by~$s$.

A \emph{rewrite rule} (over $\Sigma$) is a pair
$(l,r)\in\TFX\times\TFX$ (also denoted $l \ra r$) such that $\var
r\subseteq \var l$ and a TRS is a set of rewrite rules $\RR$ inducing
a \emph{rewriting relation} over $\TF$, denoted by
$\longrightarrow_{\RR}$ and such that $t \longrightarrow_{\RR} t'$ iff
there exist $l\ra r\in\RR$, $\omega\in\PPos(t)$, and a substitution
$\sigma$ such that $\stt t \omega = \sigma(l)$ and $t' = \rmp t \omega
{\sigma(r)}$.
In this case, $\stt t \omega$ is a redex, $l$ matches $\stt t \omega$,
and $\sigma$ is the solution of the corresponding matching problem.
The reflexive and transitive closure of $\longrightarrow_{\RR}$ is
denoted by $\multieval{\RR}$.
In what follows, we generally use the notation $\trsred{{\RR}}{t}{t'}$
to denote $t \multieval{\RR} t'$.
A TRS $\RR$ is \emph{terminating} if there exists no infinite
rewriting sequence
$t_1\longrightarrow_{\RR}t_2\longrightarrow_{\RR}\ldots\longrightarrow_{\RR}t_n\longrightarrow_{\RR}\ldots$

\subsection{Rewriting strategies}
\label{sec:strategy}
Rewriting strategies allow one to specify how rules should be
applied. Classical strategies like innermost or outermost are
considered often when implementing or reasoning about strategy
controlled rewriting, but it could be useful to have more fine-grained
strategies to control the application of rewrite rules. For example,
we sometimes need to apply a set of rules only once, without retrying
to apply the rules on the obtained result as with an innermost
strategy. Consider a rule $\vare(x)\ra\fresh(\vare(x))$, whose purpose
is to indicate that every variable in the abstract syntax tree of a
program should be replaced by a corresponding fresh variable; to mark
all the concerned variables it is sufficient to apply the rule once on
the corresponding leaves of the tree, and it is needless to try the
rule on upper positions in the tree or on the newly generated fresh
variable.  We could then combine this strategy with another one in
charge of handling these fresh variables.

Taking the same terminology as the one proposed by {\elan} and
{\stratego}, a rewrite rule is considered to be an elementary
strategy, and a strategy is an expression built over a strategy
language. 

The strategy language we consider in this paper consists of the main
operators used in {\tom}, {\elan}, and {\stratego}. Let~$\Sigma$ be a
signature and $\XS$ be a set of strategy variables, ranged over by
$X$. In what follows, we use uppercase for strategy variables and
lowercase for term variables. We define the strategy language over
$\Sigma$ as follows:
$$
\begin{array}{rcl}
  S & ::=  &  \id \mid \fail \mid l \ra r \mid \seq{S}{S} \mid \choice{S}{S} \mid \one(S) \mid \all(S) \mid \mu X\dotsym{S} \mid X \\
\end{array}
$$
where $l \ra r$ is any rewrite rule over~$\Sigma$. 
We call the term to which the strategy is applied the subject.

The recursion operator $\mu X \dotsym S$ binds $X$ in $S$; a strategy
variable is said to be free if it is not bound. We write $\varS S$ for
the set of free variables of $S$. As usual, we work modulo
\emph{$\alpha$-conversion} and we adopt Barendregt's
\emph{``hygiene-convention''}, {\ie} free and bound variables have
different names.

Informally, the ${\id}$ strategy can be applied to any term without
changing it, and thus ${\id}$ always succeeds. Conversely, the
strategy ${\fail}$ always fails when applied to a term. As mentioned
above, a rewrite rule is an elementary strategy which is (successfully
or not) applied at the root position of its subject. By combining
these elementary strategies, more complex strategies can be built: we
can apply sequentially two strategies, make a choice between the
application of two strategies, apply a strategy to one or to all the
immediate sub-terms of the subject, and apply recursively a strategy.

The application of a strategy to a subject may diverge because
recursion, fail because of the strategy $\fail$ has been used or
because a rewrite rule cannot be applied, or return a (unique)
result. We use the symbol {\failres} to signal failure, and we let $u$
range over terms and \failres.
We implement recursion using a context $\Gamma$ which maps variables
to strategies; its syntax is defined as $ \Gamma ::= \emptyctx \mid
\Gamma ; X \colon S$. When writing a context
${X_1\colon{S_1};}\ldots;{X_n\colon{S_n}}$, we omit the empty context
$\emptyctx$, and we assume that the variables $(X_i)_{1 \leq i \leq
  n}$ are pairwise distinct. We write $\domS \Gamma$ for the domain of
$\Gamma$, defined as $\domS \emptyctx = \emptyset$, and $\domS{\Gamma;
  X \colon S} = \domS \Gamma \cup \{ X \}$.  We define the 
evaluation judgment $\step{\stratappctx{\Gamma}{S}{t}}{u}$ inductively
by the rules of Figure~\ref{fig:sem}; it means that the application of
the strategy $S$ to the subject $t$ produces the result $u$ under the
context~$\Gamma$ (the context may be omitted if empty). The semantics
being defined inductively and in a big-step style, there is no
derivation for $\stratappctx \Gamma S t$ if the application of $S$ to
$t$ diverges (\eg if $S = \mu X \dotsym X$). There is also no
derivation if a free variable $X$ of $S$ is applied to a term and $X
\notin \domS \Gamma$.

\begin{figure}[!th]
\vspace{0.2em}
\textsf{Elementary strategies}\\[-0.5em]
\newlength{\largeur}
\setlength{\largeur}{\textwidth}
\addtolength{\largeur}{0.6em}

\noindent\hspace{-0.3em}\hbox to \largeur{\hrulefill}
\begin{mathpar}
\inferrule{ }{\step{\stratappctx%
    {\Gamma}%
    {\id}{t}}{t}}~\mathbf{(id)}
\and
\inferrule{ }
{\step{\stratappctx%
    {\Gamma}%
    {\fail}{t}}{\failres}}~\mathbf{(fail)}
\and
\inferrule{\exists \sigma, \applysubs{\sigma}{\mathit{l}} = t}
{\step{\stratappctx%
    {\Gamma}%
    {\mathit{l} \ra \mathit{r}}%
    {t}%
    }%
        {\applysubs{\sigma}{\mathit{r}}}}~\mathbf{(r_1)}
\and
\inferrule{\nexists \sigma, \applysubs{\sigma}{\mathit{l}} ={t}}
{\step{\stratappctx%
    {\Gamma}
    {\mathit{l} \ra \mathit{r} }%
    {t}%
    }%
    {\failres}}~\mathbf{(r_2)}
\end{mathpar}

\noindent\hspace{-0.3em}\hbox to \largeur{\hrulefill}

\textsf{Control combinators}

\vspace{-0.5em}
\noindent\hspace{-0.3em}\hbox to \largeur{\hrulefill}

\vspace*{-1em}
\begin{mathpar}
\inferrule{\step{\stratappctx%
    {\Gamma}%
    {S_1}{t}}{t'}}
{\step{\stratappctx%
    {\Gamma}%
    {(\choice{S_1}{S_2})}{t}}{t'}}~\mathbf{(choice_1)}
\and
\hspace{-1.05em}\inferrule{\step{\stratappctx%
    {\Gamma}%
    {S_1}{t}}{\failres}
 \\ \step{\stratappctx%
   {\Gamma}%
   {S_2}{t}}{u}}
{\step{\stratappctx%
    {\Gamma}%
    {(\choice{S_1}{S_2})}{t}}{u}}~\mathbf{(choice_2)}
\and
\inferrule{\step{\stratappctx%
    {\Gamma}%
    {S_1}{t}}{t'}
 \\ \step{\stratappctx%
    {\Gamma}%
    {S_2}{t'}}{u}}
{\step{\stratappctx%
    {\Gamma}%
    {(\seq{S_1}{S_2})}{t}}{u}}~\mathbf{(seq_1)}
\and
\inferrule{\step{\stratappctx%
    {\Gamma}%
    {S_1}{t}}{\failres}}
{\step{\stratappctx%
    {\Gamma}%
{(\seq{S_1}{S_2})}{t}}{\failres}}~\mathbf{(seq_2)}
\and
\inferrule{\step{\stratappctx%
      {\Gamma; {X\colon S}}%
      {S}%
      {t}}{u}}
{\step{\stratappctx%
    {\Gamma}%
    {\mu X\dotsym S}%
    {t}}{u}}~\mathbf{(mu)}
\and
\inferrule{\step{\stratappctx%
      {\Gamma; {X\colon S}}%
      {S}
      {t}}{u}}
{\step{\stratappctx%
    {\Gamma; {X\colon S}}%
    {X}{t}}{u}}~\mathbf{(muvar)}
\end{mathpar}

\noindent\hspace{-0.3em}\hbox to \largeur{\hrulefill}

\textsf{Traversal combinators}

\vspace{-0.5em}
\noindent\hspace{-0.3em}\hbox to \largeur{\hrulefill}

\vspace*{-1em}
\begin{mathpar}
\inferrule{\exists i \in [1, n], 
       (\step{\stratappctx%
         {\Gamma}%
         {S}{t_i}}{t_i'}  \wedge
       \forall j \in [1, i-1], 
       \step{\stratappctx%
         {\Gamma}%
         {S}{t_j}}{\failres}) }
{\step{\stratappctx%
      {\Gamma}%
  {\one(S)}{f(t_1,\ldots,t_n)}}{f(t_1,\ldots,t_{i-1},t_i',t_{i+1},\ldots,t_n)}}~\mathbf{(one_1)}
\and
\inferrule{\forall i \in [1, n],   \step{\stratappctx%
    {\Gamma}%
    {S}{t_i}}{\failres}}
{\step{\stratappctx%
    {\Gamma}%
    {\one(S)}{f(t_1,\ldots,t_n)}}{\failres}}~\mathbf{(one_2)}
\and
\inferrule{\forall i \in {[}1{,}n{]}, 
    \step{\stratappctx%
    {\Gamma}%
    {S}{t_i}}{ t_i'}}
{\step{\stratappctx%
    {\Gamma}%
    {\all(S)}{f(t_1,\ldots,t_n)}}{f(t_1',\ldots,t_n')}}~\mathbf{(all_1)}
\and
\inferrule{\exists i \in {[}1{,}n{]},  \step{\stratappctx%
    {\Gamma}%
    {S}{t_i}}{\failres}}
{\step{\stratappctx%
    {\Gamma}%
    {\all(S)}{f(t_1,\ldots,t_n)}}{\failres}}~\mathbf{(all_2)}
\end{mathpar}
\caption{\label{fig:sem}Strategy semantics; $t$ and its indexed and
  primed versions denote 
terms (which cannot be {$\failres$}), whereas $u$ denotes
a result which
is either a well-formed term, or {$\failres$}.
}
\end{figure}

\noindent
We distinguish between three kinds of operators in the strategy language:
\begin{itemize}
\item \emph{elementary strategies} consisting of $\id$, $\fail$, and rewrite
  rules, which are the basic building blocks for strategies;
\item \emph{control combinators} consisting of choice
  $\choice{S_1}{S_2}$, sequence $\seq{S_1}{S_2}$, and recursion
  $\mu X\dotsym S$, that compose strategies but are still applied at
  the root position of
  the subject;
\item \emph{traversal combinators} $\one(S)$ and $\all(S)$ that modify the current application
  position.  
\end{itemize}

\noindent
We stress that a rewrite rule taken as an elementary strategy is not applied at
every possible position of a subject, but only at the root position, as we can see with
the following example.

\begin{exa}
\label{ex:rr-strategy}
Let $\Sigma$ be the signature corresponding to Peano natural numbers such that $\FF^0=\{\Zero\}$, $\FF^1=\{\Succ\}$,
and $\FF^2=\{\Plus\}$, and consider the rewrite rule $\Plus(\Zero,x)\ra x$. Then,
we have
\noindent$\step{\stratappctx{}{\Plus(\Zero,x)\ra x}{\Plus(\Zero,\Succ(\Zero))}}%
{\Succ(\Zero)}$, but 
\noindent$\step{\stratappctx{}{\Plus(\Zero,x)\ra x}{\Succ(\Plus(\Zero,\Zero))}}{\failres}$,
because the rule cannot be applied at the root position
 in the second case.
Similarly, only the redex at
the root position is reduced in 
\noindent$\step{\stratappctx{}{\Plus(\Zero,x)\ra x}{\Plus(\Zero,
    \Plus(\Zero, \Succ(\Zero)))}}{\Plus(\Zero, \Succ(\Zero))}$.
\end{exa}

Control combinators are also applied at the root position. The (left-)choice
$\choice{S_1}{S_2}$ tries to apply $S_1$ and considers $S_2$ only if $S_1$
fails. Using this operator, we can then define a strategy \combs{\Try(S)}
which tries to apply~\combs{S} and applies the identity if~\combs{S} fails:
\combs{\Try(S) = \choice{S}{\id}}. Given a set of rules $R_1,\ldots,R_n$, the
strategy $\choice{R_1}{(\choice{\cdots}{R_n})}$ can be used to express an order
on the rules, mirroring how pattern matching is done in most functional
programming languages. The sequential application $\seq{S_1}{S_2}$ succeeds if
$S_1$ succeeds on the subject and~$S_2$ succeeds on the subsequent term; it
fails if one of the two strategy applications fails. We consider these two
operators to be right-associative, so that $\choice{S_1}{(\choice{S_2}{S_3})}$
is written $\choice{S_1}{\choice{S_2}{S_3}}$ and $\seq{S_1}{(\seq{S_2}{S_3})}$
is written $\seq{S_1}{\seq{S_2}{S_3}}$.

\begin{exa}
\label{ex:control}
We can apply sequentially twice the rewrite rule $\Plus(\Zero,x)\ra x$
using the strategy $\seq{\Plus(\Zero,x)\ra x}{\Plus(\Zero,x)\ra x}$.
As we have seen in Example~\ref{ex:rr-strategy}, the rule $\Plus(\Zero,x)\ra x$
applied once to 
$\Plus(\Zero,\Plus(\Zero,\Succ(\Zero)))$ gives $\Plus(\Zero,
\Succ(\Zero))$, and since the rule can then be applied again at
the root position, we eventually obtain $\Succ(\Zero)$ as a final result:
  \begin{mathpar}
    \hspace{-0.5em}
    \inferrule*[right=$\mathbf{(seq_1)}$]
    {
      \inferrule*[right=$\mathbf{(r_1)}$]
      { }
      {\step{\stratappctx{}{\Plus(\Zero,x)\ra x}{\Plus(\Zero,\Plus(\Zero,\Succ(\Zero)))}}{\Plus(\Zero,\Succ(\Zero))}}
      \;\;
      \inferrule*[Right=$\mathbf{(r_1)}$]
      { }
      {\step{\stratappctx{}{\Plus(\Zero,x)\ra x}{\Plus(\Zero,\Succ(\Zero))}}{\Succ(\Zero)}}
    }
    {\step{\stratappctx{}{\seq{\Plus(\Zero,x)\ra x}{\Plus(\Zero,x)\ra x}}{\Plus(\Zero,\Plus(\Zero,\Succ(\Zero)))}}{\Succ(\Zero)}}
  \end{mathpar}

  \noindent
  However, the application of the same strategy on the term
  $\Plus(\Zero,\Succ(\Plus(\Zero,\Zero)))$ fails because the first
  application produces $\Succ(\Plus(\Zero,\Zero))$, on which we cannot
  apply the rule again:
  \begin{mathpar}
    \hspace{-0.5em}
    \inferrule*[right=$\mathbf{(seq_2)}$]
    {
      \inferrule*[right=$\mathbf{(r_1)}$]
      { }
      {\step{\stratappctx{}{\Plus(\Zero,x)\ra x}{\Plus(\Zero,\Succ(\Plus(\Zero,\Zero)))}}{\Succ(\Plus(\Zero,\Zero))}}
      \;\;
      \inferrule*[Right=$\mathbf{(r_2)}$]
      { }
      {\step{\stratappctx{}{\Plus(\Zero,x)\ra x}{\Succ(\Plus(\Zero,\Zero))}}{\failres}}
    }
    {\step{\stratappctx{}{\seq{\Plus(\Zero,x)\ra x}{\Plus(\Zero,x)\ra x}}{\Plus(\Zero,\Succ(\Plus(\Zero,\Zero)))}}{\failres}}
  \end{mathpar}

  \noindent
  We can avoid failure using the strategy $\Try(\seq{\Plus(\Zero,x)\ra
    x}{\Plus(\Zero,x)\ra x})$ which represents in fact the strategy
  $\choice{(\seq{\Plus(\Zero,x)\ra x}{\Plus(\Zero,x)\ra x})}{\id}$. If
  we denote $T$ the above derivation tree, we obtain the
  following evaluation when applying the $Try$ strategy:
  \begin{mathpar}
    \hspace{-0.5em}
    \inferrule*[right=$\mathbf{(choice_2)}$]
    {
      T
      \;\;\;\;\;\;\;\;\;\;\;\;
      \inferrule*[Right=$\mathbf{(id)}$]
      { }
      {\step{\stratappctx{}{\id}{\Plus(\Zero,\Succ(\Plus(\Zero,\Zero)))}}{\Plus(\Zero,\Succ(\Plus(\Zero,\Zero)))} }
    }
    {\step{\stratappctx{}{\Try(\seq{\Plus(\Zero,x)\ra x}{\Plus(\Zero,x)\ra x})}{\Plus(\Zero,\Succ(\Plus(\Zero,\Zero)))}}{\Plus(\Zero,\Succ(\Plus(\Zero,\Zero)))}}
  \end{mathpar}
  
  \noindent
  Note that the resulting term is the initial subject, not the
  intermediate term $\Succ(\Plus(\Zero,\Zero))$ responsible for the
  failure. To obtain this intermediate term, we should have applied
  $\seq{\Try(\Plus(\Zero,x)\ra x)}{\Try(\Plus(\Zero,x)\ra x)}$ to the
  subject.
\end{exa}

The application of a recursive strategy $\mu X\dotsym S$ to a subject
$t$ applies $S$ to $t$ in a context where $X$ is mapped to $S$ (rule
$\mathbf{mu}$). During this evaluation, the strategy variable $X$ may
be applied to a term $t'$, triggering the application of $S$ to $t'$ (rule
$\mathbf{muvar}$), allowing thus recursion. For the recursion to stop,
the evaluation of $S$ should not involve $X$ anymore at some point,
otherwise the application of $\mu X\dotsym S$ to $t$ diverges. For
example, the strategy $\mu X \dotsym X$ always diverges when applied
to any subject, and we cannot derive any evaluation judgment for it. A
more interesting example of recursive strategy is \combs{\Repeat(S) =
  \mu X\dotsym \Try(\seq{S}{X})} which applies a strategy~$S$ as much
as possible at the root position of the subject.  The last successful
result of the application of $S$ is the overall result of $\Repeat(S)$. Note
that applying a $Repeat$ strategy either diverges or 
evaluates to a term (and not to $\failres$).

\begin{exa}
  To see how recursion works, let us build the derivation tree for the
  application of the strategy $\Repeat(\Plus(\Zero,x)\ra x) = \mu
  X\dotsym\Try(\seq{\Plus(\Zero,x)\ra x}{X})$ to the term $t \eqdef
  \Plus(\Zero,\Plus(\Zero,\Succ(\Zero)))$. 
  To evaluate this application we necessarily start by applying the
  rule~$\mathbf{(mu)}$, which adds to the context the mapping of $X$
  to $\Try(\seq{\Plus(\Zero,x)\ra x}{X})$ and triggers the
  evaluation of this latter strategy. Thus, supposing the derivation
  tree can be eventually built then, it has the shape
  \begin{mathpar}
    \inferrule
    {\vdots}
    {\inferrule*[Right=$\mathbf{(mu)}$]
      {\step{\stratappctx{X:\tryShort}{\tryShort}{\Plus(\Zero,\Plus(\Zero,\Succ(\Zero)))}}
        {u} 
      } 
      {\step{\stratappctx{}{\mu X\dotsym\tryShort}{\Plus(\Zero,\Plus(\Zero,\Succ(\Zero)))}}
        {u}
      }
    }
  \end{mathpar}

  The subsequent evaluations in this derivation tree are performed
  \wrt\ this context $\ctxPlusZero=X:\Try(\seq{\Plus(\Zero,x)\ra x}{X})$.
  In order to apply the strategy
  $\Try(\seq{\Plus(\Zero,x)\ra x}{X})=\choice{\seq{\Plus(\Zero,x)\ra x}{X}}{\id}$ to
  $t$, we have to decide between the two $\mathbf{choice}$ rules; in both cases,
  we have to apply the strategy $\seq{\Plus(\Zero,x)\ra x}{X}$, and subsequently the
  rewrite rule $\Plus(\Zero,x)\ra x$, to~$t$.
  This rewrite rule succeeds when applied at the root position of $t$,
  and thus the derivation tree has the form:
  \begin{mathpar}
    \hspace{-0.5em}\inferrule*[right=$\mathbf{(seq_1)}$]
    {\inferrule*[right=$\mathbf{(r_1)}$]
      { }
      {\step{\stratappctx{\ctxPlusZero}{\Plus(\Zero,x)\ra x}{\subj}}{\Plus(\Zero,\Succ(\Zero))}}
       \\ \hspace{-1.5em}
      \inferrule*[Right=$\mathbf{(muvar)}$]
      {T_1}
      {\step{\stratappctx{\ctxPlusZero}{X}{\Plus(\Zero,\Succ(\Zero))}}{u}}
    }
    {\inferrule*[Right=$\mathbf{(choice_1)}$]
      {
        \step{\stratappctx{\ctxPlusZero}{\seq{\Plus(\Zero,x)\ra x}{X}}{\subj}}{u}
      }
      {\inferrule*[Right=$\mathbf{(mu)}$]
        {\step{\stratappctx{\ctxPlusZero}
            {\tryShort}{\subj}}{u}}
        {\step{\stratappctx{}{\mu
              X\dotsym\tryShort}{\subj}}{u}
        }
      }
    }
  \end{mathpar}
  with $T_1$ the derivation tree for the judgment
  $\step{\stratappctx{\ctxPlusZero} {\tryShort}{\Plus(\Zero,\Succ(\Zero))}}{u}$
  obtained by instantiating the variable $X$ with the corresponding
  strategy from the context. We can apply the same reasoning as above,
  and since the rewrite rule $\Plus(\Zero,x)\ra x$ also succeeds on
  the new term $\Plus(\Zero,\Succ(\Zero))$, the derivation tree
  $T_1$ has the following shape:
  \begin{mathpar}
    \inferrule*[right=$\mathbf{(seq_1)}$]
    {\inferrule*[right=$\mathbf{(r_1)}$]
      { }
      {\step{\stratappctx{\ctxPlusZero}{\Plus(\Zero,x)\ra
            x}{\Plus(\Zero,\Succ(\Zero))}}{\Succ(\Zero)}} 
      \\ 
      \inferrule*[Right=$\mathbf{(muvar)}$]
      {T_2}
      {\step{\stratappctx{\ctxPlusZero}{X}{\Succ(\Zero)}}{u}}
    }
    {\inferrule*[Right=$\mathbf{(choice_1)}$]
      {\step{\stratappctx{\ctxPlusZero}{\seq{\Plus(\Zero,x)\ra x}{X}}{\Plus(\Zero,\Succ(\Zero))}}{u}}
      {\step{\stratappctx{\ctxPlusZero}
          {\tryShort}{\Plus(\Zero,\Succ(\Zero))}}{u}}
    }
  \end{mathpar}
  with $T_2$ the derivation tree for the judgment
  $\step{\stratappctx{\ctxPlusZero} {\tryShort}{\Succ(\Zero)}}{u}$ obtained, as
  before, by instantiating accordingly the variable $X$. The rewrite
  rule $\Plus(\Zero,x)\ra x$ fails when applied to the term
  $\Succ(\Zero)$ and the derivation tree $T_2$ is therefore as
  follows:
  \begin{mathpar}
    \inferrule*[right=$\mathbf{(choice_2)}$]
    {
      \inferrule*[right=$\mathbf{(seq_2)}$]
      {
        \inferrule*[Right=$~\mathbf{(r_2)}$]
        { }
        {\step{\stratappctx{\ctxPlusZero}{\Plus(\Zero,x)\ra x}{\Succ(\Zero)}}{\failres}}
      }
      {\step{\stratappctx{\ctxPlusZero} {\seq{\Plus(\Zero,x)\ra x} X}{\Succ(\Zero)}}{\failres} }
      \\ 
      \hspace{-1em}\inferrule*[Right=$\mathbf{(id)}$]
      { }
      {\step{\stratappctx{\ctxPlusZero}{\id}{\Succ(\Zero)}}{\Succ(\Zero)} }
    }
    {\step{\stratappctx{\ctxPlusZero} {\tryShort}{\Succ(\Zero)}}{\Succ(\Zero)} }
  \end{mathpar}
 
 \noindent
  This completes the derivation tree, and we can thus conclude that
  the result of applying the strategy $\Repeat(\Plus(\Zero,x)\ra x)$
  to the term $\Plus(\Zero,\Plus(\Zero,\Succ(\Zero)))$ is
  $u=\Succ(\Zero)$.
  
\end{exa}

So far, the strategies are applied only at the root position; we use traversal
combinators to move to other positions. The strategy $\one(S)$ tries to apply
$S$ to a sub-term of the subject, starting from the leftmost one (rule
$\mathbf{one_1}$). If $S$ fails on all the sub-terms, then $\one(S)$ also fails
(rule $\mathbf{one_2}$). In contrast, $\all(S)$ applies $S$ to all the sub-terms
of the subject (rule~$\mathbf{all_1}$) and fails if $S$ fails on one of them
(rule~$\mathbf{all_2}$). Note that $\one(S)$ always fails when applied to a
constant while $\all(S)$ always succeeds in this case.

\begin{exa}
\label{ex:traversal}
In $\step{\stratappctx{}{\one(\Plus(\Zero,x)\ra
    x)}{\Plus(\Plus(\Zero, \Zero),\Plus(\Zero,\Succ(\Zero)))}}%
{\Plus(\Zero,\Plus(\Zero,\Succ(\Zero)))}$, only the first sub-term is reduced,
while in $\step{\stratappctx{}{\all(\Plus(\Zero,x)\ra
    x)}{\Plus(\Plus(\Zero, \Zero),\Plus(\Zero,\Succ(\Zero)))}}%
{\Plus(\Zero,\Succ(\Zero))}$, they are both reduced. These operators give access
only to the immediate sub-term: we have $\step{\stratappctx{}{\one(\Plus(\Zero,x)\ra
    x)}{\Succ(\Succ(\Plus(\Zero, \Zero)))}}{\failres}$, and similarly with $\all$.

\end{exa}

Combined with recursion, traversal combinators can be used to have access
to any position of a term. In fact, most of the classic reduction strategies can
be defined using recursion and traversal operators:
\[
\begin{array}{rcl}
  \OnceBottomUp(S) & = & \mu X\dotsym\choice{\one(X)}{S} \\
  \BottomUp(S) & = & \mu X\dotsym\seq{\all(X)}{S} \\
  \OnceTopDown(S) & = & \mu X\dotsym\choice{S}{\one(X)} \\
  \TopDown(S) & = & \mu X\dotsym\seq{S}{\all(X)} \\
\end{array}
\]
\combs{\OnceBottomUp} (denoted \combs{\OBU} in the
following) tries to apply a strategy~\combs{S} once, starting from
the leftmost-innermost leaves.  \combs{\BottomUp} behaves almost like
\combs{\OBU} except that~\combs{S} is applied to all nodes, starting
from the leaves. $\OnceTopDown$ and $\TopDown$ are similar, except they start
from the root position. 
The strategy which applies~\combs{S} as many times as possible,
starting from the leaves can be either defined naively as
\combs{\Repeat(\OnceBottomUp(S))} or using a more efficient
approach~\cite{visser-icfp98}:
\combs{\Innermost(S) = \mu X\dotsym \seq{\all(X)}{\Try(\seq{S}{X})}}.
\begin{exa}
\label{ex:strategies}
We use $\OnceTopDown$, $\OnceBottomUp$, and $\Innermost$, when applying
$S = \choice{\Plus(\Zero,x)\ra x}{\Plus(\Succ(x),y)\ra \Succ(\Plus(x,y))}$ to
$t = \Plus(\Plus(\Zero, \Zero), \Plus(\Succ(\Zero),\Zero))$. 
With $\OnceTopDown(S)$, we start at the root position of $t$ and since
there is no redex at this position, we recursively try to apply $S$ to
one of its children: we can apply it on the first child and obtain
$\Plus(\Zero,\Plus(\Succ(\Zero),\Zero))$.  If we apply
$\OnceTopDown(S)$ on this latter term we obtain
$\Plus(\Succ(\Zero),\Zero)$.
With $\OnceBottomUp(S)$, we start from the leaves, and we obtain the
same term $\Plus(\Zero,\Plus(\Succ(\Zero),\Zero))$ as before but if we
apply $\OnceBottomUp(S)$ once again on this term we obtain now
$\Plus(\Zero,\Succ(\Plus(\Zero,\Zero)))$.
$\Innermost(S)$ also starts from the leaves but
makes a recursive call after each application of $S$, so we obtain
$\step{\stratappctx{}{\Innermost(S)} t}{\Succ(\Zero)}$.

Both $\BottomUp(S)$ and  $\TopDown(S)$ fail on $t$ but we have
$\step{\stratappctx{}{\BottomUp(\Try(S))}t}{\Succ(\Plus(\Zero,\Zero))}$ 
and 
$\step{\stratappctx{}{\TopDown(\Try(S))}t}{\Plus(\Zero,\Succ(\Zero))}$.
\end{exa}

\begin{rem}
  \label{rem:not-det}
  The semantics of $\one$ is deterministic, as it looks for the leftmost
  sub-term that can be successfully transformed; a non-deterministic behavior
  can easily be obtained by removing the second condition in the premises of the
  inference rule $\mathbf{one_1}$. Similarly, we can adopt a non-deterministic
  semantics for $\choice{S_1}{S_2}$, where either $S_1$ or $S_2$ is tried first,
  by removing the first judgment in the premises of the inference rule
  $\mathbf{choice_2}$.
  We focus here on the encoding of the deterministic semantics but we
  also explain in the next section (Remark~\ref{rem:not-det-tr}) how
  our encoding can be adapted to handle the non-deterministic versions
  of $\one$ and choice.
\end{rem}

\noindent
Given this semantics, we can define a notion of termination for strategies.
\begin{defi}
  Let $S, \Gamma$ such that $\varS S \subseteq \domS \Gamma$.  The strategy $S$
  is terminating under $\Gamma$ if for all $t \in \TF$, there exists $u$ such
  that $\step{\stratappctx \Gamma S t} u$.
\end{defi}
\noindent 
The condition on the variables of $S$ and $\Gamma$ ensures that if 
there is no derivation for $\stratappctx \Gamma S t$, this is not because a
free variable of $S$ does not occur in the domain of $\Gamma$. If
$\varS S = \emptyset$, then $\Gamma$ 
does not play a role in the derivation, and we simply say
that $S$ is terminating as a shorthand for $S$ is terminating under any
context. Note that we allow terminating strategies to fail, as $u$ can be
$\failres$. A terminating strategy may contain non-terminating rewrite rules:
the strategy $\Zero \ra \Plus(\Zero, \Zero)$ is terminating because the rule is
applied once at the root position (and either succeeds or fails). However,
$\Innermost(\Zero \ra \Plus(\Zero, \Zero))$ is not terminating. A terminating
strategy may also contain a non-terminating strategy:
$\choice \id {\mu X \dotsym X}$ is terminating because the diverging part
$\mu X \dotsym X$ is never applied.

\section{Encoding rewriting strategies with rewrite rules}
\label{se:encodings}
We translate a strategy into a faithful TRS, in the sense that a strategy
applied to a term produces a result iff the TRS also rewrites the encoded term
into the encoded result.

\subsection{Strategy translation}
\label{se:strategyTranslation}
The evaluation of the application of a strategy $S$ on a term~$t$ consists in
selecting and applying the corresponding sub-strategies of $S$ to $t$ when the
head operator of $S$ is a control combinator ({\eg} selecting $S_1$ in
$\choice{S_1}{S_2}$ in the inference rule $\mathbf{choice_1}$), in applying $S$
on the corresponding sub-terms of $t$ when the head operator of~$S$ is a
traversal combinator ({\eg} applying $S$ on the sub-term $t_i$ of
$f(t_1,\ldots,t_n)$ in the inference rule $\mathbf{one_1}$), and eventually in
applying elementary strategies. 
The translation function presented in Figure~\ref{fig:encoding}
associates to each strategy a set of rewrite rules which encodes
exactly this behaviour and preserves the original evaluation:
$\trsred{\TRctx{\Gamma}{S}}{\vpS(t)}{t'}$ whenever
$\step{\stratappctx{}{S}{t}}{t'}$
(the exact relationship between a strategy and its encoding is
formally stated in Section~\ref{se:properties}). 

Since the application of the rewrite rules of the encoding cannot be
explicitly controlled, the first issue is to ensure that only the
rules of the appropriate {(sub-)strategy} can be applied to the
current subject, and that an encoded strategy is applied to the right
sub-term(s) in the case of traversal combinators.  The other
difficulties concern the representations of contexts~$\Gamma$, as well
as of (matching and strategy) failures, since these notions are not
explicit in plain term rewriting. To deal with the first problem, we
introduce some symbols $\varphi$ used to implicitly control the
application of the rewrite rules encoding a strategy.
The set $\TRctx{\Gamma}{S}$ thus contains rules whose left-hand
sides are headed by a symbol $\vpS$ and which encode the behaviour of the strategy
$S$ by using, in the right-hand sides, the symbols $\varphi$
corresponding to the sub-strategies of $S$ and potentially some
auxiliary symbols.
All these symbols are supposed to be freshly
generated and uniquely identified, {\ie} there will be only one $\varphi_{S}$
symbol for each encoded (sub-)strategy $S$ and each auxiliary $\varphi$ symbol
can be identified by the strategy it has been generated for. For example, in the
encoding 
$\TRctx{\Gamma}{{S_1}\seqsym{S_2}}$,
 the symbol~$\phiseq$ is just an
abbreviation for $\phiseq^{\seqsub{S_1}{S_2}}$, {\ie} the specific~$\phiseq$
used for the encoding of the strategy 
${S_1}\seqsym{S_2}$.

\begin{figure}[!tp]
\begin{mathpar}
\begin{array}{lr@{\hspace{5pt}}c@{\hspace{5pt}}l}
  \mathbf{(E1)}&
  \TRctx{\Gamma}{\id} & =  & \{  ~ \varphi_{\id}(\Xx \at \ap \bot(\any)  ) \ra \Xx, 
                                  \quad
                                  \varphi_{\id}(\bot(\Xx)) \ra \bot(\Xx)~\} \\[+4pt]
                                  \mathbf{(E2)}&
  \TRctx{\Gamma}{\fail} & = & \{ ~ \varphi_{\fail}(\Xx \at \ap \bot(\any)  ) \ra \bot(\Xx),
                                  \quad
                                  \varphi_{\fail}(\bot(\Xx)) \ra \bot(\Xx) ~\} \\ [+4pt]
                                  \mathbf{(E3)}&
  \TRctx{\Gamma}{l \ra r} & = & \{ ~   \varphi_{l\ra r}(l) \ra r, \\[+2pt]
                          &&&  ~~~ \varphi_{l\ra r}(\Xx \at \ap l) \ra  \bot(\Xx),
                              \quad 
                              \varphi_{l\ra r}(\bot(\Xx)) \ra \bot(\Xx)~~\} 
                              \\[+4pt]
\\[+1pt]
\mathbf{(E4)}&
  \TRctx{\Gamma}{\seq{S_1}{S_2}} & = &  \TRctx{\Gamma}{S_1} ~\cup~ \TRctx{\Gamma}{S_2}  \\[+2pt]
                          &\bigcup & \{ &  \varphi_{\seqsub{S_1}{S_2}}(\Xx \at \ap \bot(\any)  ) \ra \phiseq( \vpt(\vpo(\Xx)), \Xx) , 
                          \quad
                          \varphi_{\seqsub{S_1}{S_2}}(\bot(\Xx)) \ra \bot(\Xx) , \\[+2pt]
                          &&&  \phiseq(\Xx \at \ap \bot(\any), \any) \ra  \Xx, 
                          \quad
                          \phiseq(\bot(\any), \Xx) \ra  \bot(\Xx) ~~\} 
                          \\[+4pt]
                          \mathbf{(E5)}&
  \TRctx{\Gamma}{\choice{S_1}{S_2}} & =  & \TRctx{\Gamma}{S_1} ~\cup~ \TRctx{\Gamma}{S_2}  \\[+2pt]
                           &\bigcup & \{ &  \varphi_{\choicesub{S_1}{S_2}}(\Xx \at \ap \bot(\any)  ) \ra  \phich(\vpo(\Xx)) , 
                           \quad
                           \varphi_{\choicesub{S_1}{S_2}}(\bot(\Xx)) \ra \bot(\Xx) , \\[+2pt]
                           &&&  \phich(\bot(\Xx)) \ra  \vpt(\Xx) , 
                           \quad
                           \phich(\Xx \at \ap \bot(\any)) \ra  \Xx ~~\}
                           \\[+4pt]
                           \mathbf{(E6)}&
  \TRctx{\Gamma}{\mu X\dotsym S} & =  & \TRctx{\Gamma; {\Xx\colon S}}{S} \\[+2pt]
                 &\bigcup & \{ &  \varphi_{\mu X\dotsym S}(\Xx \at \ap \bot(\any)) \ra  \vpS(\Xx) ,
                 \quad
                 \varphi_{\mu X\dotsym S}(\bot(\Xx)) \ra \bot(\Xx),  \\[+2pt]
                 &&&  \varphi_{X}(\Xx \at \ap \bot(\any)) \ra  \vpS(\Xx) , 
                 \quad
                 \varphi_{X}(\bot(\Xx)) \ra \bot(\Xx)  ~~\} 
                 \\[+4pt]
                 \mathbf{(E7)}&
  \TRctx{\Gamma; {\Xx\colon S}}{X} & =  & \emptyset
  \\[+4pt]
\\[+1pt]
\mathbf{(E8)}&
  \TRctx{\Gamma}{\all(S)} & =  & \TRctx{\Gamma}{S}  \\[+2pt]
                 &\bigcup & \{ &  \varphi_{\all(S)}(\bot(\Xx)) \ra \bot(\Xx)  ~~\} \\[+2pt]
                 &\bigcup\limits_{c\in \FF^{0}} & \{ &  \varphi_{\all(S)}(c) \ra  c  ~~\}\\[+2pt]
                 &\bigcup\limits_{f\in \FF^{\supzero}} &\{ &  \varphi_{\all(S)}(f(\Xx_1,\ldots,\Xx_n)) \ra  \psiF(\vpS(\Xx_1),\ldots,\vpS(\Xx_n),f(\Xx_1,\ldots,\Xx_n)) , \\[+2pt]
                                                      &&&  \psiF(\Xx_1\at\ap\bot(\any),\ldots,\Xx_n\at\ap\bot(\any),\any) \ra  f(\Xx_1,\ldots,\Xx_n), \\[+2pt]
                                                      &&&  \psiF(\bot(\any),\any,\ldots,\any,x) \ra  \bot(x) , \\[+2pt]
                                                      &&&  \vdots  \\[+2pt]
                                                      &&&  \psiF(\any,\ldots,\any,\bot(\any),x) \ra  \bot(x)  ~~\}
                              \\[+4pt]
\mathbf{(E9)}&
  \TRctx{\Gamma}{\one(S)} & = & \TRctx{\Gamma}{S} \\[+2pt]
                &\bigcup & \{ &  \varphi_{\one(S)}(\bot(\Xx)) \ra \bot(\Xx)  ~~\} \\[+2pt]
                &\bigcup\limits_{c\in \FF^{0}} & \{ &  \varphi_{\one(S)}(c) \ra  \bot(c)  ~~\}\\[+2pt]
                &\bigcup\limits_{f\in \FF^{\supzero}} & \{ &  \varphi_{\one(S)}(f(\Xx_1,\ldots,\Xx_n)) \ra  \psif{1}(\vpS(\Xx_1),\Xx_2,\ldots,\Xx_n) ~~\} \\[+2pt]
                \multicolumn{2}{r}{\bigcup\limits_{f\in \FF^{\supzero}}~\bigcup\limits_{1\leq i\leq ar(f)}} & \{ &  \psif{i}(\bot(\Xx_1),\ldots,\bot(\Xx_{i-1}),\Xx_i\at\ap\bot(\any),\Xx_{i+1},\ldots,\Xx_n) \ra f(\Xx_1,\ldots,\Xx_n) \} \\[+2pt]
                \multicolumn{2}{r}{\bigcup\limits_{f\in \FF^{\supzero}}~\bigcup\limits_{1\leq i<ar(f)}} & \{ &  \psif{i}(\bot(\Xx_1),\ldots,\bot(\Xx_{i}),\Xx_{i+1},\ldots,\Xx_n) \ra\\[-7pt] 
                                                                                    &&&\multicolumn{1}{r}{\psif{i+1}(\bot(\Xx_1),\ldots,\bot(\Xx_{i}),\vpS(\Xx_{i+1}),\Xx_{i+2},\ldots,\Xx_n)  ~~\}} \\[+2pt]
                &\bigcup\limits_{f\in \FF^{\supzero}} & \{ &  \psif{n}(\bot(\Xx_1),\ldots,\bot(\Xx_n)) \ra  \bot(f(\Xx_1,\ldots,\Xx_n)) ~~\}
  \\[+4pt]
  \\[+4pt]
  \mathbf{(E10)}&
  \TRg{\Gamma'}{\Gamma; {X\colon S}} & =  &  \TRg{\Gamma'}{\Gamma} ~\cup~ \TRctx{\Gamma'}{S}             \\[+2pt]
                           &\bigcup &  \{  & \varphi_{X}(\Xx \at \ap \bot(\any)) \ra  \vpS(\Xx) , 
                           \quad
                           \varphi_{X}(\bot(\Xx)) \ra \bot(\Xx)   ~~\}
  \\[+4pt]
  \mathbf{(E11)}&
  \TRg{\Gamma'}{\emptyctx} & =  &  \emptyset
\end{array}
\end{mathpar}
\caption{\label{fig:encoding}Strategy translation.
}
\end{figure}

In addition to these $\varphi$ symbols, we also use a
particular symbol $\bot$ of arity $1$ which encodes failure and
whose argument is used, as explained later on in this section, to keep
track of the term on which the strategy failed.
For example, the fact that the strategy $\fail$ applied to $t$ results
in $\failres$ corresponds in our encoding to the reduction of
$\varphi_{\fail}(t)$ \wrt\ $\TRctx{}{\fail}$ into $\bot(t)$.  This
additional information attached to a failure is particularly helpful
when encoding strategies where a sub-strategy is tried successively on
different terms, namely for the choice and for $\one$.

Apart from the application of $\fail$, strategy failures only come
from matching failures. We have to be careful in encoding these
failures since we should conclude that a matching failure occurred
only if the terms involved are built over the initial signature and
contain no generated symbols.  For this we use the so-called
\emph{anti-terms}\footnote{We restrict here to a limited form of
  anti-terms; we refer to~\cite{jsc2010} for the complete semantics of
  anti-terms.}  of the form $\ap t$ with $t\in \TFX$ or $\ap \bot(x)$
with $x\in\XX$.  An anti-term $\ap t$ represents all the ground terms
in $\TF$ which do not match $t$ and the anti-pattern $\ap \bot(x)$
denotes all the ground terms $t\in\TF$ of the original signature; the
way the finite representation of these terms is generated is explained
in Section~\ref{se:implementation}. For example, if we consider the
signature from Example~\ref{ex:rr-strategy}, $\ap \Plus(\Zero,x)$
denotes exactly the terms matched by $\Zero$, $\Succ(x_1)$,
$\Plus(\Succ(x_1),x_2)$, or $\Plus(\Plus(x_1,x_2),x_3)$.
We should emphasize that  in this encoding, the semantics of $\ap t$ 
is considered {\wrt} the terms in $\TF$. For example, $\ap c$ for
some constant $c$ does not include $\bot(c)$ or terms of the form
$\varphi(t_1,\ldots,t_n)$, because $\bot$ and $\varphi$ symbols do not
belong to the original signature.

To keep the presentation of the translation compact and intuitive, we express it
using rule schemas involving anti-terms and a special aliasing symbol~``$\at$'':
each rule in Figure~\ref{fig:encoding} involving these notations represent in
fact a set of rewrite rules.
Anti-terms in the left-hand sides of rules can be aliased with $\at$ by
variables which can be then conveniently used in the right-hand sides of the
corresponding rules.
The rewrite rule $l\ra r$ with~$l$,~$r$ such that $\stt l \omega=x\at\ap u$ and
$\stt r {\omega'}=x$, for some $u\in\TFX\cup\{\bot(x)\mid x\in \XX\}$ and
positions~$\omega$,~$\omega'$
is thus just an alias for the set of rewrite rules
$\rmp{l}{\omega}{u_1} \ra
\rmp{r}{\omega'}{u_1},\ldots,\rmp{l}{\omega}{u_n} \ra
\rmp{r}{\omega'}{u_n}$, where $u_1,\ldots,u_n$ are the terms
represented by $\ap u$.
Moreover, the variable symbol ``$\any$'' can be used in the left-hand
side of a rule to indicate a variable that does not appear in the
right-hand side.

For example, the rule schema
$\varphi(\Yy\at\ap\Plus(\Zero,\any))\ra\bot(\Yy)$ denotes the set of
rewrite rules consisting of $\varphi(\Zero)\ra\bot(\Zero)$,
$\varphi(\Succ(\Yy_1))\ra\bot(\Succ(\Yy_1))$,
$\varphi(\Plus(\Succ(\Yy_1),\Yy_2))\ra\bot(\Plus(\Succ(\Yy_1),\Yy_2))$,
and
$\varphi(\Plus(\Plus(\Yy_1,\Yy_2),\Yy_3))\ra\bot(\Plus(\Plus(\Yy_1,\Yy_2),\Yy_3))$.

The translation of the $\id$ strategy ($\mathbf{E1}$) consists of a
rule whose left-hand side matches any term in the
signature\footnote{The rule is in fact expanded into $n$ rewrite rules
  with $n$ the number of symbols in $\FF$.}  (contextualized by the
corresponding $\varphi$ symbol) and whose right-hand side is the
initial term, and of a rule encoding strict propagation of failure.
This latter rule guarantees a faithful encoding of the strategy guided
evaluation and is in fact present, in different forms, in the
translations of all the strategy operators.

\begin{exa}
\label{ex:encodeId}
If we consider the signature in Example~\ref{ex:rr-strategy},
the following encoding is obtained for the $\id$ strategy:

$
\begin{array}{lll}
\TRctx{\Gamma}{\id} = \{ 
  & \varphi_{\id}(\Zero) \ra \Zero,
 \\
 &  \varphi_{\id}(\Succ(\Yy_1)) \ra  \Succ(\Yy_1), 
 \\
 &  \varphi_{\id}(\Plus(\Yy_1,\Yy_2)) \ra  \Plus(\Yy_1,\Yy_2), 
 \\
 &  \varphi_{\id}(\bot(\Xx)) \ra \bot(\Xx)
 ~~\} 
\end{array}
$
\end{exa}
\noindent
Similarly, the translation of the $\fail$ strategy ($\mathbf{E2}$)
contains a failure propagation rule, and a rule whose left-hand side
matches any term and whose right-hand side is a failure keeping track
of this term.
A rewrite rule (which is an elementary strategy applicable at the root
of the subject) is translated ($\mathbf{E3}$) by two rules encoding
the behaviour in case of respectively a matching success or a failure,
together with a rule for failure propagation.

\begin{exa}
\label{ex:encodeElem}
The strategy $S_{pz}\eqdef\Plus(\Zero,x)\ra x$ is encoded by the following
rules:

$
\begin{array}{ll}
\TRctx{\Gamma}{S_{pz}} =  \{ &   \varphi_{pz}(\Plus(\Zero,x)) \ra x, \\
                             & \varphi_{pz}(\Yy \at \ap \Plus(\Zero,x)) \ra  \bot(\Yy), \\
                             & \varphi_{pz}(\bot(\Xx)) \ra \bot(\Xx)~~\} 
\end{array}
$

\noindent
which lead,  when the anti-terms are expanded {\wrt} to the signature,
to the TRS:

$
\begin{array}{ll}
\TRctx{\Gamma}{S_{pz}} =  \{ &   \varphi_{pz}(\Plus(\Zero,x)) \ra x, \\
& \varphi_{pz}(\Zero) \ra  \bot(\Zero), \\
& \varphi_{pz}(\Succ(\Yy_1)) \ra  \bot(\Succ(\Yy_1)), \\
& \varphi_{pz}(\Plus(\Succ(\Yy_1),\Yy_2)) \ra  \bot(\Plus(\Succ(\Yy_1),\Yy_2)), \\
& \varphi_{pz}(\Plus(\Plus(\Yy_1,\Yy_2),\Yy_3)) \ra  \bot(\Plus(\Plus(\Yy_1,\Yy_2),\Yy_3)), \\
& \varphi_{pz}(\bot(x)) \ra \bot(x)~~\} 
\end{array}
$

\noindent
The term $\varphi_{pz}(\Plus(\Zero,\Succ(\Zero)))$ reduces {\wrt} this
latter TRS to $\Succ(\Zero)$ and $\varphi_{pz}(\Zero)$ reduces to
$\bot(\Zero)$.
\end{exa}
%
The translation of the sequential application of two strategies ($\mathbf{E4}$)
includes the translation of the respective strategies and some
specific rules.
A term $\varphi_{\seqsub{S_1}{S_2}}(t)$ is reduced by the first rule
into a term $\phiseq( \vpt(\vpo(t)), t)$, which guarantees that the
rules of the encoding of~$S_1$ are applied before the ones of
$S_2$. Indeed, a term of the form $\varphi(t)$ can be reduced only if
$t \in \TF$ or $t = \bot(\any)$ and thus, the rules for $\vpt$ can be
applied to a term $\vpt(\vpo(\any))$ only after $\vpo(\any)$ is
reduced to a term in $\TF$ (or failure).  The original subject $t$ is
kept during the evaluation (of $\phiseq$), so that $\bot(t)$ can be
returned if the evaluation of $S_1$ or $S_2$ fails ({\ie} produces a
$\bot$) at some point. If $\vpt(\vpo(t))$ evaluates to a term $t' \in
\TF$, then the evaluation of $\varphi_{\seqsub{S_1}{S_2}}(t)$
succeeds, and $t'$ is the final result.

In a similar manner, the translation for the choice operator
($\mathbf{E5}$) uses a first rule which triggers the application of
the rules for $S_1$. If we start with a term
$\varphi_{\choicesub{S_1}{S_2}}(t)$ we first obtain $\phich(\vpo(t))$
and this latter term can be further reduced only after the reduction
of $\vpo(t)$ to a term in the signature or to a failure.
Recording the origin of the failure and propagating it together with
its origin is crucial here. 
The eventual failure for the reduction of $\vpo(t)$ can be thus
detected and the original subject, \ie\ the term $t$, can be retrieved
and used to trigger the rules for $S_2$ on $\vpt(t)$.

The translation for a strategy $\mu X\dotsym S$ ($\mathbf{E6}$) triggers the
application of the rules for $S$ at first, and then each time the symbol
$\varphi_{X}$  is encountered. As in all
the other cases, failure is strictly propagated.
There is no rewrite rule for the translation of a strategy variable
($\mathbf{E7}$) but we should note that the corresponding $\varphi_{X}$ symbol
is used 
when the sub-strategy $X$ is reached during the
translation of the strategy $S$ (in $\mu X\dotsym S$).

\begin{exa}
\label{ex:encodeControl}

The strategy $S_{rpz}=\mu \varX\dotsym\choice{(\seq{\Plus(\Zero,x)\ra
    x~}{~\varX})}{\id}$ which applies repeatedly (as long as  possible) the
rewrite rule from Example~\ref{ex:encodeElem} is encoded by:
  \[
\begin{array}{rll}
\TRctx{\Gamma}{S_{rpz}} =  & \{~  \varphi_{rpz}(\Xx \at \ap \bot(\any)) \ra  \varphi_{tpz}(\Xx) ,
                           & \varphi_{rpz}(\bot(\Xx)) \ra \bot(\Xx)  ~~\} 
                           \\
                      \cup & \{~    \varphi_{\varX}(\Xx \at \ap \bot(\any)) \ra  \varphi_{tpz}(\Xx) , 
                           & \varphi_{\varX}(\bot(\Xx)) \ra \bot(\Xx)  ~~\} \\
                      \cup & \{~  \varphi_{tpz}(\Xx \at \ap \bot(\any)  ) \ra  \phich(\varphi_{pz\varX}(\Xx)) , 
                           & \varphi_{tpz}(\bot(\Xx)) \ra \bot(\Xx) , \\
                           &  ~~~\phich(\Xx \at \ap \bot(\any)) \ra  \Xx , 
                           &  \phich(\bot(\Xx)) \ra  \varphi_{\id}(\Xx) ~~\} \\
                      \cup & \{~ \varphi_{pz\varX}(\Xx \at \ap \bot(\any)  ) \ra \phiseq( \varphi_{\varX}(\varphi_{pz}(\Xx)), \Xx) , 
                           & \varphi_{pz\varX}(\bot(\Xx)) \ra \bot(\Xx) , \\
                           &  ~~~ \phiseq(\Xx \at \ap \bot(\any), \any) \ra  \Xx , 
                           & \phiseq(\bot(\any), \Xx) \ra  \bot(\Xx) ~~\} \\
                              \cup & \TRctx{\Gamma}{S_{pz}} ~\cup~ \TRctx{\Gamma}{\id}
\end{array}
  \]
For presentation purposes, we separated the TRS in subsets of rules
corresponding to the translation of each operator occurring in the
initial strategy. Note that the symbol $\varphi_{\varX}$ used in the
rules for the inner sequence can be reduced with the rules generated
to handle the recursion operator.
The term $\varphi_{rpz}(\Plus(\Zero,\Plus(\Zero,\Succ(\Zero))))$
reduces {\wrt} the TRS to $\Succ(\Zero)$.
\end{exa}

The rules encoding the traversal operators follow the same principle
-- the rules corresponding to the translation of the argument strategy
$S$ are applied, depending on the traversal operator, to one or all
the sub-terms of the subject.
For the $\all$ operator ($\mathbf{E8}$), if the application of $S$ to
all the sub-terms succeeds (produces terms in $\TF$), then the final
result is built using the results of each evaluation. If the
evaluation of one of the sub-terms produces a $\bot$, a failure with
the original subject as origin is returned as a result.  Special rules
encode the fact that $\all$ applied to a constant always succeeds; the
same behaviour could have been obtained by instantiating the rules for
non-constants with $n=0$, but we preferred an explicit approach for
uniformity and efficiency reasons.
The strict propagation of failure together with its origin is
essential in the case of the $\one$ operator ($\mathbf{E9}$) since it
allows for a more economic encoding approach than the duplication of the
original subject used for the $\all$ operator.  If the evaluation for
one sub-term results in a failure, then the evaluation of the strategy
$S$ is triggered on the next one and if $S$ fails on all sub-terms,
the final result can be directly built using the origin of these
failures which represent nothing else than the original arguments.
Note that the failure in case of constants is necessarily encoded by
specific rules.

Finally, each binding $X \colon S$ of a context ($\mathbf{E10}$) is
translated by two rules, including the one that propagates
failure. The other rule operates as in the recursive case
($\mathbf{E6}$): applying the strategy variable $X$ to a subject
$t$ leads to the application of the rules encoding~$S$ to $t$.

\begin{rem}
  \label{rem:not-det-tr}
  It is possible to adapt our encoding to a non-deterministic
  version of $\choice{S_1}{S_2}$ and $\one(S)$ (cf
  Remark~\ref{rem:not-det}). For the choice combinator, the encoding becomes
  \[
  \begin{array}{rcl}
    \TRctxND{\Gamma}{\choice{S_1}{S_2}} & = &  \TRctxND{\Gamma}{S_1} ~\cup~ \TRctxND{\Gamma}{S_2}  \\[+2pt]
                                      \bigcup & \{ &
                                                  \varphi_{\choicesub{S_1}{S_2}}(\Xx \at
                                                  \ap \bot(\any)  ) \ra
                                                  \phich(\varphi_{S_1}(\Xx), \varphi_{S_2}(\Xx)), \quad
                                                  \varphi_{\choicesub{S_1}{S_2}}(\bot(\Xx))
                                                  \ra \bot(\Xx) , \\[+2pt] 
                                      & &\phich(\Xx \at \ap \bot(\any), \any) \ra
                                          \Xx, \quad \phich(\any, \Xx \at \ap
                                          \bot(\any)) \ra \Xx, \\[+2pt] 
                                      & &   \phich(\bot(\Xx),\bot(\Xx)) \ra  \bot(\Xx)
                                          \quad
                                           ~~\}
  \end{array}
  \]
  In this modified encoding, the subject $t$ is 
  duplicated so that the rules of
  $\TRrND{S_1}$ can be applied to $\varphi_{S_1}(t)$ and those of $\TRrND{S_2}$ to
  $\varphi_{S_2}(t)$. If either copy rewrites into a term $t'$ in $\TF$, then it
  can be selected as the final result by the rule
  $\phich(\Xx \at \ap \bot(\any), \any) \ra \Xx$ or its symmetric. 
  If both of them could be eventually rewritten into the terms 
  $t'_1,t'_2$ in $\TF$, the final result is either~$t'_1$ or $t'_2$ depending
  on whether $\phich(\Xx \at \ap \bot(\any), \any) \ra \Xx$ or
  $\phich(\any, \Xx \at \ap \bot(\any)) \ra \Xx$ is applied first in
  the corresponding reduction.This non-deterministic choice reflects
  the non-deterministic nature of the encoded strategy.

  \noindent
  The encoding for the non-deterministic $\one$ follows the same idea.
  \[
  \begin{array}{rcl}
    \TRctxND{\Gamma}{\one(S)} & =  & \TRctxND{\Gamma}{S}  \mathrel{\bigcup} \{~
                                   \varphi_{\one(S)}(\bot(\Xx)) \ra \bot(\Xx)
                                   ~\}  \bigcup\limits_{c\in \FF^{0}} \{~
                                   \varphi_{\one(S)}(c) \ra  \bot(c)
                                   ~\}\\[+2pt]
    \bigcup\limits_{f\in \FF^{\supzero}} &\{ &  \varphi_{\one(S)}(f(\Xx_1,\ldots,\Xx_n)) \ra  \psiF(\vpS(\Xx_1),\ldots,\vpS(\Xx_n),f(\Xx_1,\ldots,\Xx_n)) , \\
                            &&  \psiF(\bot(\any),\ldots,\bot(\any), x) \ra  \bot(x) , \\[+2pt]
                            &&  \psiF(y \at
                               \ap\bot(\any),\any,\ldots,\any,f(x_1, \ldots, x_n))
                               \ra f(y, x_2, \ldots, x_n) , \\[+2pt]
                            &&  \vdots  \\[+2pt]
                            &&  \psiF(\any,\ldots,\any,y\at\ap\bot(\any),f(x_1, \ldots, x_n))
                               \ra f(x_1, \ldots, x_{n-1}, y)  ~~\}
  \end{array}
  \]
  The obtained encoding is similar to the one for $\all$: the rules
  corresponding to the translation $S$ are triggered on all the sub-terms of the
  subject. We remember the original subject so it can be returned as origin of
  an eventual failure, or to construct the final result with one of the last $n$
  rules.  Again, if the encoding of $S$ succeeds on several sub-terms, different
  results can be obtained depending on which of the last $n$ rules is being
  applied.  

\end{rem}

\subsection{Properties of the translation}
\label{se:properties}

The goal of the translation is twofold: use well-established methods
and tools for plain TRS in order to prove properties of strategy
controlled rewrite rules, and offer a generic compiler for user
defined strategies.  For both items, it is crucial to have a sound and
complete translation, and this turns out to be true in our case.

\begin{restatable}[Simulation]{thmm}{simulation}
\label{th:simulation}
Given a term $t\in \TF$, a strategy $S$  and a context $\Gamma$ such that $\varS
S \subseteq \domS \Gamma$,
\begin{enumerate}
\item 
$\stratappctx{\Gamma}{S}{t} \sred t'$
~~iff~~
$\trsred{\TRctx{\Gamma}{S} \cup \TRg{\Gamma}{\Gamma}}{\vpS(t)}{t'}$, $t'\in \TF$
\item 
\textnormal{$\stratappctx{\Gamma}{S}{t} \sred \failres$}
~~iff~~
$\trsred{\TRctx{\Gamma}{S} \cup \TRg{\Gamma}{\Gamma}}{\vpS(t)}{\bot(t)}$  
\end{enumerate}
\end{restatable}

\begin{proof}[Sketch]
  The completeness is shown by induction on the height of the
  derivation tree and the soundness by induction on the length of the
  reduction.  The base cases consisting of the strategies with a
  constant length reduction -- $\id$, $\fail$, and the rewrite rule --
  are straightforward to prove since, in particular, the translation
  of a rule explicitly encodes matching success and failure.
  Induction is applied for all the other cases and the corresponding
  proofs rely on some auxiliary properties.

  First, the failure is strictly propagated: if $\trsred{\TRgGG \cup
    \TRctx{{\varphi}}{S}}{\vpS(\bot(t))}{u}$, then $u = \bot(t)$. This
  is essential, in particular, for the sequence case where a failure
  of the first strategy should be strictly propagated as the final
  result of the overall sequential strategy.
   
  Second, terms in $\TF$ are in normal form {\wrt} the translation of any
  strategy and terms of the form $\vpS(t)$ are head-rigid {\wrt} strategies
  other than $S$, \ie, they can be reduced at the head position only by the
  rules in $\TRctx{}{S}$ and only if $t$ is a term in the signature. More
  precisely, if for a strategy~$S'$ and a context~$\Gamma$,
  $\trsred{\TRgGG \cup \TRctx{{\varphi}}{S'}}{\vpS(t)}{u}$ then $t \in \TF$ and
  $\TRctx{}{S} \subseteq \TRgGG \cup \TRctx{{\varphi}}{S'}$ (or $S=X$
  and~$X \in \domS \Gamma$).
  This guarantees that the steps in the strategy derivation are
  encoded accurately by the evaluations {\wrt} the rules in the
  translation.

  Finally, the origin of the failure is preserved in the sense that if
  for a $t \in \TF$, $\vpS(t)$ reduces to a failure, then the reduct
  is necessarily $\bot(t)$. This is crucial in particular for the
  choice strategy: if the (translation of the) first strategy fails,
  then the (translation of the) second one should be applied on the
  initial subject.
\end{proof}

The main goal is to prove the termination of some strategy guided
system by proving the property for the plain TRS obtained by the
translation.
As a direct consequence of Theorem~\ref{th:simulation}, we obtain that
the termination of the TRS encoding a strategy implies the termination
of the strategy.

\begin{cor}[Termination]
  Given a strategy $S$ and a context $\Gamma$ such that
  $\varS S \subseteq \domS \Gamma$, $S$ terminates under $\Gamma$ if
  $\TRctx{\Gamma}{S}\cup\TRg{\Gamma}{\Gamma}$ is terminating.
\end{cor}

Because of the modular encoding of strategies, the non-termination of
the TRS does not necessarily imply the non-termination of the original
strategy, and could instead indicate just the non-termination of one
of its sub-strategies.
Given the shape of the rewrite rules in the encoding, an eventual
non-terminating reduction could be always reduced to a non-terminating
reduction of a term of the form $\vpS(t)$ with $\vpS$ encoding the
behaviour of a sub-strategy $S$ and $t\in\TF$.
For example, the terminating strategies 
$\seq{\fail}{\mu X\dotsym(\seq{\id}{X})}$,
$\choice{\id}{\mu X\dotsym(\seq{\id}{X})}$, and
$\seq{a\ra a}{\mu X\dotsym(\seq{b\ra b}{X})}$
use some non-terminating sub-strategies whose encodings are also
non-terminating. 

Note that even if the non-termination of the encoding reflects the
non-termination of the strategy, the counterexample exhibited (by an automatic
tool) could concern only a sub-strategy. For example, the strategy $\all(\mu
X\dotsym(\seq{\id}{X}))$ is non-terminating because of its sub-strategy ${\mu
  X\dotsym(\seq{\id}{X})}$ and an automatic tool would usually provide a
counterexample for (the encoding of) this latter strategy.

Another direct consequence of Theorem~\ref{th:simulation} is that the
confluence of the strategy is implied by the confluence of the
corresponding encoding TRS.

\begin{cor}[Confluence]
  Given a strategy $S$ and a context $\Gamma$ such that $\varS S
  \subseteq \domS \Gamma$, if
  $\TRctx{\Gamma}{S}\cup\TRg{\Gamma}{\Gamma}$ is confluent then, for
  any evaluation judgments $\stratappctx{\Gamma}{S}{t} \sred u$ and
  $\stratappctx{\Gamma}{S}{t} \sred u'$ we have $u=u'$.
\end{cor}

Moreover, the encoding of a deterministic strategy always leads to
a confluent TRS: 
\begin{restatable}{lemm}{confluence}
\label{lemma:confluentEncoding}
Given a strategy $S$ and a context $\Gamma$ such that $\varS S
\subseteq \domS \Gamma$, the TRS
$\TRctx{\Gamma}{S}\cup\TRg{\Gamma}{\Gamma}$ is confluent.
\end{restatable}
\begin{proof}[Sketch]
  We split $\TRctx{\Gamma}{S}\cup\TRg{\Gamma}{\Gamma}$ into two TRS and show that they are
  both linear, confluent and orthogonal to each other, {\ie} there is
  no overlap between a rule from one and a rule from the other. As a
  consequence, we have the confluence of
  $\TRctx{\Gamma}{S}\cup\TRg{\Gamma}{\Gamma}$~\cite{Ohlebusch2002}.
\end{proof}

\section{Meta-encoding rewriting strategies with rewrite rules}
\label{se:metaencodings}
The strategy translation of Section~\ref{se:encodings} generates a
potentially large number of rewrite rules essentially because of the
way $\all$ and $\one$ are encoded, and this may impact the efficiency
of the corresponding implementation or prevent an automatic tool from
deciding termination. We propose in this section a meta-encoding of
terms which allows for a more economic encoding of strategies into
rewrite rules.

\subsection{Meta-level representation of terms}
The non-failure $\ap\bot(\any)$ has been encoded so far by the set of
terms $\bigcup_{f\in\FF}\{f(x_1,\ldots,x_n)\}$, with $\arty(f)=n$, and
thus all rule schemas in Figure~\ref{fig:encoding} using this
construction are expanded into a number of rules proportional to the
number of symbols in the signature.  The size of the encodings for
$\all$ and $\one$ in the translation depends also on the signature
$\FF$, as the number of rules generated grows with the number of
symbols and with their arities.
To avoid this potential explosion, we represent terms in $\TF$
using a meta-level application and lists: roughly, we write $f(t_1, \ldots,
t_n)$ as $\appl {f}{\args}$, where $\args$ is the list of the 
meta-encoded arguments, and a constant $c$ is represented as $\appl
{c} \nil$.

Formally, given a signature $\FF$, we define a meta-signature $\FFmeta \eqdef \{
  \rawappl, \conss, \nil\} \cup \{\trmeta f \mid f \in \FF \}$ so that $\rawappl$ and
$\conss$ are of arity 2, and $\nil$ and $\trmeta f$ are of arity 0 for all $f
\in \FF$. The symbols $\conss$ and $\nil$ are the usual building blocks for
lists, and we use $\conss$ as a right-associative infix operator, so that $x_1
\conss (x_2 \conss \nil)$ is written $x_1\conss x_2 \conss \nil $.  Given a term
$t \in \TFX$, we define its meta-level encoding $\trmeta t \in \TFXmeta$ as
$\trmeta{f(t_1, \ldots, t_n)} = \appl {\trmeta f}{\trmeta{t_1} \conss \ldots
  \conss \trmeta{t_n} \conss \nil}$, $\trmeta{c} = \appl {\trmeta c} \nil$ if $c
\in \FF^0$, and $\trmeta x = x$ if $x \in \XX$. 

\subsection{Strategy meta-encoding} 

Given the above meta-encoding of terms, we define a translation~$\rawTRm$ of
rewriting strategies on terms in $\TF$ into plain rewrite rules on terms built
from $\FFmeta$ extended with a set of generated $\varphi$ symbols, with the
symbols $\bot$, $\rawbotlist$ encoding failure, and with the symbols for lists
manipulation presented in Figure~\ref{fig:aux-symbols}. The translations of the
elementary strategies and of the control combinators are almost the same as in
Section~\ref{se:encodings}, so we present in Figure~\ref{fig:metaencoding} only
the strategies whose encoding differs the most from
Figure~\ref{fig:encoding}, and we give the complete translation in 
Appendix~\ref{se:appendixMeta}
(Figure~\ref{fig:metaencoding-complete}).

\begin{figure}[!tp]
\begin{mathpar}
\begin{array}{lr@{\hspace{5pt}}c@{\hspace{5pt}}l}
\mathbf{(ME3)}&
\TRm{l \ra r} & = & \{ ~   \varphi_{l\ra r}(\symb{l}) \ra \symb{r}, 
                                \quad
                                \varphi_{l\ra r}(\Xx \at \symb{\ap l}) \ra  \bot(\Xx),\\[+2pt]
                          &&&  ~~~~~ 
                              \varphi_{l\ra r}(\bot(\Xx)) \ra \bot(\Xx)~~\} 
                              \\[+4pt]
\mathbf{(ME8)}&
\TRm{\all(S)} & =  & \TRm{S} ~\cup~ \TRlists  \\[+2pt]
&\bigcup & \{ &  \varphi_{\all(S)}(\bot(\Xx)) \ra \bot(\Xx) , \\[+2pt]
&&& \varphi_{\all(S)}(\appl{\f}{\args}) \ra \propag{\appl{\f}{\varphi^{list}_{\all(S)}(\args)}} , \\[+2pt]
&&& \varphi^{list}_{\all(S)}(\cons{\head}{\tail}) \ra \varphi'_{\all(S)}(\vpS(\head),\tail,\cons{\head}{\nil},\nil), \\[+2pt]
&&& \varphi^{list}_{\all(S)}(\nil) \ra \nil, \\[+2pt]
&&& \varphi'_{\all(S)}(\bot(\any),\todo,\revtried,\any) \ra \botlist{\rconcat{\revtried}{\todo}}, \\[+2pt]
&&& \varphi'_{\all(S)}(\Xx \at \appl{\any}{\any},\nil,\any,\revdone) \ra \reverse{\cons{\Xx}{\revdone}}, \\[+2pt]
&&& \varphi'_{\all(S)}(\Xx \at \appl{\any}{\any},\cons{\head}{\tail},\revtried,\revdone) \ra \\[+2pt]
&&& \multicolumn{1}{r}{\varphi'_{\all(S)}(\vpS(\head), \tail, \cons{\head}{\revtried}, \cons{\Xx}{\revdone}) ~~\}} \\[+2pt]
    \\
\mathbf{(ME9)}&
  \TRm{\one(S)} & = & \TRm{S}~\cup~ \TRlists \\[+2pt]
&\bigcup & \{ &  \varphi_{\one(S)}(\bot(\Xx)) \ra \bot(\Xx) , \\[+2pt]
&&& \varphi_{\one(S)}(\appl{\f}{\args}) \ra \propag{\appl{\f}{\varphi^{list}_{\one(S)}(\args)}} , \\[+2pt]
&&& \varphi^{list}_{\one(S)}(\cons{\head}{\tail}) \ra  \varphi'_{\one(S)}(\vpS(\head),\tail,\cons{\head}{\nil}), \\[+2pt]
&&& \varphi^{list}_{\one(S)}(\nil) \ra \botlist{\nil}, \\[+2pt]
&&& \varphi'_{\one(S)}(\bot(\any),\nil, \revtried) \ra \botlist{\reverse{\revtried}}, \\[+2pt]
&&& \varphi'_{\one(S)}(\bot(\any),\cons{\head}{\tail}, \revtried) \ra \varphi'_{\one(S)}(\vpS(\head),\tail,\cons{\head}{\revtried}), \\[+2pt]
&&& \varphi'_{\one(S)}(\Xx \at \appl{\any}{\any},\todo,\cons{\any}{\revtried})
\ra \\[+2pt]
&&& \multicolumn{1}{r}{\rconcat{\revtried}{\cons{\Xx}{\todo}}}
~~\} 
\end{array}
\end{mathpar}
\caption{\label{fig:metaencoding}Strategy translation for meta-encoded terms;
  $\TRlists$ is defined in Figure~\ref{fig:aux-symbols}.}
\end{figure}

First, note that meta-encoded terms are of the form $\appl \any \any$,
so instead of using $x \at \ap \bot(\any)$ (as in the rule schemas in
Figure~\ref{fig:encoding}) to filter all the ground terms of the
original signature, we directly use $x \at \appl \any \any$.
As explained in the previous section, the rule schemas relying on $x
\at \ap \bot(\any)$ are expanded into a set of rewrite rules whose
cardinality depends on the original signature. By using the term $x
\at \appl \any \any$ instead of this anti-term, each set of rules
represented by a rule schema in the original encoding is replaced by
only one rewrite rule in the meta-encoding.  Consequently, the number
of rules in the meta-encoding can be significantly smaller than in the
previous encoding.

We should point out that the meta-encoding of a rewrite rule
$\mathbf{(ME3)}$ uses $\trmeta{\ap l}$ and not $\ap{\trmeta l}$,
meaning that we meta-encode the terms of $\TF$ that do not match $l$,
instead of considering all the terms of $\TFmeta$ that do not match
$\trmeta l$. Indeed, $\FFmeta$ is more expressive than $\FF$ since we
can write terms $\appl {\trmeta f}{\args}$ where the size of $\args$
differs from the arity of $f$ in $\FF$. Such ill-formed terms are
matched by $\ap \trmeta{l}$, but not by $\trmeta{\ap l}$; however,
they do not have to be considered since they are never produced during
the evaluation of a meta-encoded term {\wrt} a strategy meta-encoding
(see Lemma~\ref{lem:meta}). As a result, we generate as many rules to
represent this anti-pattern with the meta-encoding as with the
original encoding.

\begin{exa}
\label{ex:encodeIdEtControlMeta}
If we consider the signature in Example~\ref{ex:rr-strategy}, the meta-encoding
of the strategy $S_{pz}\eqdef\Plus(\Zero,x)\ra x$ contains the same rules as in
the regular encoding in Example~\ref{ex:encodeElem}, except that we replace the
terms of the original signature by their meta-level
encodings:\footnote{For readability reasons we write in what follows
  $\appl{f}{\ldots}$ instead of $\appl{\trmeta f}{\ldots}$ for all $f\in\FF$.
}
$$
\begin{array}{ll}
\TRm{S_{pz}} =  \{ &   \varphi_{pz}(\appl{\Plus}{\cons{\cons{\appl{\Zero}{\nil}}{x}}{\nil}}) \ra x, \\
& \varphi_{pz}(\appl{\Zero}{\nil}) \ra  \bot(\appl{\Zero}{\nil}), \\
& \varphi_{pz}(\appl{\Succ}{\cons{\Yy_1}{\nil}}) \ra  \bot(\appl{\Succ}{\cons{\Yy_1}{\nil}}), \\

& \varphi_{pz}(\appl{\Plus}{\cons{\appl{\Succ}{\cons{\Yy_1}{\nil}}}{\cons{\Yy_2}{\nil}}}) 
\ra\\
\multicolumn{2}{r}{\bot(\appl{\Plus}{\cons{\appl{\Succ}{\cons{\Yy_1}{\nil}}}{\cons{\Yy_2}{\nil}}}),} \\

& \varphi_{pz}(\appl{\Plus}{\cons{\appl{\Plus}{\cons{\Yy_1}{\cons{\Yy_2}{\nil}}}}{\cons{\Yy_3}{\nil}}}) 
\ra\\
\multicolumn{2}{r}{\bot(\appl{\Plus}{\cons{\appl{\Plus}{\cons{\Yy_1}{\cons{\Yy_2}{\nil}}}}{\cons{\Yy_3}{\nil}}}),} \\

& \varphi_{pz}(\bot(x)) \ra \bot(x)~~\} 
\end{array}
$$
The following meta-encoding is obtained for the $\id$ strategy:
$$\TRm{\id} = \{~ \varphi_{\id}(\appl{\Yy_1}{\Yy_2}) \ra \appl{\Yy_1}{\Yy_2},~
\varphi_{\id}(\bot(\Xx)) \ra \bot(\Xx)
 ~~\} 
$$
and the  strategy $S_{rpz}=\mu \varX\dotsym\choice{(\seq{\Plus(\Zero,x)\ra
    x~}{~\varX})}{\id}$ from Example~\ref{ex:encodeControl} is encoded by:
$$
\begin{array}{rll}
\TRm{S_{rpz}} =  & \{~  \varphi_{rpz}(\Xx \at \appl{\any}{\any}) \ra  \varphi_{tpz}(\Xx) ,
                           & \varphi_{rpz}(\bot(\Xx)) \ra \bot(\Xx)  ~~\} 
                           \\
                      \cup & \{~    \varphi_{\varX}(\Xx \at \appl{\any}{\any}) \ra  \varphi_{tpz}(\Xx) , 
                           & \varphi_{\varX}(\bot(\Xx)) \ra \bot(\Xx)  ~~\} \\
                      \cup & \{~  \varphi_{tpz}(\Xx \at \appl{\any}{\any}  ) \ra  \phich(\varphi_{pz\varX}(\Xx)) , 
                           & \varphi_{tpz}(\bot(\Xx)) \ra \bot(\Xx) , \\
                           &  ~~~\phich(\Xx \at \appl{\any}{\any}) \ra  \Xx , 
                           &  \phich(\bot(\Xx)) \ra  \varphi_{\id}(\Xx) ~~\} \\
                      \cup & \{~ \varphi_{pz\varX}(\Xx \at \appl{\any}{\any}  ) \ra \phiseq( \varphi_{\varX}(\varphi_{pz}(\Xx)), \Xx) , 
                           & \varphi_{pz\varX}(\bot(\Xx)) \ra \bot(\Xx) , \\
                           &  ~~~ \phiseq(\Xx \at \appl{\any}{\any}, \any) \ra  \Xx , 
                           & \phiseq(\bot(\any), \Xx) \ra  \bot(\Xx) ~~\} \\
                              \cup & \TRm{S_{pz}} \cup \TRm{\id}
\end{array}
$$ 

\noindent In contrast with the translations of the previous section, only
the encoding of the strategy~$S_{pz}$ depends on the
considered signature, all the other rules in the encoding would be the
same independently of the signature.
\end{exa}

\begin{figure}[!tp]
$$
\begin{array}{lr@{\hspace{5pt}}lcl}
  \TRlists & =   \{ &  
    \append{\nil}{\Xx} &\ra& \cons{\Xx}{\nil}, \\
&&  \append{\cons{\head}{\tail}}{\Xx} &\ra& \cons{\head}{\append{\tail}{\Xx}}, \\
&&  \reverse{\nil} &\ra& \nil, \\
&&  \reverse{\cons{\head}{\tail}} &\ra& \append{\reverse{\tail}}{\head}, \\
&&  \rconcat{\nil}{\Xx} &\ra& \Xx, \\
&&  \rconcat{\cons{\head}{\tail}}{\Xx} &\ra& \rconcat{\tail}{\cons{\head}{\Xx}}, \\
&&  \propag{\appl{\f}{\botlist{\args}}} &\ra& \bot(\appl{\f}{\args}), \\
&&  \propag{\appl{\f}{\cons{\head}{\tail}}} &\ra& \appl{\f}{\cons{\head}{\tail}}, \\
&&  \propag{\appl{\f}{\nil}} &\ra& \appl{\f}{\nil} ~~\}
\end{array}
$$
\caption{\label{fig:aux-symbols}Auxiliary symbols and rewrite rules
  for lists manipulation.}
\end{figure}

The meta-level representation of the terms has a more important impact
on the way $\all(S)$ and $\one (S)$ are
translated. When applying (the encoding of) one of these strategies to
a term $\appl f \args$
we want to apply $S$ to each of the elements of the list $\args$ to
check if $S$ succeeds or not on these arguments.
To do so, we have to manipulate explicitly lists of arguments and
we introduce thus some extra symbols and the corresponding rewrite
rules to handle them: $\append \llist x$ adds the element $x$
at the end of $\llist$, $\reverse \llist$ computes the reverse of
$\llist$, $\rconcat {\llist_1}{\llist_2}$ concatenates the reverse of
$\llist_1$ with $\llist_2$, and $\rawpropag$ propagates the failure or
the success of the application of a strategy on the arguments of a
meta-level term to the whole term.
The rewrite rules for these symbols are
given in Figure~\ref{fig:aux-symbols}.

Given a subject $\appl f \args$, the encoding of $\all(S)$ first checks whether
the list $\args$ is empty or not with the rules for the symbol
$\varphi^{list}_{\all(S)}$.  If it is empty, then $\all(S)$ has been applied to
a (meta-)constant $\appl f \nil$, and the result should be the term itself,
which is the case because $\varphi^{list}_{\all(S)}(\nil) \ra \nil$ and
$\propag {\appl f \nil} \ra \appl f \nil$.  If
$\args = t_1 \conss t_2 \conss \ldots \conss t_n \conss \nil$, we go through the
list of arguments, using terms of the form
$\varphi'_{\all(S)}(x, \todo, \revtried, \revdone)$, where $x$ is the element
being reduced, initialized successively with $\varphi_S(t_i)$ for
$ 1 \leq i \leq n$, $\todo$ is a list containing the terms $t_{i+1} \ldots t_n$
remaining to reduce, $\revtried$ is a list containing the terms $t_i \ldots t_1$
already reduced, and $\revdone$ is the list containing the reducts
$t'_i \ldots t'_1$ of the terms in $\revtried$. If $\varphi_S(t_i)$ reduces to
$\bot(t_i)$, then we abort the whole process by producing first
$\botlist{\rconcat \revtried \todo}$, which reduces to $\botlist{\args}$, and
then propagating failure with $\rawpropag$ to generate the term
$\bot(\appl f \args)$.
If reducing $\varphi_S(t_i)$ produces a meta-term $t'_i$, then we
store this result in  $\revdone$ and proceed with the next term in
$\todo$.
If there are no terms left to be evaluated in $\todo$, we build the
final result using the last result $t'_n$ and the previously reduced
arguments $t'_{n-1} \ldots t'_1$ to obtain 
$\appl f {t'_1 \conss \ldots \conss t'_n \conss \nil}$.

\begin{exa}
\label{ex:encodeAll}
The meta-encoding of the strategy $S_{all}\eqdef\all(\Plus(\Zero,x)\ra
x)$ consists of:
$$
\begin{array}{rll}
\TRm{S_{all}} = \{  & 
\varphi_{all}(\bot(\Xx)) \ra \bot(\Xx) , \\
& \varphi_{all}(\appl{\f}{\args}) \ra \propag{\appl{\f}{\varphi^{list}_{all}(\args)}} , \\[2pt]
& \varphi^{list}_{all}(\cons{\head}{\tail}) \ra \varphi'_{all}(\varphi_{pz}(\head),\tail,\cons{\head}{\nil},\nil), \\[2pt]
& \varphi^{list}_{all}(\nil) \ra \nil, \\[2pt]
& \varphi'_{all}(\bot(\any),td,rt,\any) \ra \botlist{\rconcat{rt}{td}}, \\[2pt]
& \varphi'_{all}(\appl{f}{\args},\nil,\any,rd) \ra \reverse{\cons{\appl{f}{\args}}{rd}}, \\[2pt]
& \varphi'_{all}(\appl{f}{\args},\cons{\head}{\tail},rt,rd) \ra \\[2pt]
& \hfill \varphi'_{all}(\varphi_{pz}(\head), \tail, \cons{\head}{rt}, \cons{\appl{f}{\args}}{rd}) ~~\} \\[2pt]
    \cup & \TRm{S_{pz}} ~~\cup~~ \TRlists
\end{array}
$$ 
The translation $\TRctx{\Gamma}{\all(\Plus(\Zero,x)\ra x)}$ would
contain much more rules and their number depends on the considered
signature. 
We can see in  Example~\ref{ex:encodeAllTyped} the impact of the
signature on the size of the encoding using the original terms instead
of their meta-level representation.
\end{exa}

The encoding of $\one(S)$ follows the same principles. If
$\args = t_1 \conss t_2 \conss \ldots \conss t_n \conss \nil$, then we explore the list
with $\varphi'_{\one(S)}(x, \todo, \revtried)$, where $x$ is initialized
successively with $\varphi_S(t_i)$ until the encoding of $S$ succeeds on one of
these terms. The list $\todo$ contains what is left to try (\ie
$t_{i+1} \conss \ldots \conss t_n \conss \nil$), and $\revtried$ contains
$t_i \conss t_{i-1} \conss \ldots \conss t_1 \conss \nil$. If~$x$ becomes a meta-term
$t'_i$, then we rewrite it into
$\rconcat {t_{i-1} \conss \ldots \conss t_1 \conss \nil}{t'_i \conss t_{i+1} \conss
  \ldots \conss t_n \conss \nil}$
which then reduces to the required result. If $S$ fails on $t_i$ (\ie $x$ becomes
$\bot(t_i)$) and $\todo$ is not empty, then we consider the next term by
reducing to
$\varphi'_{\one(S)}(\varphi_S(t_{i+1}), \todo, t_{i+1} \conss \revtried)$.
Otherwise, we reduce to $\botlist{\reverse{\revtried}}$, which produces
$\botlist{\args}$, and the rules for $\rawpropag$ then generate
$\bot(\appl f \args)$, as wished.

\subsection{Properties of the meta-encoding}
We state now the correctness of the meta-encoding by first
establishing the correspondence between $\rawTRm$ and the translation
of Section~\ref{se:encodings}. The next lemma also shows that applying the
strategy meta-encoding to a meta-encoded term results in a meta-encoded term, as
expected. We extend $\trmeta \cdot$ to $\bot$ by defining $\trmeta{\bot(t)} =
\bot(\trmeta t)$.

\begin{restatable}{lemm}{meta}
  \label{lem:meta}
  Given a term $t\in \TF$, a strategy $S$ and a context $\Gamma$ such that
  $\varS S \subseteq \domS \Gamma$, 
  \begin{enumerate}
  \item if $\trsred{\TRctx{\Gamma}{S} \cup \TRg{\Gamma}{\Gamma}}{\vpS(t)}{u}$
    with $u \in \TF$ or $u = \bot(t)$, then $\trsred{\TRm{S} \cup
      \TRmg{\Gamma}}{\vpS(\symb t)}{\symb{u}}$;
  \item if $\trsred{\TRm{S} \cup \TRmg{\Gamma}}{\vpS(\symb t)}{t''}$ and $t''
    \in \TFmeta$, then there exists
    $t'$ such that $\symb{t'} = t''$ and $\trsred{\TRctx{\Gamma}{S} \cup
      \TRg{\Gamma}{\Gamma}}{\vpS(t)}{t'}$.
  \item if $\trsred{\TRm{S} \cup \TRmg{\Gamma}}{\vpS(\symb t)}{\bot(t'')}$ and
    $t'' \in \TFmeta$, then $t''=\symb t$ and $\trsred{\TRctx{\Gamma}{S} \cup
      \TRg{\Gamma}{\Gamma}}{\vpS(t)}{\bot(t)}$.
  \end{enumerate}
\end{restatable}

\begin{proof}[Sketch]
  The proof is by induction on $S$. The correspondence is direct for the
  reduction steps which do not involve $\all$ or $\one$. In these two cases, we
  relate how $\TRm{\all(S)}$ or $\TRm{\one(S)}$ behaves on $\appl {\trmeta f}
  {\trmeta{t_1} \conss \ldots \conss \trmeta{t_n} \conss \nil}$ depending on how $\TRm
  S$ behaves on each of the $\trmeta{t_i}$. For example, we show that
  $\trsred{\TRm{S} \cup \TRmg{\Gamma}}{\vpn{S}(\symb {t_i})}{\symb{t'_i}}$ with
  $\symb{t'_i} \neq \bot(\symb{t_i})$ for all $i$ iff
  \begin{multline*}
  \trsrednll{\TRm{\all(S)}
    \cup \TRmg{\Gamma}}{\varphi'_{\all(S)}(\varphi_{S}(\symb{t_1}), \symb{t_2}
    \conss \ldots \conss \symb{t_n} \conss \nil, \symb{t_1} \conss \nil,
    \nil)}{\symb{t'_1} \conss \ldots \conss \symb{t'_n} \conss \nil}
  \end{multline*}
  From there the induction hypothesis tells us that $\trsred{\TRm{S} \cup
    \TRmg{\Gamma}}{\vpn{S}(\symb {t_i})}{\symb{t'_i}}$ iff
  $\trsred{\TRctx{\Gamma}{S} \cup \TRg{\Gamma}{\Gamma}}{\vpS(t_i)}{t'_i}$, and
  then we can show that $\trsred{\TRctx{\Gamma}{S} \cup
    \TRg{\Gamma}{\Gamma}}{\vpS(t_i)}{t'_i}$ iff $\trsred{\TRctx{\Gamma}{\all(S)}
    \cup \TRg{\Gamma}{\Gamma}}{f(t_1, \ldots, t_n)}{f(t'_1, \ldots, t'_n)}$ to
  conclude. In total, we have four different cases (depending on whether the
  applications of $\all(S)$ or $\one(S)$ succeed or not), which are treated
  similarly. 
\end{proof}

Combined with Theorem~\ref{th:simulation} this result allows us to
deduce the correspondence between the strategy application and $\rawTRm$.

\begin{thm}[Simulation]
\label{th:simulation-meta}
Given a term $t\in \TF$, a strategy $S$ and a context $\Gamma$ such that $\varS
S \subseteq \domS \Gamma$,
\begin{enumerate}
\item 
$\stratappctx{\Gamma}{S}{t} \sred t'$
~~iff~~
$\trsred{\TRm{S} \cup \TRmg{\Gamma}}{\vpS(\trmeta t)}{\trmeta{t'}}$, $t'\in \TF$
\item 
\textnormal{$\stratappctx{\Gamma}{S}{t} \sred \failres$}
~~iff~~
$\trsred{\TRm{S} \cup \TRmg{\Gamma}}{\vpS(\trmeta t)}{\bot(\trmeta t)}$  
\end{enumerate}
\end{thm}

\begin{proof}
  By Lemma~\ref{lem:meta} and Theorem~\ref{th:simulation}.
\end{proof}

\section{Encoding rewriting strategies with typed rewrite rules}
\label{se:typing}
Many programming languages use type systems to classify values and
expressions into types, to define how those types can be manipulated
and how they interact.
In logic~\cite{Manzano:1993:IML},
many-sorted signatures are used similarly to partition the
universe into disjoint subsets, one for every sort, and
a many-sorted logic
naturally leads to a type theory~\cite{JacobsCLTT}.
Indeed, sorts of many-sorted signatures are also known as algebraic
data-types in programming languages.

When using many-sorted signatures the strategy translations should
guarantee that all the generated terms and rewrite rules respect the
corresponding construction constraints.
Traversal strategies propagate the application of a given strategy to
(sub-)terms of potentially different sorts and consequently the
strategies we consider here are intrinsically polymorphic. This
polymorphic nature should be clearly retrieved in the encoding and
thus, the sorts of the symbols generated during the translation should
cope with this constraint. To support this behaviour we consider
many-sorted signatures with potentially overloaded symbols and we show
that the proposed translations can be adapted easily to work even when
overloaded symbols are not supported by the targeted language.

We consider in what follows sort preserving rewrite rules and
consequently sort preserving strategies. We propose a translation
generating sort preserving TRS which are, as before, faithful
encodings of the corresponding strategies.

\subsection{Many-sorted signatures and term rewriting systems}
A \emph{many-sorted signature} $\Sigma=(\SO,\FF)$, or simply $\Sigma$,
consists of a set of sorts $\SO$ and a set of symbols $\FF$. A symbol
$f$ with \emph{domain} $w=\sst_1\prosep\ldots\prosep\sst_n\in\SO^*$ and
\emph{co-domain} $\sst$ is written $f{\sigd}w{\sarrow}\sst$; $n$ is its arity
and $w{\sarrow}\sst$ its profile.
Symbols can be overloaded, {\ie} a symbol $f$ can have profiles
$w{\sarrow}\sst$ and $w'{\sarrow}\sst'$ with $w\neq w'$.

We write $\FFs{\sst}$ for the subset of symbols of codomain $\sst$.
Variables are also sorted and $\Xx{\sigd}s$ means that variable $\Xx$ has
sort $s$. We sometimes annotate the name of a variable by its sort
and use, for example, $\Xxs{}{s}$ to implicitly indicate a variable of
sort $s$. The set $\XX_{\sst}$ denotes a set of variables of sort
$\sst$ and $\XX = \bigcup_{\sst\in\SO} \XX_{\sst}$ is the set of
sorted variables.

The set of terms of sort $\sst$, denoted $\TFXs{\sst}$ is the smallest
set containing ${\XX}_{\sst}$ and such that $f(t_1,\ldots,t_n)$ is in
$\TFXs{\sst}$ whenever $f{\sigd}\sst_1\prosep\ldots\prosep\sst_n {\sarrow}
\sst$ and $t_i \in \TFXs{\sst_i}$ for $i\in[1,n]$.
We write $t{\sigd}s$ when the term $t$ is of sort $s$, {\ie} when $t\in\TFXs{\sst}$.
The set of \emph{sorted terms} is defined as
$\TFXs{\SO}=\bigcup_{\sst\in\SO}\TFXs{\sst}$.  Sorted substitutions
are defined as mappings $\sigma$ from sorted variables to sorted terms
such that if $x{\sigd}s$ then $\sigma(x) \in \TFXs{s}$. Note that for
any such sorted substitution $\sigma$, $t{\sigd}s$ iff
$\sigma(t){\sigd}s$.

A \emph{sorted rewrite rule} is a rewrite rule $l \ra r$ with
$l,r{\sigd}\sst$ and a sorted TRS is a set of sorted rewrite rules
inducing a corresponding {rewriting relation} over sorted terms.
Given two sorted terms $l,t{\sigd}\sst$, we say that $l$ matches $t$ if
there exists a sorted substitution $\sigma$ such that
$t=\sigma(l)$; in this case we say that $t$ is a sorted instance
of $l$.

\subsection{Typed encoding of rewriting strategies}
Given a strategy built over a many-sorted signature, the translation defined in
Section~\ref{se:encodings} can still be used to generate a faithful encoding in
the sense of Theorem~\ref{th:simulation}. Since any term in $\TFs{\SO}$ is also
a term in $\TF$, the termination of the (unsorted) TRS encoding the strategy
guarantees the termination of the sorted strategy. Nevertheless, we can use the
extra information provided by sorts to refine this translation and remove the
generated rewrite rules that are useless, because they cannot be applied to
sorted terms. More importantly, this translation cannot be used as a strategy
compiler if the target language is many-sorted and accepts only sorted terms and
rewrite rules.

The presence of ill-sorted terms in the encoding generated by the translation in
Figure~\ref{fig:encoding} is essentially due to the unsorted semantics we have
considered so far for anti-terms. In an unsorted world, the anti-term $\ap t$
represents all the terms which do not match $t$, and $\ap \bot(x)$ denotes all
the ground terms in $\TF$. We adapt these notions to a many-sorted setting as
follows: given a term $t$ of sort $s$, we write $\apt{\sst} t$ for the terms of
sort $\sst$ which do not match~$t$, and given a sort $\sst$ we write
$\apt{\sst} \bott{\sst}(x)$ to denote all the ground sorted terms in $\TFs{\sst}$.
Note that $\bigcup_{\sst\in\SO}\{\apt{\sst} \bott{\sst}(x)\}$ denotes all the
ground \emph{sorted} terms $\TFs{\SO}$ of the original signature.

\begin{exa}
\label{ex:typedAPs}

Consider the many-sorted signature $(\SO,\FF)$ where $\SO=\{\ints\sigsep\bool\}$ and
$\FF=\FFs{\ints}\cup\FFs{\bool}$ with
$\FFs{\ints}=\{\Zero{\sigd}{\sarrow}\ints \sigsep \Succ{\sigd}\ints{\sarrow}\ints
\sigsep \Plus{\sigd}\ints \prosep \ints{\sarrow}\ints\}$,
$\FFs{\bool}=\{\true{\sigd}{\sarrow}\bool \sigsep \false{\sigd}{\sarrow}\bool \sigsep \odd{\sigd}\ints{\sarrow}\bool \sigsep \even{\sigd}\ints{\sarrow}\bool\}$.
The anti-term $\apt{\ints} \Plus(\Zero,x)$ denotes exactly the sorted
terms in $\TFs{\ints}$ matched by $\Zero$, $\Succ(x_1)$,
$\Plus(\Succ(x_1),x_2)$ or $\Plus(\Plus(x_1,x_2),x_3)$. In particular,
it does not denote the (ill-sorted) term $\Plus(\true,\Zero)$.
The anti-term $\apt{\bool} \Plus(\Zero,x)$ denotes the
sorted terms in $\TFs{\bool}$ matched by $\true$, $\false$,
$\odd(x_1)$ or $\even(x_1)$. 

The anti-term $\apt{\bool} \bott{\bool}(x)$ denotes all sorted instances of $\true$,
$\false$, $\odd(x_1)$ and $\even(x_1)$ while $\apt{\ints} \bott{\ints}(x)$
denotes all sorted instances of $\Zero$, $\Succ(x_1)$ and
$\Plus(x_1,x_2)$. The union of all these instances represent indeed
all the sorted terms in $t\in\TFs{\SO}$.
\end{exa}

We now adapt the unsorted translation defined in Section~\ref{se:encodings} to
accommodate many-sorted signatures. Given a sorted signature $(\SO, \FF)$, we
define the signature $(\SO,\FFgen{S})$ where the sorts are unchanged, and $\FF$
is extended with the generated $\varphi$ symbols and $\bot$. Unlike in the
unsorted case, we have to specify their profile: the symbols $\bot$ and
$ \varphi_{\id}$, $\varphi_{\fail}$, $\varphi_{l\ra r}$,
$\varphi_{\seqsub{S_1}{S_2}}$, $\varphi_{\choicesub{S_1}{S_2}}$, $\phich$,
$\varphi_{\mu X\dotsym S}$, $\varphi_{X}$, $\varphi_{\all(S)}$,
$\varphi_{\one(S)}$ have sort $\sst{\sarrow}\sst$ for all $\sst \in \SO$, and
$ \phiseq{\sigd}\sst\prosep\sst{\sarrow}\sst $ for all $\sst\in\SO$. These
symbols are overloaded, since they have as many profiles as there are sorts in
$\SO$. In contrast, the profiles of the  generated symbols of the form $\psiF$ and
$\psif{i}$ depend on the profile of $f$: if
$f{\sigd}\sst_1\prosep\ldots\prosep\sst_n {\sarrow} \sst \in
\FFs{\sst}^{\supzero}$
then $\psiF{\sigd}\sst_1\prosep\ldots\prosep\sst_n\prosep\sst {\sarrow} \sst$
and $\psif{i}{\sigd}\sst_1\prosep\ldots\prosep\sst_n {\sarrow} \sst$.

The translation $\TRctxS{}{S}$ transforms the strategy $S$ into a set of sorted
rewrite rules; since its definition is the same as for the unsorted case
(Figure~\ref{fig:encoding}), except for extra sort annotations on variables and
$\ap$, we give the rules only in Appendix~\ref{se:propertiesSortedAppendix} (Figure~\ref{fig:encodingTyped}). Even
though the definitions of $\mathbb T$ and $\mathbb T_{\mathcal S}$ are almost
the same, taking sorts into account could have an important impact on
the number of rules
generated in the resulting encoding. Sorts may introduce duplication: since
$\bot$ has now a profile, for every rule of the form
$\varphi(\bot(\Xx)) \ra \bot(\Xx)$ generated in the unsorted case, we generate
now the set of rules
$\bigcup_{\sst\in\SO}\{\varphi(\bot(\Xxs{}{\sst})) \ra \bot(\Xxs{}{\sst})\}$. In
contrast, expanding $\ap\bot(\any)$ in the unsorted case produces as many rules
as considering $\bigcup_{\sst\in\SO}\{\apt{\sst}\bott{\sst}(x)\}$ in the sorted
case, assuming $\FF$ has as many elements as $\bigcup_{\sst \in \SO} \FFs\sst$.
As detailed in Section~\ref{sec:genTRS}, $\ap\bot(\any)$ is expanded into the
set of terms $\bigcup_{f\in\FF}\{f(x_1,\ldots,x_{\arty(f)})\}$, while
$\bigcup_{\sst\in\SO}\{\apt{\sst}\bott{\sst}(x)\}$
becomes the generally smaller set
$\bigcup_{\sst\in\SO}\bigcup_{f\in\FFs{\sst}}\{f(x_1,\ldots,x_n)\}$.

\begin{exa}
\label{ex:encodeIdTyped}
If we consider the many-sorted signature in Example~\ref{ex:typedAPs},
we obtain the following encoding for the $\id$ strategy when expanding
the rule schemas in Figure~\ref{fig:encodingTyped}:\\
\[
\begin{array}{lll}
\TRctxS{\Gamma}{\id} =  \{ 
  & \varphi_{\id}(\Zero) \ra \Zero 
  & \varphi_{\id}(\true) \ra \true, 
 \\
 &  \varphi_{\id}(\Succ(\Yy_1)) \ra  \Succ(\Yy_1), 
 & \varphi_{\id}(\false) \ra \false, 
 \\
 &  \varphi_{\id}(\Plus(\Yy_1,\Yy_2)) \ra  \Plus(\Yy_1,\Yy_2), 
 & \varphi_{\id}(\odd(\Yy_1)) \ra  \odd(\Yy_1), 
 \\
 &  \varphi_{\id}(\bot(\Xx)) \ra \bot(\Xx),
 & \varphi_{\id}(\even(\Yy_1)) \ra  \even(\Yy_1), 
 \\
 &
 & \varphi_{\id}(\bot(\Zz)) \ra \bot(\Zz)
 ~~\} 
 \\
\end{array}
\]
To improve readability, we do not annotate variables but instead use an
environment mapping variables to sorts: here, we have
$\Xx,\Yy_1,\Yy_2{\sigd}\ints$
and $\Zz{\sigd}\bool$.
Besides, the signature is enriched in this case with the symbols $
\{~\bot{\sigd} \ints {\sarrow} \ints \sigsep$ $\varphi_{\id}{\sigd} \ints
{\sarrow} \ints \sigsep$ $\bot{\sigd} \bool {\sarrow} \bool
\sigsep$ $\varphi_{\id}{\sigd} \bool {\sarrow} \bool~\} $.

For this signature, the unsorted encoding would contain the same
rewrite rules modulo syntactic equivalence, {\ie} the rules in
$\TRctx{\Gamma}{\id}$ are those in $\TRctxS{\Gamma}{\id}$ with the
rules $\varphi_{\id}(\bot(\Xx)) \ra \bot(\Xx)$ and
$\varphi_{\id}(\bot(\Zz)) \ra \bot(\Zz)$ collapsed into only one rule.

\end{exa}

The unsorted and sorted encodings differ more on how they handle rewrite rules
and traversal combinators. As explained before, the anti-pattern
$\varphi_{l\ra r}(\Xx \at \ap l) \ra \bot(\Xx)$ is expanded differently between
the unsorted and sorted cases: we consider only the terms of the same sort as
$l$ in the sorted translation, thus reducing the number of generated rules
compared to the unsorted case.

\begin{exa}
\label{ex:encodeElemTyped}

Let $S_{pz}\eqdef\Plus(\Zero,x)\ra x$, $\SO=\{\ints,\bool\}$, and
$\FF=\FFs{\ints}\cup\FFs{\bool}$ with
$\FFs{\bool}=\{\odd{\sigd}\ints{\sarrow}\bool \sigsep
\even{\sigd}\ints{\sarrow}\bool \sigsep \true~{\sarrow}\bool \sigsep
\false~{\sarrow}\bool\}$. Then 
\[
\begin{array}{lll}
\TRctxS{\Gamma}{S_{pz}} =  \{ &   \varphi_{pz}(\Plus(\Zero,x)) \ra x, \\
  & \varphi_{pz}(\Zero) \ra  \bot(\Zero), 
  & \varphi_{pz}(\true) \ra  \bot(\true),\\
  & \varphi_{pz}(\Succ(\Yy_1)) \ra  \bot(\Succ(\Yy_1)), 
  & \varphi_{pz}(\false) \ra  \bot(\false),\\
  & \varphi_{pz}(\Plus(\Succ(\Yy_1),\Yy_2)) \ra \bot(\Plus(\Succ(\Yy_1),\Yy_2)), 
  &   \varphi_{pz}(\odd(\Yy_1)) \ra  \bot(\odd(\Yy_1)), \\
  & \varphi_{pz}(\Plus(\Plus(\Yy_1,\Yy_2),\Yy_3)) \ra  \bot(\Plus(\Plus(\Yy_1,\Yy_2),\Yy_3)), 
  & \varphi_{pz}(\even(\Yy_1)) \ra  \bot(\even(\Yy_1)), \\
  & \varphi_{pz}(\bot(x)) \ra \bot(x),
  & \varphi_{pz}(\bot(z)) \ra \bot(z)~~\} 
\end{array}
\]
so that $\Xx,\Yy_1,\Yy_2,\Yy_3{\sigd}\ints, z{\sigd}\bool$, and the signature is
extended with $\{~\bot{\sigd}\ints{\sarrow}\ints \sigsep$
$\varphi_{pz}{\sigd}\ints{\sarrow}\ints \sigsep$
$\bot{\sigd}\bool{\sarrow}\bool \sigsep$
$\varphi_{pz}{\sigd}\bool{\sarrow}\bool~\}$.

The unsorted translation $\TRctx{\Gamma}{S_{pz}}$ 
contains the rules in $\TRctxS{\Gamma}{S_{pz}}$ as well as the rules
\[
\begin{array}{lll}
  ~\{ 
  & \varphi_{pz}(\Plus(\odd(\Yy_1),\Yy_2)) \ra \bot(\Plus(\odd(\Yy_1),\Yy_2)), 
  & \varphi_{pz}(\Plus(\true,\Yy_2)) \ra \bot(\Plus(\true,\Yy_2)), \\
  & \varphi_{pz}(\Plus(\even(\Yy_1),\Yy_2)) \ra \bot(\Plus(\even(\Yy_1),\Yy_2)), 
  & \varphi_{pz}(\Plus(\false,\Yy_2)) \ra \bot(\Plus(\false,\Yy_2))
  ~~\} 
\end{array}
\]
which operate on ill-sorted terms {\wrt} the sorted signature we consider here.
\end{exa}

For the encodings of $\all$ and $\one$, we can also exploit the
profiles of the symbols and generate less rules to filter the ground
terms of the original many-sorted signature. Given
$f{\sigd}\sst_1\prosep\ldots\prosep\sst_n{\sarrow} \sst$, encoding
produces now schemas of the form
$\psiF(\Xx_1\at\apt{\sst_1}\bott{\sst_1}(\any),\ldots,\Xx_n\at\apt{\sst_n}\bott{\sst_n}(\any),\Yys{}{\sst})
\ra f(\Xx_1,\ldots,\Xx_n)$ and
$\psif{i}(\bot(\Xxs{1}{\sst_1}),\ldots,\bot(\Xxs{i-1}{\sst_{i-1}}),\Xx_i\at\apt{\sst_i}\bott{s_i}(\any),\Xxs{i+1}{\sst_{i+1}},\ldots,\Xxs{n}{\sst_n})
\ra f(\Xxs{1}{\sst_1},\ldots,\Xxs{n}{\sst_n})$ for respectively $\all$
and $\one$ and thus generates significantly less rules than in the
unsorted case where we used the exhaustive $\ap\bot(\any)$.
 
The encodings of the other
operators (sequence, choice, \ldots) also filter terms using $x \at
\apt{\sst}\bott{\sst}(\any)$, but for all $\sst \in \SO$ and not for a
given $\sst$, so the number of generated rules for these constructs is
the same as in the unsorted case.

\begin{exa}
\label{ex:encodeAllTyped}
If we consider the many-sorted signature in Example~\ref{ex:typedAPs},
we obtain\\
 $\TRctxS{\Gamma}{\all(\Plus(\Zero,x)\ra x)} =  \TRctxS{\Gamma}{S_{pz}} ~\bigcup~ \{$\\[-5pt]
\begin{minipage}{.57\textwidth}
\[
\begin{array}{lll}
\varphi_{\all}(\bot(\Xx)) \ra \bot(\Xx),
\\
\varphi_{\all}(\Plus(\Xx_1,\Xx_2)) \ra \varphi_{\Plus}(\varphi_{pz}(\Xx_1),\varphi_{pz}(\Xx_2),\Plus(\Xx_1,\Xx_2)),
\\
\varphi_{\Plus}(\bot(\Xx_1),\Xx_2,\Xx) \ra \bot(\Xx), 
\\
\varphi_{\Plus}(\Xx_1,\bot(\Xx_2),\Xx) \ra \bot(\Xx), 
\\
\varphi_{\Plus}(\Plus(\Xx_1,\Xx_2),\Plus(\Yy_1,\Yy_2),\Xx) \ra \Plus(\Plus(\Xx_1,\Xx_2),\Plus(\Yy_1,\Yy_2)), 
\\
\varphi_{\Plus}(\Plus(\Xx_1,\Xx_2),\Succ(\Yy_1),\Xx) \ra \Plus(\Plus(\Xx_1,\Xx_2),\Succ(\Yy_1)), 
\\
\varphi_{\Plus}(\Plus(\Xx_1,\Xx_2),Z,\Xx) \ra \Plus(\Plus(\Xx_1,\Xx_2),Z), 
\\
\varphi_{\Plus}(\Succ(\Xx_1),\Plus(\Yy_1,\Yy_2),\Xx) \ra \Plus(\Succ(\Xx_1),\Plus(\Yy_1,\Yy_2)), 
\\
\varphi_{\Plus}(\Succ(\Xx_1),\Succ(\Yy_1),\Xx) \ra \Plus(\Succ(\Xx_1),\Succ(\Yy_1)), 
\\
\varphi_{\Plus}(\Succ(\Xx_1),Z,\Xx) \ra \Plus(\Succ(\Xx_1),Z), 
\\
\varphi_{\all}(\bot(\Zz)) \ra \bot(\Zz),
\\
\varphi_{\all}(\even(\Xx_1)) \ra \varphi_{\even}(\varphi_{pz}(\Xx_1),\even(\Xx_1)),
\\
\varphi_{\even}(\bot(\Xx_1),\Zz) \ra \bot(\Zz), 
\\
\varphi_{\even}(\Plus(\Xx_1,\Xx_2),\Zz) \ra \even(\Plus(\Xx_1,\Xx_2)), 
\\
\varphi_{\even}(\Succ(\Xx_1),\Zz) \ra \even(\Succ(\Xx_1)), 
\\
\varphi_{\even}(Z,\Zz) \ra \even(Z), 
\\
\varphi_{\all}(\odd(\Xx_1)) \ra \varphi_{\odd}(\varphi_{pz}(\Xx_1),\odd(\Xx_1)),
\end{array}
\]
\end{minipage}%
\begin{minipage}{.35\textwidth}
\[
\begin{array}{lll}
~
\\
\varphi_{\Plus}(Z,\Plus(\Xx_1,\Xx_2),\Xx) \ra \Plus(Z,\Plus(\Xx_1,\Xx_2)), 
\\
\varphi_{\Plus}(Z,\Succ(\Xx_1),\Xx) \ra \Plus(Z,\Succ(\Xx_1)), 
\\
\varphi_{\Plus}(Z,Z,\Xx) \ra \Plus(Z,Z), 
\\
\varphi_{\all}(\Succ(\Xx_1)) \ra \varphi_{\Succ}(\varphi_{pz}(\Xx_1),\Succ(\Xx_1)), 
\\
\varphi_{\Succ}(\bot(\Xx_1),\Xx) \ra \bot(\Xx), 
\\
\varphi_{\Succ}(\Plus(\Xx_1,\Xx_2),\Xx) \ra \Succ(\Plus(\Xx_1,\Xx_2)), 
\\
\varphi_{\Succ}(\Succ(\Xx_1),\Xx) \ra \Succ(\Succ(\Xx_1)), 
\\
\varphi_{\Succ}(Z,\Xx) \ra \Succ(Z), 
\\
\varphi_{\all}(Z) \ra Z,
\\
~
\\
\varphi_{\odd}(\bot(\Xx_1),\Zz) \ra \bot(\Zz), 
\\
\varphi_{\odd}(\Plus(\Xx_1,\Xx_2),\Zz) \ra \odd(\Plus(\Xx_1,\Xx_2)),
\\
\varphi_{\odd}(\Succ(\Xx_1),\Zz) \ra \odd(\Succ(\Xx_1)), 
\\
\varphi_{\odd}(Z,\Zz) \ra \odd(Z), 
\\
\varphi_{\all}(\false) \ra \false, 
\\
\varphi_{\all}(\true) \ra \true ~~\} 
\end{array}
\]
\end{minipage}\\[2pt]

\noindent so that $\Xx,\Xx_1,\Xx_2,\Yy_1,\Yy_2{\sigd}\ints, \Zz{\sigd}\bool$,
and we extend the original signature with 
\begin{align*}
\{~& 
\varphi_{\all}{\sigd}\ints{\sarrow}\ints \sigsep
\varphi_{\all}{\sigd}\bool{\sarrow}\bool \sigsep 
\bot{\sigd}\ints{\sarrow}\ints \sigsep
\bot{\sigd}\bool{\sarrow}\bool \sigsep \\
& \varphi_{pz}{\sigd}\ints{\sarrow}\ints \sigsep 
\varphi_{pz}{\sigd}\bool{\sarrow}\bool \sigsep 
\varphi_{\Plus}{\sigd}\ints\prosep\ints\prosep\ints{\sarrow}\ints \sigsep
\varphi_{\Succ}{\sigd}\ints\prosep\ints{\sarrow}\ints \sigsep \\
&\varphi_{\even}{\sigd}\ints\prosep\bool{\sarrow}\bool \sigsep 
\varphi_{\odd}{\sigd}\ints\prosep\bool{\sarrow}\bool ~\}.
\end{align*}
The unsorted translation $\TRctx{\Gamma}{\all(\Plus(\Zero,x)\ra x)}$
would contain not only the above rules but also a significant number
of rules which would be ill-sorted {\wrt} the many-sorted signature
considered here.
\end{exa}

We show in Section~\ref{se:propertiesSorted} that although the sorted
translation generates less rules (modulo syntactic equivalence) than the
unsorted translation, we still obtain a faithful encoding for sorted terms.

\begin{rem}
  We suppose in this section that the targeted language supports overloaded
  symbols, but if it is not the case, like, \eg in {\tom}, the sorted translation
  can be easily adapted to fit this constraint: instead of overloading the
  generated symbols, we add one symbol for each sort. For example, the original
  signature would be extended in this case with the symbol(s)
  $\bigcup_{\sst\in\SO}\{\bot^{\sst}{\sigd}\sst{\sarrow}\sst\}$.
  We then use the symbol of the appropriate sort in the rules of
  Figure~\ref{fig:encodingTyped}, depending of the sort of its argument(s). For
  example, the translation of a rewrite rule $l\ra r$ with $l{\sigd}s_l$
  becomes:
$$
\TRctxS{\Gamma}{l \ra r} = \{ ~ \varphi^{\sst_{l}}_{l\ra r}(l) \ra r~\}
\bigcup_{\sst\in\SO}  \{~ \varphi^{\sst}_{l\ra r}(\Xx \at \apt{\sst} l) \ra \bot^{\sst}(\Xx),
                            \quad 
                            \varphi^{\sst}_{l\ra r}(\bot^{\sst}(\Xxs{}{\sst})) \ra \bot^{\sst}(\Xxs{}{\sst}) ~\} 
$$
\end{rem}

\subsection{Properties of the many-sorted encoding}
\label{se:propertiesSorted}

It is easy to see that for each of the generated rewrite rules, for
any sort assignment for the variables such that the left-hand side is
well-sorted, the right-hand side is also well-sorted and has the same
sort as the left-hand side. Starting from this observation we can show
that the reduction of a sorted term {\wrt} a strategy encoding is
sort preserving:

\begin{restatable}[Subject reduction]{lemm}{subRed}
\label{th:subRed}
Consider a many-sorted signature $(\SO,\FF)$, a strategy $S$, a context $\Gamma$
such that $\varS S \subseteq \domS \Gamma$, and the term rewriting system
$\RR$=$\TRctxS{\Gamma}{S} \cup \TRgS{\Gamma}{\Gamma}$ built over the extended
signature $(\SO,\FFgen{})$.  Given a term $t\in\TFsf{s}{\FFgen{}}$ for some sort
$\sst\in\SO$,
if $t\longrightarrow_{\RR}t'$ then $t'\in\TFsf{\sst}{\FFgen{}}$.
\end{restatable}
\begin{proof}[Sketch]
  By induction on the structure of the strategy $S$ (and the context $\Gamma$)
  and by case analysis on the rewrite rule applied in the reduction.  For each
  case we consider reductions at the top position since the replacement of a
  sub-term by another one of the same sort is obviously sort preserving.
\end{proof}

As we have seen in the previous section, the rewrite rules generated by the
sorted translation are the ones generated by the unsorted encoding that are
well-sorted {\wrt} the considered many-sorted signature.  Since an ill-sorted
rewrite rule cannot be applied on a sorted term, the reduction of a such a
sorted term is the same if we use the sorted or unsorted encoding.

\begin{restatable}[Equivalent reductions]{lemm}{equivalenceSorted}
\label{th:equivalenceSorted}
Consider a many-sorted signature $(\SO,\FF)$, a strategy~$S$, a context $\Gamma$
such that $\varS S \subseteq \domS \Gamma$, and the term rewriting systems
$\RR$=$\TRctx{\Gamma}{S} \cup \TRg{\Gamma}{\Gamma}$ and
$\RR_{\SO}$=$\TRctxS{\Gamma}{S} \cup \TRgS{\Gamma}{\Gamma}$ built over the
extended signature $(\SO,{\FFgen{}})$.
Given a term $t\in\TFsf{\SO}{\FFgen{}}$,
$t\longrightarrow_{\RR_{\SO}}t'$ iff $t\longrightarrow_{\RR}t'$.
\end{restatable}
\begin{proof}[Sketch]
  Every rewrite rule in $\RR_{\SO}$ is also included in $\RR$ and
  thus, if $t\longrightarrow_{\RR_{\SO}}t'$ then
  $t\longrightarrow_{\RR}t'$.
  For the other direction we proceed by induction on the structure of the
  strategy $S$ (and the context $\Gamma$) and by case analysis on the rewrite
  rule applied in the reduction.

  The interesting cases concern the rule schemas using anti-patterns in the
  encodings of a rewrite rule, $\all$, and $\one$. For the first one, we remark
  that if a rewrite rule corresponding to the rule schema
  $\varphi_{l\ra r}(\Xx \at \apt{\sst} l) \ra \bot(\Xx)$ from the unsorted
  translation is applied to a sorted term then, there exists an identical
  rewrite rule corresponding to the similar rule schema from the sorted
  translation which can be applied. For the other two cases, we proceed
  similarly and use the fact that if a rewrite rule corresponding to one of the
  rule schemas relying on a $\ap\bot(\any)$ is applied, then the same rewrite
  rule is also exhibited by the similar rule schema using the corresponding
  pattern $\apt{\sst}\bott{\sst}(\any)$ from the sorted translation.
\end{proof}

The one-step reduction of a sorted term is thus exactly the same
when using the sorted or unsorted encodings and since the latter is
also sort preserving we can conclude that the sorted translation
produces faithful strategy encodings.

\begin{restatable}[Simulation]{thmm}{simulationSorted}
\label{th:simulationSorted}
Given a a many-sorted signature $(\SO,\FF)$, 
a strategy $S$, 
two terms $t,t'\in\TFs{\SO}$, 
and a context $\Gamma$ such that $\varS S \subseteq \domS \Gamma$,
\begin{enumerate}
\item 
$\stratappctx{\Gamma}{S}{t} \sred t'$
~~iff~~
$\trsred{\TRctxS{\Gamma}{S} \cup \TRgS{\Gamma}{\Gamma}}{\vpS(t)}{t'}$, 
%
\item 
\textnormal{$\stratappctx{\Gamma}{S}{t} \sred \failres$}
~~iff~~
$\trsred{\TRctxS{\Gamma}{S} \cup \TRgS{\Gamma}{\Gamma}}{\vpS(t)}{\bot(t)}$  
\end{enumerate}
\end{restatable}
\begin{proof}
  Follows immediately from Lemma~\ref{th:subRed},
  Lemma~\ref{th:equivalenceSorted} and Theorem~\ref{th:simulation}.
\end{proof}

The sorted translation can be consequently used as a strategy compiler
for many-sorted languages and, as explained above, it can be used for
languages allowing symbol overloading or not. If we collapse all
syntactically equivalent rules in a sorted encoding we obtain less
rules than in the corresponding unsorted encoding but we can still
feed it into (usually unsorted) termination tools to verify the
termination of the corresponding strategy for sorted terms.

\section{Implementation and experimental results}
\label{se:implementation}
The strategy translations presented in the previous sections
have been implemented in a tool called
{\strategyanalyser}\footnote{\url{http://github.com/rewriting/tom/tree/master/applications/strategyAnalyzer}},
written in {\tom}, a language that extends {\java} with high level
constructs for pattern matching, rewrite rules and strategies ({\ie}
the tool itself is written using rules and strategies!).  Given
a set of rewrite rules guided by a strategy, the tool generates a
plain TRS in {\aprove}/{\TTT}
syntax\footnote{\url{http://aprove.informatik.rwth-aachen.de/}} or
{\tom} syntax (restricted to rewrite rules only).

The tool can be configured to generate TRS at meta-level or not, in a
many-sorted context or not, to use the alias notation
or not, and to use the notion of anti-term or not.  An encoding using
anti-terms and aliasing can be directly used in a {\tom} program but
for languages and tools which do not offer such primitives, aliases
and anti-terms have to be expanded into plain rewrite rules.
We explain first how this expansion is realized and we illustrate then
our approach on several representative examples.

\subsection{Expansion of anti-terms}
The rules given in Figure~\ref{fig:encoding} can generate two kinds of
rules which contain anti-terms.  The first family is of the form
$\varphi(\ldots,y_i\at\ap \bot(\any),\ldots) \ra u$ with $y_i\in\XX$,
and with potentially several occurrences of $\ap\bot(\any)$.  These
rules can be easily expanded into a family of rules
$\varphi(\ldots,y_i\at f(x_1,\ldots,x_n),\ldots)\ra u$ with such a
rule for all $f\in\FF$, and with $x_1,\ldots,x_n\in\XX$ and
$n=\arty(f)$.  This expansion is performed recursively to eliminate
all the instances of $\ap\bot(\any)$.
The other rules containing anti-terms come from the translation of
rewrite rules $\mathbf{(E3)}$ and have in the unsorted case the form
$\varphi(y\at\ap f(t_1,\ldots,t_n)) \ra \bot(y)$, with $f\in\FF^{n}$
and $t_1,\ldots,t_n\in \TFX$.  If the term
$f(t_1,\ldots,t_n)$ is linear, then the tool generates two families of
rules:
\begin{itemize}
  \item $\varphi(g(x_1,\ldots,x_m))\ra \bot(g(x_1,\ldots,x_m))$
    for all
    $g\in\FF$, 
    $g\neq f$, $x_1,\ldots,x_m\in\XX$, $m=\arty(g)$,
  \item $\varphi(f(x_1,\ldots,x_{i-1},x_i\at\ap
    t_i,x_{i+1},\ldots,x_n))\ra \bot(f(x_1,\ldots,x_n))$
    for all
    $i\in[1,n]$ and $t_i\not\in\XX$,
\end{itemize}
with the second family of rules recursively expanded, using the same
algorithm, until there is no anti-term left.

This expansion mechanism is more difficult when we want to find a
convenient (finite) encoding for non-linear anti-terms, and in this
case the expansion should be done, in fact, {\wrt} the entire
translation of a rewrite rule.
Given the rules $\varphi(l) \ra r$ and $\varphi(y\at\ap l) \ra
\bot(y)$ with $l\in\TFX$ a non-linear term, we consider the linearized
version of~$l$, denoted $l'$, with all the variables $x_i\in\var l$
appearing more than once ($m_i$ times, with $m_i>1$) renamed into
$z_i^1,\ldots,z_i^{m_i-1}$ (the first occurrence of $x_i$ is not
renamed).  Then, these two rules can be translated into:
\begin{itemize}
  \item $\varphi(y\at\ap l') \ra \bot(y)$
  \item $\varphi(l') \ra \varphi'(l', 
    x_1=z_1^1\land\cdots\land x_1=z_1^{m_1-1}
    \land\cdots\land
    x_n=z_n^1\land\cdots\land x_n=z_n^{m_n-1})$
  \item $\varphi'(l', \truebi) \ra r$
  \item $\varphi'(l', \falsebi) \ra \bot(l')$
\end{itemize}
with the first rule containing now the linear anti-term $\ap l'$
expanded as previously. 
The rules generated for equality and conjunction are as expected.

The expansion goes in the same way in the many-sorted case, but using
only symbols of the appropriate sort; a sort $\boolbi$, an overloaded
equality symbol
$\bigcup_{\sst\in\SO}\{\eqto{\sigd}\sst\prosep\sst\sarrow\boolbi\}$
together with the symbols $\{\truebi{\sigd}\sarrow\boolbi \sigsep
\falsebi{\sigd}\sarrow\boolbi \sigsep
\land{\sigd}\boolbi\prosep\boolbi\sarrow\boolbi\}$ are added to the
extended signature.
The rules generated for equality are then restricted to those whose left-hand
side and right-hand side have the same sort. Given a many-sorted
signature $(\SO,\FF)$, this set of rules consists of
$\bigcup_{\sst\in\SO}\bigcup_{f\in\FFs{\sst}^n}
\{~f(\Xxs{1}{\sst_1},\ldots,\Xxs{n}{\sst_n})\eqto f(\Yys{1}{\sst_1},\ldots,\Yys{n}{\sst_n})
\ra \Xxs{1}{\sst_1}\eqto\Yys{1}{\sst_1}\land\ldots\land\Xxs{n}{\sst_n}\eqto\Yys{n}{\sst_n}\land\truebi~\}$
$\bigcup_{\sst\in\SO}\bigcup_{f\in\FFs{\sst}^n}\bigcup_{g\neq f\in\FFs{\sst}^n}
\{~f(\Xxs{1}{\sst_1},\ldots,\Xxs{n}{\sst_n})\eqto g(\Yys{1}{\sst_1},\ldots,\Yys{m}{\sst_m})
\ra \falsebi~\}$.
For the translation $\rawTRm$
in Section~\ref{se:metaencodings}, the expansion is performed as in the unsorted
case, except that the expansion process generates meta-level representations
directly.

\subsection{Examples }
\label{sec:examples}
%
The first example we consider in what follows comes from an optimizer for multi-core architectures~\cite{Yuki2013}, a
project where abstract syntax trees are manipulated and
transformations are expressed using rewrite rules and strategies, and
consists of two rewrite rules identified as patterns occurring often
in various forms in the project.  First, the rewrite rule $g(f(x)) \ra
f(g(x))$ corresponds to the search for an operator~$g$ (which can have
more than one parameter in the general case) which is pushed down
under another operator~$f$ (again, this operator may have more than
one parameter).  This rule is important since the corresponding
(innermost) reduction of a term of the form
$\tgf=\overbrace{g(\underbrace{f(\cdots(f}_n(g(\underbrace{f(\cdots(f}_n(g(\underbrace{f(\cdots(f}_n(g}^{m}(a)))\cdots)$,
with, for example, $n=10$ and $m=18$ occurrences of~$g$, involves 
 $O(n^2 m^2)$
computations and could be a performance bottleneck. 

Second, the
rewrite rule $h(x) \ra g(h(x))$ corresponds to wrapping some parts of
a program by some special constructs, like \texttt{try/catch} for
example, and it is interesting since its uncontrolled application is
obviously non-terminating.

At present, strategy definitions given as input to {\strategyanalyser} are
written in a simple functional style.
\begin{exa}
\label{ex:tgf}
The syntax allowing the definition of the above rewrite rules and
possible corresponding strategies could be defined as follows:
\begin{small}
\begin{alltt}
abstract syntax
  T = a() | b() | f(T) | g(T) | h(T)
strategies
  gfx()       = [ g(f(x)) -> f(g(x)) ]
  hx()        = [ h(x) -> g(h(x)) ]
  obu(t)      = mu x.(one(x) <+ t)              # obu stands for OnceBottomUp
  bu(t)       = mu x.(all(x) ; (t <+ Identity)) # bu stands for BottomUp
  repeat(s)   = mu y.((s ; y) <+ Identity)      # naive definition of innermost
  mainStrat() = repeat(obu(gfx()))              # strategy to compile
\end{alltt}
\end{small}
\end{exa}

\noindent
\smallskip As a second example, we consider a strategy involving
rewrite rules which are either non left-linear or non right-linear and
which are non-terminating if their application is not guided by a
strategy.
\begin{exa}
\label{ex:distfact}
We consider the following rewrite rules which implement the
distributivity and factorization of symbolic expressions composed
of~$\Plus$ and~$\Mult$ and their application under a specific
strategy:
\begin{small}
\begin{alltt}
abstract syntax
  T = Plus(T,T) | Mult(T,T) | Val(V)
  V = a() | b()
strategies
  dist()       = [ Mult(x,Plus(y,z)) -> Plus(Mult(x,y),Mult(x,z)) ]
  fact()       = [ Plus(Mult(x,y),Mult(x,z)) -> Mult(x,Plus(y,z)) ]
  innermost(s) = mu x.(all(x) ; ((s ; x) <+ Identity))
  mainStrat()  = innermost(dist()) ; innermost(fact())
\end{alltt}
\end{small}
\end{exa}

As a third example, we consider a larger program inspired by
the {\tom} compiler itself, where the signature is composed of seven
sorts and contains a significant number of constructors.

\begin{exa}
\label{ex:refactor}
We consider two rewrite rules.  The purpose of the first one,
\texttt{compile}, is to identify \texttt{Match} constructs and to
replace them by instructions of the form \texttt{If}, \texttt{Assign},
\texttt{WhileDo}, \ldots, which implement the matching algorithm. In
the {\tom} compiler, these instructions are transformed by the backend
into executable code written in {\java} or {\C} for instance.  The
second rule identifies occurrences of variables and produces a term
which contains both the new variable and the initial one (this rule
represents the first stage of a refactoring process).

\begin{small}
\begin{alltt}
abstract syntax # refactor example
  CodeList = NilCode() | ConsCode(Code,CodeList)
  Code = Match(TermList) | Assign(Name,Exp) | If(Exp,Code,Code) | WhileDo(Exp,Code) 
       | Nop() | ...
  Exp  = Or(Exp,Exp) | And(Exp,Exp) | IsFsym(Name,Term) | EqualTerm(Term,Term) 
       | TrueTL() | FalseTL() | ...
  TermList = ConsTerm(Term,TermList) | NilTerm()
  Term = VarTerm(Name) | ApplTerm(Name,TermList) | RenamedTerm(Term,Term)
  Nat = Z() | S(Nat)
  Name = Name(Nat)
strategies
  compile() = [ Match(l) -> <...Code...> ]
  rename() = [ VarTerm(Name(n)) -> RenamedTerm(VarTerm(Name(S(n))), VarTerm(Name(n))) ]
  td(s) = mu x.((s <+ Identity) ; all(x))      # td stands for TopDown 
  tdstoponsucces(s) = mu x.(s <+ all(x))
  mainStrat() = tdt(compile()) ; tdstoponsucces(rename())
\end{alltt}
\end{small}
The right-hand side of the rule \texttt{rename} contains the left-hand
side of the rule and thus a top-down strategy on this rule would not
be terminating.  Furthermore, it would rename the second argument of
\texttt{RenamedTerm} that we want to keep unchanged.  The strategy we
consider (\texttt{tdstoponsucces}) is interesting because it searches
for a \texttt{VarName} constructor in a top-down way, but performs the
replacement only once and does not continue the search into sub-terms
when a transformation is performed.  Intuitively, this strategy is
terminating, and we will see that termination tools are able to prove
it.
\end{exa}

We also consider a relatively large example containing 4 sorts,
13 constructors, and 20 rules.
\begin{exa}
\label{ex:rbtree}
The following represents the implementation of red-black-trees based
on~\cite{DBLP:journals/jfp/Okasaki99}, but expressed using rules and
strategies. 
\begin{small}
\begin{alltt}
abstract syntax # rbTree example
  Tree = E() | T(Color,Tree,Nat,Tree) | balance(Tree) | ins(Nat,Tree)
       | insAux(Nat,Tree,Cmp)
  Color = R() | B()
  Nat = Z() | S(Nat)
  Cmp = lt() | gt() | eq0() | cmp(Nat,Nat)
strategies
  b1() = [ balance(T(B(),T(R(),T(R(),a1,a2,a3),x,b),y,T(R(),c,z,d))) -> 
           T(R(),T(B(),T(R(),a1,a2,a3),x,b),y,T(B(),c,z,d)) ]
  b2() = [ balance(T(B(),T(R(),a,x,T(R(),b1,b2,b3)),y,T(R(),c,z,d))) -> 
           T(R(),T(B(),a,x,T(R(),b1,b2,b3)),y,T(B(),c,z,d)) ]
  ... # rules b3()...b8()
  b9() = [ balance(t) -> t ] # no balancing necessary
  ... # rules i1()...i5() and c1()...c4()
  mainStrat() = repeat(obu(b1() <+ b2() <+ b3() <+ ...))
\end{alltt}
\end{small}
\end{exa}
The complexity here comes also from the presence of a
constructor~\texttt{T} of arity~4, whose negation (${\ap}T(\ldots)$)
generates a large list of patterns to capture the cases where the rule
cannot be applied. 
The anti-term $\ap
balance(T(B(),T(R(),T(R(),a1,a2,a3),x,b),y,T(R(),c,z,d)))$ is expanded
into $108$ patterns.

\subsection{Generation of TRS for termination analysis}
%
When run with the flag \texttt{-aprove}, the {\strategyanalyser} tool generates a
TRS in {\aprove}/{\TTT} syntax which can be analyzed by any tool
accepting this syntax.
In this case, aliases and anti-terms are always completely expanded
leading generally to a significant number of plain rewrite rules.

The tool can be configured to generate many-sorted TRS or to generate
the meta-level representation of the TRS.  The number of generated
rules for a strategy could thus vary a lot.  In
Table~\ref{fig:benchdef}, we give for each example the number of
generated rules in the (U)nsorted case, in the many-(S)orted case, and
in the (M)eta-level case. The last column indicates whether the
termination of the generated TRS has been proven or disproven by
{\aprove}.

\begin{table}[hbt]
\begin{center}
{\small
\begin{tabular}{|l|l|c|c|c|c|}
  \hline
  \tvi Name & Strategy & U & S & M & \aprove\\
  \hline\hline
  \reps(\dist)            
  & $\mu X\dotsym((\dist\seqsym X) \choicesym {\id})$&
  49 &    57 &  25& \cmark\\
  \reps(\factorial)       
  & $\mu X\dotsym(({\factorial}\seqsym X) \choicesym {\id})$&
  84 &    78 &  60& \cmark\\
  $\reps(\dist\seqsym\factorial)$ 
  & $\mu X\dotsym((({\dist}\seqsym {\factorial})\seqsym X) \choicesym {\id})$&
  110&    107&  77& \xmark\\
  \td(\dist)
  & $\mu X\dotsym(({\dist}\choicesym\id)\seqsym\all(X))$&
  97 &    68 &  35& \cmark\\
  \obu(\factorial)
  & $\mu X\dotsym(\one(X)\choicesym{\factorial})$&
  102&    83 &  70& \cmark\\
  \reps(\obu(\factorial))
  & $\mu X\dotsym(({\obu(\factorial)}\seqsym X) \choicesym\id)$&
  138&   125 &  82& \cmark\\
  $\mathsf{factorize}$
  & $\mu X\dotsym(\all(X)\seqsym (({\factorial}\seqsym\all(X))\choicesym\id))$&
  162&    124&  80& \cmark\\
  $\mathsf{simplify}$
  & $\mathsf{\td(\dist)}\seqsym\mathsf{factorize}$&
  272&    206&  110& \cmark\\
  {\im(\dist)}
  & $\mu X\dotsym(\all(X)\seqsym(({\dist}\seqsym X)\choicesym\id))$&
  127&    103&  45& \cmark\\
  {\im(\factorial)}
  & $\mu X\dotsym(\all(X)\seqsym(({\factorial}\seqsym X)\choicesym\id))$&
  162&    124&  80& \cmark\\
  {\reps(\td(\dist))}
  & $\mu X\dotsym((\mathsf{\td(\dist)}\seqsym X) \choicesym\id)$&
  133&    110&  47& \xmark\\
\hline
\hline
  {\bu(\rf)}
  & $\mu X\dotsym(\all(X)\seqsym({\rf}\choicesym\id))$&
  51 &     51 &  31& \cmark\\
  {\td(\rf)}
  & $\mu X\dotsym(({\rf}\choicesym\id)\seqsym\all(X))$&
  51 &    51 &  31& \xmark\\
  {\reps(\obu(\rgf))}
  & $\mu X\dotsym(({\obu(\rgf)}\seqsym X) \choicesym\id)$&
  91 &    91 &  47& \cmark\\
  {\im(\rgf)}
  & $\mu X\dotsym(\all(X)\seqsym(({\rgf}\seqsym  X)\choicesym\id))$&
  85 &    85 &  45& \cmark\\
  $\mathsf{propagate}$
  & $\mu X\dotsym({\rgf}\seqsym(\all(X)\choicesym\id))$&
  73 &    73 &  41& \cmark\\
  \bup
  & $\mu X\dotsym(\all(X)\seqsym(\propagate\choicesym\id))$&
  127&    127&  58& \cmark\\
  \buptwo
  & \multicolumn{1}{c|}{-- \bup with     $\texttt{f},\texttt{g},\texttt{h}$ of  arity $2$ --} &
  991&    378&  66& \cmark\\
\hline
\hline
  $\mathsf{refactor}$
  & $\td(\compile) \seqsym \tdsos(\rename)$&
  63065 &  2350&  145& $\emptyset\quad$\cmark$\quad$\cmark\\
  $\mathsf{rbTree}$
  &$\mu X\dotsym(({\obu(\mathsf{b1}\choicesym\mathsf{b2}\choicesym\cdots)}\seqsym X) \choicesym {\id})$&
  1956&  1449&  1260& $\emptyset$\\
\hline
\end{tabular}
} 
\end{center}
\caption{\label{fig:benchdef} Termination analysis: the columns U, S, and M indicate the number
  of plain rewrite rules generated for the strategy, respectively in the
  (U)nsorted case, many-(S)orted case, and (M)eta-level case.
  The column~{\aprove} indicates (for the U, S and M cases) whether the termination of the rules has been
  proven (\cmark) or disproven (\xmark) by {\aprove}; $\emptyset$ is used when
  {\aprove} gives no information. A unique symbol in the row  indicates that the results are the same for the three cases.}
\end{table}

In practice, {\aprove} is able to handle relatively big sets of rules
and, for example, the termination of the strategy {\reps(\obu(\rgf))},
which is translated into $91$ rules, is proven in approximately $10$~s
(using the web interface).  Similarly, for the example {\buptwo} which
corresponds to an extension of the example {\bup} where symbols
\texttt{f}, \texttt{g}, and \texttt{h} become binary and the rule
{\rgf} is replaced by $g(f(x,y),z) \ra f(g(x,z),y)$, we generate $991$
rules and the proof can be done in less then $80$~s. For this example,
when considering the 378 many-sorted rules, the proof can be done in
less than $15$~s and in the meta-level case, consisting of 66 rules,
the proof can be done in approximately $12$~s.
The size of the (left-hand and right-hand sides of the) rules
seems to be an important factor since the termination for the $\mathsf{rbTree}$
example consisting of roughly $1000$ rules cannot be (dis)proven
while the $2350$ rules of the many-sorted encoding of
$\mathsf{refactor}$ can be handled in $160$~s. 

The termination of some strategies like, for example,
{\reps(\obu(\rgf))} might look pretty easy to show for an expert, but
termination is less obvious for more complex strategies like, for
example, {\bup}, which is a specialized version of
{\reps(\obu(\rgf))}, or {\rbu(\factorial)}, which is a variant of
{\bu(\factorial)}.

The approach was effective not only in proving termination of some
strategies, but also in disproving it when necessary. Once again this
might look obvious for some strategies like, for example, {\td(\rf)},
which involves a non-terminating rewrite rule, but it is less clear
for strategies combining terminating rewrite rules or strategies like,
\eg, \reps(\dist$\seqsym$\factorial).

\subsection{Generation of executable TRS}
\label{sec:genTRS}
%
When run with the flag \texttt{-tom}, the {\strategyanalyser} tool generates a TRS
in {\tom} syntax which can be subsequently compiled into {\java} code and executed. 

By default, {\tom} executes a plain TRS with a built-in leftmost-innermost strategy
encoded using function calls.  But {\tom} can also execute a rule controlled by
a user-defined strategy. In that case, the user-defined strategy is encoded into
{\java} objects and is evaluated using a library written in {\java}. This
library provides several implementations where the notion of failure can be
encoded by a {\java} exception or by a special value to provide an
exception-free implementation.

\begin{table}[hbt]
\begin{center}
\begin{tabular}{|l|c|c|c|c|}
  \hline
\tvi Name  & TRS & Meta TRS & {\tom} & {\tom}*\\
  \hline\hline
  \reps(\dist) 
  &  $<5$ & $<5$ & $<5$ & $<5$\\

  \reps(\factorial) &
  $<5$ & $<5$ & 13  & $<5$\\

  \obu(\factorial) &
  $<5$ & $<5$ & $<5$ & $<5$\\
  
  \reps(\obu(\factorial)) &
  190 & 323 & 2460 &  120\\

  {\im(\dist)} &
  332 & 347 & 650&  230\\

  {\im(\factorial)} &
  310& 472   & 308&  149\\

  \hline\hline
{\bu(\rf)} &
  $<5$ & $<5$ & $<5$ & $<5$\\

{\reps(\obu(\rgf))} &
  400 & 780 & 6300 & 414\\

{\im(\rgf)} &
  553 & 433 & 4180 & 365\\

{\propagate} &
  $<5$ & $<5$ & $<5$ & $<5$\\

\bup &
  49 & 108 & 46 & 42\\
\hline
\hline
$\mathsf{refactor}$ &
  - & 570 & 220 & 115\\
$\mathsf{rbTree}$ &
  1480 & 2200 & 2070 & 840\\

\hline
\end{tabular}
\end{center}

\caption{\label{fig:bench}
Benchmarks: the column {TRS} indicates the execution time in milliseconds
for the generated TRS compiled using {\tom} ({\ie} using a built-in leftmost-innermost strategy),
the column {Meta TRS}
indicates the execution time for a meta-level TRS compiled using {\tom}, the
column {\tom} indicates the execution time of the same strategy written
  directly in {\tom}, using {\java} exception-based implementation, and the
  column {\tom}* indicates the execution time of the {\tom} {\java}
exception-free implementation.}
\end{table}

It is interesting to see how the varying number of generated rules for a
strategy impacts the efficiency of the execution of such a system.
If we execute a {\tom}+{\java} program corresponding to the
\texttt{\reps(\obu(\rgf))} strategy with a classic built-in implementation where
strategy failure is implemented by a {\java} exception, the normalization of the
term $\tgf$ takes $6.3$~s\footnote{on a MacPro 3GHz} (Table~\ref{fig:bench},
column {\tom}).  With an alternative built-in implementation which uses a
special value and does not throw {\java} exceptions, the computation time
decreases to $0.41$~s (Table~\ref{fig:bench}, column {\tom*}).
The strategy \texttt{\reps(\obu(\rgf))} is translated into an
executable TRS containing $91$ {\tom} plain rewrite rules and the
normalization takes in this case $0.4$~s!
When generating a meta-level TRS, the number of rules decrease to~47, but the
normalization takes~$0.78$~s.
When implementing the \texttt{\im(\rgf)} strategy natively in {\tom}~{\java}, using 
a built-in implementation with a special encoding of failure, the normalization takes~$0.36$~s.
The same strategy is translated into an executable TRS which contains $85$
plain rewrite rules in the unsorted case and $45$ rules in the meta-level case.
The first TRS normalizes the term $\tgf$ in $0.55$~s, whereas it takes $0.43$~s
to the meta-level TRS to normalize the term.
This example is interesting because it shows that the meta-level approach
allows to considerably reduce the number of rules, and thus the size of the
generated code, without slowing down the execution time.
The performances are confirmed for the large examples; the $63065$
rules generated for $\mathsf{refactor}$ with the plain encoding
could not be compiled with {\tom}.

We observe that, although the number of generated rules could be
significant, the execution times of the resulting plain TRS are
comparable to those obtained with the native implementation of {\tom}
strategies. This might look somewhat surprising but can be explained
when we take a closer look to the way rewriting rules and strategies
are generally implemented:
\begin{itemize}
  \item the implementation of a TRS can be done in an efficient way
    since the complexity of syntactic pattern matching depends only
    on the size of the term to reduce and, thanks to many-to-one
    matching algorithms~\cite{GrafRTA91, maranget}, the number of
    rules has almost no impact.
  \item in {\tom},   each native strategy constructor is implemented
    by a {\java} class with a \texttt{visit} method which implements
    ({\ie} interprets) the semantics of the corresponding operator.
    The evaluation of a strategy \texttt{S} on a term \texttt{t} is
    implemented thus by a call \texttt{S.visit(t)} and an exception
    (\texttt{VisitFailure}) is thrown when the application of a
    strategy fails.
\end{itemize}
In the generated TRS, the memory allocation involved in the construction of
terms headed by the $\bot$ symbol encoding failure appears to be more efficient
than the costly {\java} exception handling. This is reflected by better
performances of the plain TRS implementation compared to the exception-based
native implementation (especially when the strategy involves a lot of
failures). We obtain performances with the generated TRS comparable to an
exception-free native implementation of strategies (as we can see with the
columns \textsf{TRS} and \textsf{Tom*} in Table~\ref{fig:bench}), because
efficient normalization techniques can be used for the plain TRS, since its
rewrite rules are not controlled by a programmable strategy.

\section{Conclusions and further work}
\label{se:conclusions}
We have proposed a translation of programmable strategies into plain
rewrite rules that we have proven sound and complete;
Figure~\ref{fig:contrib} 
summarizes the obtained
results. Well-established termination methods can be thus used to
(dis)prove the termination of the obtained TRS and we can deduce, as a
direct consequence, the property for the corresponding
strategy. Alternatively, the translation can be used as a strategy
compiler for languages which do not implement natively such
primitives.

\begin{figure}[!tp]
\begin{mathpar}
\,\,\,\,\,\,\,\,\,\,\,\,\,\,\,\,\xymatrix@R+10pt@C+30pt{
      & t \multieval{\TRctx{\Gamma}{S}} u
      \ar@{<=>}[dddl]_{\mbox{ Lemma~\ref{lem:meta}}} \ar@{<=>}[dddr]^{\mbox{ Lemma~\ref{th:equivalenceSorted}}} &  \\
      &  &  & \\
      & \step{\stratappctx{}{S}{t}}{u} 
      \ar@{<=>}[uu]_{\txt{~\\ ~\\  Sect.~\ref{se:encodings}\\ \cite{CirsteaLM15}}} 
      \ar@{<=>}[dl]^*+<3pt>{\mbox{ Sect.~\ref{se:metaencodings}}} 
      \ar@{<=>}[dr]_*+<3pt>{\mbox{ Sect.~\ref{se:typing}}} &  \\
      \trmeta{t} \multieval{\TRm{S}} \trmeta{u} & ~~  &  t:s \multieval{\TRctxS{\Gamma}{S}} u:s  \\
    }
\end{mathpar}
\caption{\label{fig:contrib} Relationships between the different encodings of
  this paper
}
\end{figure}

The translation has been also adapted to cope with many-sorted
signatures, and although the size of the obtained encodings could be
smaller than in the unsorted case, it still depends strongly on the
underlying signature. We have proposed a meta-level representation of
the terms and a corresponding translation which produces encodings
whose size depends to a lesser extent on the signature and are significantly
smaller than the ones obtained with the (un)sorted translation.

The translation has been implemented in {\tom} and can generate, for
the moment, plain TRS using either a {\tom} or an {\aprove}/{\TTT}
syntax.
We have experimented with classic strategies and {\aprove} and {\TTT}
have been able to (dis)prove the termination even when the number of
generated rules was significant. The performances for the generated
executable TRS are comparable to the ones of the {\tom} built-in
(exception-free) strategies.

The framework can be of course improved.  When termination is
disproven and a counterexample can be exhibited, it is interesting to
reproduce the corresponding infinite reductions in terms of strategy
derivations.  Since the TRS reductions corresponding to distinct
(sub-)strategy derivations are not interleaved, we think that, when
the infinite reduction starts with a term headed by the symbol
encoding the strategy, the back-translation of the counterexample
provided by the termination tools can be automatized.
When the counterexample concerns another symbol than the one encoding
the strategy we could try to rebuild a complete infinite reduction by
a backward application of the rules in the encoding until an
appropriate term (headed by the symbol encoding the strategy) is
found; although such an approach wouldn't work in all the cases it
could give valuable warnings concerning the design of the strategy
under investigation.

As far as the executable TRS is concerned, we intend to develop new
backends allowing the integration of programmable strategies in other
languages than {\tom}.

\smallbreak\noindent\textit{Acknowledgments.}
We would like to thank the anonymous referees for their valuable comments and
suggestions that led to a substantial improvement of the paper.


\begin{thebibliography}{10}

\bibitem{AlarconGL10}
Beatriz Alarc{\'o}n, Ra{\'u}l Guti{\'e}rrez, and Salvador Lucas.
\newblock Context-sensitive dependency pairs.
\newblock {\em Information and Computation}, 208(8):922--968, 2010.

\bibitem{Arts2000133}
Thomas Arts and J{\"u}rgen Giesl.
\newblock Termination of term rewriting using dependency pairs.
\newblock {\em Theoretical Computer Science}, 236(1–2):133 -- 178, 2000.

\bibitem{BaaderN98}
Franz Baader and Tobias Nipkow.
\newblock {\em {Term Rewriting and All That.}}
\newblock Cambridge University Press, 1998.

\bibitem{BallandBKMR-RTA2007}
Emilie Balland, Paul Brauner, Radu Kopetz, Pierre-Etienne Moreau, and Antoine
  Reilles.
\newblock Tom: Piggybacking rewriting on java.
\newblock In {\em RTA '07}, volume 4533 of {\em LNCS}, pages 36--47.
  Springer-Verlag, 2007.

\bibitem{BallandMR-SPE2012}
Emilie Balland, Pierre-Etienne Moreau, and Antoine Reilles.
\newblock Effective strategic programming for {J}ava developers.
\newblock {\em Software: Practice and Experience}, 44(2):129--162, 2012.

\bibitem{ElanBKKMR-wrla98}
Peter Borovansk{\'y}, Claude Kirchner, H{\'e}l{\`e}ne Kirchner, Pierre-Etienne
  Moreau, and Christophe Ringeissen.
\newblock An overview of {ELAN}.
\newblock In {\em WRLA '98}, volume~15. ENTCS, 1998.

\bibitem{jsc2010}
Horatiu Cirstea, Claude Kirchner, Radu Kopetz, and Pierre-Etienne Moreau.
\newblock Anti-patterns for rule-based languages.
\newblock {\em Journal of Symbolic Computation}, 45(5):523 -- 550, 2010.
\newblock Symbolic Computation in Software Science.

\bibitem{CirsteaLM15}
Horatiu Cirstea, Sergue{\"{\i}} Lenglet, and Pierre{-}Etienne Moreau.
\newblock A faithful encoding of programmable strategies into term rewriting
  systems.
\newblock In {\em RTA '15}, volume~36 of {\em LIPIcs}, pages 74--88. Schloss
  Dagstuhl - Leibniz-Zentrum fuer Informatik, 2015.

\bibitem{Maude2:03}
Manuel Clavel, Francisco Dur{\'a}n, Steven Eker, Patrick Lincoln, Narciso
  Mart{\'\i}-Oliet, Jos{\'e} Meseguer, and Carolyn Talcott.
\newblock The {Maude} 2.0 {System}.
\newblock In {\em RTA '03}, volume 2706 of {\em LNCS}, pages 76--87.
  Springer-Verlag, 2003.

\bibitem{EndrullisH09}
J{\"o}rg Endrullis and Dimitri Hendriks.
\newblock From outermost to context-sensitive rewriting.
\newblock In {\em RTA '09}, volume 5595 of {\em LNCS}, pages 305--319.
  Springer, 2009.

\bibitem{Yuki2013}
Yuki~Tomofumi et~al.
\newblock Alphaz: A system for design space exploration in the polyhedral
  model.
\newblock {Technical Report CS-12-101}, Colorado State University, 2012.

\bibitem{maranget}
Fabrice~Le Fessant and Luc Maranget.
\newblock Optimizing pattern matching.
\newblock In {\em ICFP '01}, pages 26--37. ACM Press, 2001.

\bibitem{FissoreGK03}
Olivier Fissore, Isabelle Gnaedig, and H{\'e}l{\`e}ne Kirchner.
\newblock Simplification and termination of strategies in rule-based languages.
\newblock In {\em PPDP '03}, pages 124--135. {ACM}, 2003.

\bibitem{terminationDP-aprove2011}
Carsten Fuhs, J{\"u}rgen Giesl, Michael Parting, Peter Schneider-Kamp, and
  Stephan Swiderski.
\newblock Proving termination by dependency pairs and inductive theorem
  proving.
\newblock {\em Journal of Automated Reasoning}, 47(2):133--160, 2011.

\bibitem{GieslM02}
J{\"u}rgen Giesl and Aart Middeldorp.
\newblock Innermost termination of context-sensitive rewriting.
\newblock In {\em DLT '02}, volume 2450 of {\em LNCS}, pages 231--244.
  Springer, 2002.

\bibitem{GieslM04}
J{\"u}rgen Giesl and Aart Middeldorp.
\newblock Transformation techniques for context-sensitive rewrite systems.
\newblock {\em Journal of Functional Programming}, 14(4):379--427, 2004.

\bibitem{GieslRSST11}
J{\"u}rgen Giesl, Matthias Raffelsieper, Peter Schneider{-}Kamp, Stephan
  Swiderski, and Ren{\'{e}} Thiemann.
\newblock Automated termination proofs for haskell by term rewriting.
\newblock {\em {ACM} Transactions on Programming Languages and Systems},
  33(2):7, 2011.

\bibitem{Giesl06mechanizingand}
J{\"u}rgen Giesl, Ren{\'e} Thiemann, and Stephan Falke.
\newblock Mechanizing and improving dependency pairs.
\newblock {\em Journal of Automated Reasoning}, 37:2006, 2006.

\bibitem{GirardLafontTaylor89}
Jean-Yves Girard, Yves Lafont, and Paul Taylor.
\newblock {\em Proofs and Types}, volume~7 of {\em Cambridge Tracts in
  Theoretical Computer Science}.
\newblock Cambridge University Press, 1989.

\bibitem{GnaedigK09}
Isabelle Gnaedig and H{\'e}l{\`e}ne Kirchner.
\newblock Termination of rewriting under strategies.
\newblock {\em {ACM} Transactions on Computational Logic}, 10(2), 2009.

\bibitem{GrafRTA91}
Albert Gr{\"a}f.
\newblock Left-to-right tree pattern matching.
\newblock In {\em RTA '91}, volume 488 of {\em LNCS}, pages 323--334.
  Springer-Verlag, April 1991.

\bibitem{Hirokawa2005172}
Nao Hirokawa and Aart Middeldorp.
\newblock Automating the dependency pair method.
\newblock {\em Information and Computation}, 199(1–2):172 -- 199, 2005.

\bibitem{JacobsCLTT}
B.~Jacobs.
\newblock {\em Categorical Logic and Type Theory}.
\newblock Number 141 in Studies in Logic and the Foundations of Mathematics.
  North Holland, Amsterdam, 1999.

\bibitem{JouannaudKirchner-rob91}
Jean-Pierre Jouannaud and Claude Kirchner.
\newblock Solving equations in abstract algebras: a rule-based survey of
  unification.
\newblock In {\em Computational {L}ogic. {E}ssays in honor of {A}lan
  {R}obinson}, chapter~8, pages 257--321. The MIT press, 1991.

\bibitem{JouannaudKirchnerSIAM86}
Jean-Pierre Jouannaud and H{\'e}l{\`e}ne Kirchner.
\newblock Completion of a set of rules modulo a set of equations.
\newblock {\em {SIAM} Journal of Computing}, 15(4):1155--1194, 1986.

\bibitem{KaiserL09}
Markus Kaiser and Ralf L{\"a}mmel.
\newblock An {I}sabelle/{HOL}-based model of stratego-like traversal
  strategies.
\newblock In {\em PPDP '09}, pages 93--104. {ACM}, 2009.

\bibitem{ttt-2009}
Martin Korp, Christian Sternagel, Harald Zankl, and Aart Middeldorp.
\newblock Tyrolean {T}ermination {T}ool 2.
\newblock In {\em RTA '09}, volume 5595 of {\em LNCS}, pages 295--304.
  Springer-Verlag, 2009.

\bibitem{LammelTK13}
Ralf L{\"a}mmel, Simon~J. Thompson, and Markus Kaiser.
\newblock Programming errors in traversal programs over structured data.
\newblock {\em Science of Computer Programing}, 78(10):1770--1808, 2013.

\bibitem{Manzano:1993:IML}
Mar\'{\i}a Manzano.
\newblock Many-sorted logic and its applications.
\newblock chapter Introduction to Many-sorted Logic, pages 3--86. John Wiley \&
  Sons, Inc., New York, NY, USA, 1993.

\bibitem{meseguer2004}
Jos{\'e} Meseguer and Christiano Braga.
\newblock Modular rewriting semantics of programming languages.
\newblock In {\em Algebraic Methodology and Software Technology}, volume 3116
  of {\em LNCS}, pages 364--378. Springer Berlin Heidelberg, 2004.

\bibitem{MoreauRV-2003}
Pierre-Etienne Moreau, Christophe Ringeissen, and Marian Vittek.
\newblock {A Pattern Matching Compiler for Multiple Target Languages}.
\newblock In {\em CC '03}, volume 2622 of {\em LNCS}, pages 61--76.
  Springer-Verlag, 2003.

\bibitem{Ohlebusch2002}
Enno Ohlebusch.
\newblock {\em Advanced Topics in Term Rewriting}.
\newblock Springer-Verlag, 2002.

\bibitem{DBLP:journals/jfp/Okasaki99}
Chris Okasaki.
\newblock Red-black trees in a functional setting.
\newblock {\em Journal of Functional Programming}, 9(4):471--477, 1999.

\bibitem{RaffelsieperZ09}
Matthias Raffelsieper and Hans Zantema.
\newblock A transformational approach to prove outermost termination
  automatically.
\newblock {\em ENTCS}, 237:3--21, 2009.

\bibitem{RosuS10}
Grigore Rosu and Traian{-}Florin Serbanuta.
\newblock An overview of the {K} semantic framework.
\newblock {\em Journal of Logic and Algebraic Programming}, 79(6):397--434,
  2010.

\bibitem{Terese2002}
Terese.
\newblock {\em Term Rewriting Systems}.
\newblock Cambridge University Press, 2003.
\newblock M. Bezem, J. W. Klop and R. de Vrijer, eds.

\bibitem{Thiemann09}
Ren{\'e} Thiemann.
\newblock From outermost termination to innermost termination.
\newblock In {\em SOFSEM '09}, volume 5404 of {\em LNCS}, pages 533--545.
  Springer, 2009.

\bibitem{ThiemannM08}
Ren{\'e} Thiemann and Aart Middeldorp.
\newblock Innermost termination of rewrite systems by labeling.
\newblock {\em ENTCS}, 204:3--19, 2008.

\bibitem{Vis01.rta}
Eelco Visser.
\newblock Stratego: {A} language for program transformation based on rewriting
  strategies. {S}ystem description of {Stratego} 0.5.
\newblock In {\em RTA '01}, volume 2051 of {\em LNCS}, pages 357--361.
  Springer-Verlag, 2001.

\bibitem{VISSER2005831}
Eelco Visser.
\newblock A survey of strategies in rule-based program transformation systems.
\newblock {\em Journal of Symbolic Computation}, 40(1):831 -- 873, 2005.

\bibitem{visser-icfp98}
Eelco Visser, Zine-el-Abidine Benaissa, and Andrew Tolmach.
\newblock Building program optimizers with rewriting strategies.
\newblock In {\em ICFP '98}, pages 13--26. ACM Press, 1998.

\bibitem{Wolfram2003}
Stephen Wolfram.
\newblock {\em The Mathematica Book}.
\newblock Wolfram Media, Incorporated, 5 edition, 2003.

\end{thebibliography}

\newpage
\section{Proofs for the Regular Translation}

\begin{lem}[Propagation lemma]
\label{lemma:failurePropagation}
Let $t\in \TF$, $S$, and $\Gamma$ such that
$\varS S \subseteq \domS \Gamma$. We have
$\trsred{\TRgGG \cup \TRctx{{\varphi}}{S}}{\vpS(\bot(t))}{\bot(t)}$, and if
$\trsred{\TRgGG \cup \TRctx{{\varphi}}{S}}{\vpS(\bot(t))}{u}$, then
$u = \bot(t)$.
\end{lem}
\begin{proof}
  For all $S \neq X$, the only rule in $\TRgGG \cup \TRctx{}{S}$ that can
  rewrite $\vpS(\bot (t))$ is $\vpS(\bot(t)) \ra \bot(t)$. For $S = X$, because
  $X \in \domS \Gamma$, $\TRg{}{\Gamma}$ contains the rule $\vpn{X}(\bot(t)) \ra
  \bot(t)$, and it is also the only rule that can rewrite $\vpn{X}(\bot (t))$.
\end{proof}

\begin{lem}[Rigid]
\label{lemma:rigidTF}
Given a term $t\in \TF$, a strategy $S$, and a context $\Gamma$, the term $t$ is
in normal form w.r.t. $\TRgGG \cup \TRctx{{\varphi}}{S}$. If $\trsred{\TRgGG
  \cup \TRctx{{\varphi}}{S'}}{\vpS(t)}{u}$, then $S=X$ and $X \in \domS \Gamma$,
or $\TRctx{}{S} \subseteq \TRgGG \cup \TRctx{{\varphi}}{S'}$.
\end{lem}
\begin{proof}
A term $t \in \TF$ does not contain any $\varphi$ symbols, and all the rules in
$\TRgGG \cup \TRctx{}{S}$ assumes a $\varphi$ symbol at the root, hence the
first result holds. 

Because of this result, if $\trsred{\TRgGG \cup
  \TRctx{{\varphi}}{S'}}{\vpS(t)}{u}$, then $\TRgGG \cup \TRctx{{\varphi}}{S'}$
contains a rule of the form $\vpn{S}(\any) \ra \ldots$. If $S \neq X$, then it
is possible only if $\Gamma$ or $S'$ contains $S$, and then it is easy to prove
that $\TRctx{}{S} \subseteq \TRgGG \cup \TRctx{{\varphi}}{S'}$.

If $S=X$, then only the translation of contexts or recursive strategies generate
rules than can rewrite a term of the form $\vpn{X}(t)$. But according to the
Barendregt convention, any recursive strategy in $\Gamma$ or $S'$ will be of the
form $\mu Y.S''$, with $Y \neq X$. The only remaining possibility is
$\trsred{\TRgGG}{\vpn{X}(t)}{u}$, which is possible only if $\Gamma$ binds $X$.
\end{proof}

\begin{lem}
  \label{lemma:noPhi}
  If $\trsred{\TRgGG \cup \TRctx{{\varphi}}{S}}{\vpS(t)}{u}$, then $t \in \TF$
  or $t = \bot(\any)$.
\end{lem}

\begin{proof}
  Immediate by definition of the translation.
\end{proof}

\begin{lem}[Initial term as failure]
\label{lemma:sameFailure}
Given a term $t\in \TF$, if $\trsred{\TRgGG \cup
  \TRctx{{\varphi}}{S}}{\vpS(t)}{\bot(t')}$ then $t=t'$.
\end{lem}
\begin{proof}
  By induction on the length of the reduction $\trsred{\TRgGG \cup
    \TRctx{{\varphi}}{S}}{\varphi(t)}{\bot(t')}$.

  If $S = \id$, it is not possible to obtain $\bot(t')$ from $\vpn{\id}(t)$,
  because $t \in \TF$.

  If $S = \fail$, then because $t \in \TF$, the only rule that can be applied
  to $\vpn{\fail}(t)$ is $\vpn{\fail}(x \at \ap \bot(\any)) \ra \bot(x)$, and we
  obtain $\bot(t)$, as wished. 

  If $S = l \ra r$, then the only rule that can be applied to $\vpn{l \ra r}(t)$
  (with $t \in \TF$) to produce $\bot(t')$ is $\vpn{l \ra r}(x \at \ap l) \ra
  \bot(x)$, and we obtain  $\bot(t)$, as required.

  Suppose $S = \seq{S_1}{S_2}$. Because $\TRctx{}{S_1}$ and $\TRctx{}{S_2}$
  cannot reduce $\vpn{\seq{S_1}{S_2}}(t)$ (by Lemma~\ref{lemma:rigidTF}), the
  first rule to be applied is $\varphi_{\seqsub{S_1}{S_2}}(x \at \ap \bot(\any) )
  \ra \phiseq( \vpt(\vpo(x)), x)$, and we obtain $\phiseq(\vpt(\vpo(t)), t)$. By
  Lemma~\ref{lemma:noPhi}, $\TRgGG \cup \TRctx{}{S_2}$ cannot reduce
  $\vpt(\vpo(t))$, and by Lemma~\ref{lemma:rigidTF}, only $\TRctx{}{S_1}$ (or
  $\TRgGG$ is $S_1 = X$) can reduce $\vpo(t)$. We have several possibilities. If
  $\trsred{\TRgGG \cup \TRctx{}{S_1}}{\vpo(t)}{\bot(u)}$, then by induction, $u
  = t$, and by Lemma~\ref{lemma:failurePropagation}, $\trsred{\TRgGG \cup
    \TRctx{}{S_2}}{\phiseq(\vpt(\bot(t)), t)}{\phiseq(\bot(t), t)}$. The only
  rule that can be applied to the latter term is $\phiseq(\bot(\any), x) \ra
  \bot(x)$, and we obtain $\bot(t)$, as wished. Otherwise, we have
  $\trsred{\TRgGG \cup \TRctx{}{S_1}}{\vpo(t)}{u}$. Because we obtain $\bot(t')$
  as a result, necessarily the rule $\phiseq(\bot(\any), x) \ra \bot(x)$ has
  been applied, which means that $\trsred{\TRgGG \cup
    \TRctx{}{S_2}}{\vpt(u)}{\bot(t'')}$ for some $t''$. By induction, $t''=u$,
  and consequently, $\trsred{\TRgGG \cup \TRctx{}{S_2}}{\phiseq(\vpt(u),
    t)}{\phiseq(\bot(u), t)}$. We then have $\trsred{\{\phiseq(\bot(\any), x)
    \ra \bot(x)\}}{\phiseq(\bot(u), t)}{\bot(t)}$, as required.

  Suppose $S = \choice{S_1}{S_2}$. Because $\TRctx{}{S_1}$ and $\TRctx{}{S_2}$
  cannot reduce $\vpn{\choicesub{S_1}{S_2}}(t)$ (by Lemma~\ref{lemma:rigidTF}), the
  first rule to be applied is $\varphi_{\choicesub{S_1}{S_2}}(x \at \ap \bot(\any)
  ) \ra \phich(\vpo(x))$, and we obtain $\phich(\vpo(t))$. By
  Lemma~\ref{lemma:rigidTF}, only $\TRctx{}{S_1}$ (or $\TRgGG$ is $S_1 = X$) can
  reduce $\vpo(t)$. We have two possibilities. If $\trsred{\TRgGG \cup
    \TRctx{}{S_1}}{\vpo(t)}{\bot(u)}$, then by induction $u = t$, and then
  $\trsred{\{\phich(\bot(x)) \ra \vpt(x)\}}{\phich(\bot(t))}{\vpt(t)}$. By
  Lemma~\ref{lemma:rigidTF}, only $\TRctx{}{S_2}$ (or $\TRgGG$ is $S_2 = X$) can
  reduce $\vpt(t)$, therefore we have $\trsred{\TRgGG \cup
    \TRctx{}{S_2}}{\vpt(t)}{\bot(t')}$. By induction, we have $t=t'$, and the
  result holds. Otherwise, we have $\trsred{\TRgGG \cup
    \TRctx{}{S_1}}{\vpo(t)}{u}$ with $u \in \TF$. But then, the only rule that
  can be applied to $\phich(u)$ is $\phich(x \at \ap \bot(\any)) \ra x$, and we
  obtain $u \in \TF$ as a result of the reduction of
  $\vpn{\choice{S_1}{S_2}}(t)$, which is in contradiction with the original
  hypothesis.

  Suppose $S = \mu X.S'$. Because $\TRctx{}{S'}$ cannot reduce $\vpn{\mu
    X.S'}(t)$ (by Lemma~\ref{lemma:rigidTF}), the first rule to be applied is
  $\varphi_{\mu X.S'}(Y \at \ap \bot(\any)) \ra \vpn{S'}(Y)$, and we obtain
  $\vpn{S'}(t)$. We therefore have $\trsred{\TRgGG \cup \TRctx{{\varphi}}{\mu
      X.S'}}{\vpn{S'}(t)}{\bot(t')}$ with less steps than the original
  reduction. Besides, only $\TRctx{}{S'}$ (or $\TRgGG$ if $S'=X$) can reduce
  $\vpn{S'}(t)$ (Lemma~\ref{lemma:rigidTF}), therefore we have in
  fact $\trsred{\TRgGG \cup \TRctx{{\varphi}}{S'}}{\vpn{S'}(t)}{\bot(t')}$. By
  induction, $t'=t$, and the result holds. 

  Suppose $S = X$. By Lemma~\ref{lemma:rigidTF}, only $\TRgGG$ can reduce
  $\vpn{X}(t)$, and assuming $X \colon S'$ belongs to $\Gamma$, the only rule
  that can be applied is $\varphi_{X}(Y \at \ap \bot(\any)) \ra \vpn{S'}(Y)$,
  which generates $\vpn{S'}(t)$. Only $\TRctx{}{S'}$ (or $\TRgGG$ if $S'=Z$) can
  reduce $\vpn{S'}(t)$ (Lemma~\ref{lemma:rigidTF}), therefore we have
  $\trsred{\TRgGG \cup \TRctx{{\varphi}}{S'}}{\vpn{S'}(t)}{\bot(t')}$ with less
  steps than the original reduction. By induction, $t'=t$, and the result holds.

  Suppose $S =\all(S')$. If $t$ is a constant $c$, then we cannot obtain
  $\bot(t')$. Suppose $t=f(t_1, \ldots, t_n)$.  Because $\TRctx{}{S'}$ cannot
  reduce $\vpn{\all(S')}(t)$ (by Lemma~\ref{lemma:rigidTF}), the first rule to
  be applied is $\varphi_{\all(S)}(f(x_1,\ldots,x_n)) \ra
  \psiF(\vpn{S'}(x_1),\ldots,\vpn{S'}(x_n),f(x_1,\ldots,x_n))$ to generate
  $\psiF(\vpn{S'}(t_1),\ldots,\vpn{S'}(t_n),f(t_1,\ldots,t_n))$. Note that
  according to Lemma~\ref{lemma:rigidTF}, the rules in $\TRgGG \cup
  \TRctx{}{\all(S')}$ cannot rewrite $f(t_1, \ldots, t_n)$ in
  $\psiF(\vpn{S'}(t_1),\ldots,\vpn{S'}(t_n),f(t_1,\ldots,t_n))$. From this term,
  the only possibility to obtain $\bot(t')$ is to apply one of the rules of the
  form $\psiF(\any, \ldots \bot(\any),\ldots,\any,z) \ra \bot(z)$, which
  generates $f(t_1, \ldots, t_n)$, as wished.

  Suppose $S = \one(S')$. If $t$ is a constant $c$, then the only rule that can
  be applied is $\vpn{\one(S')}(c) \ra \bot(c)$, hence the result holds. Suppose
  $t=f(t_1, \ldots, t_n)$.  Because $\TRctx{}{S'}$ cannot reduce
  $\vpn{\one(S')}(t)$ (by Lemma~\ref{lemma:rigidTF}), the first rule to be
  applied is $\varphi_{\one(S')}(f(x_1,\ldots,x_n)) \ra
  \psif{1}(\vpn{S'}(x_1),x_2,\ldots,x_n)$, which generates
  $\psif{1}(\vpn{S'}(t_1),t_2,\ldots,t_n)$. Only $\TRctx{}{S'}$ (or $\TRgGG$ if
  $S'=X$) can reduce $\vpn{S'}(t_1)$. Suppose $\trsred{\TRgGG \cup
    \TRctx{{\varphi}}{S'}}{\vpn{S'}(t_1)}{u}$ with $u \in \TF$. Then we can only
  apply the rule $\psif{1}(x_1\at\ap\bot(\any),x_2,\ldots,x_n) \ra
  f(x_1,\ldots,x_n)$ to $\psif{1}(u,t_2,\ldots,t_n)$, and we obtain $f(u, t_2,
  \ldots, t_n)$, in contradiction with the initial hypothesis. Therefore,
  $\trsred{\TRgGG \cup \TRctx{{\varphi}}{S'}}{\vpn{S'}(t_1)}{\bot(t'_1)}$, in
  less steps than the original reduction. By induction, $t'_1=t_1$. We can only
  apply the rule $ \psif{1}(\bot(x_{1}),x_2,\ldots,x_n) \ra
  \psif{2}(\bot(x_1),\vpS(x_2),\ldots,x_n)$ to
  $\psif{1}(\bot(t_1),t_2,\ldots,t_n)$, and we obtain
  $\psif{2}(\bot(t_1),\vpn{S'}(t_2),\ldots,t_n)$. With the same reasoning, we
  have $\trsred{\TRgGG \cup \TRctx{{\varphi}}{S'}}{\vpn{S'}(t_i)}{\bot(t_i)}$
  for all $i$, and we therefore have $\trsred{\TRgGG \cup
    \TRctx{{\varphi}}{\one(S')}}{\vpn{S'}(\psif{1}(\vpn{S'}(t_1),t_2,\ldots,t_n))}{\psif{n}(\bot(t_1),\ldots,\bot(t_n))}$. We
  can apply only the rule $\psif{n}(\bot(x_1),\ldots,\bot(x_n)) \ra
  \bot(f(x_1,\ldots,x_n))$ to the latter term, and we obtain $\bot(f(t_1, \ldots
  t_n))$, as required.
\end{proof}

\simulation*

\begin{proof}
By induction on the height of derivation tree and respectively of on
the length of the reduction.  
$$
\begin{array}{rcl}
  S & :=  &   \id \mid \fail \mid l \ra r \mid \seq{S_1}{S_2} \mid \choice{S_1}{S_2} \mid \one(s) \mid \all(s) \mid \mu X.{S'} \\
\end{array}
$$

\emph{Base case:} $S  :=  \id \mid \fail \mid l \ra r$

\begin{enumerate}
\item $S:=\id$

Independently of $\Gamma$, for all $t\in \TF$ we have 
$$\stratappctx{\Gamma}{\id}{t}\sred t$$
and 
$$\trsred{\{  \varphi_{\id}(x \at \ap \dummy) \ra x ~\} \cup
  \TRg{\Gamma}{\Gamma}}{\varphi_{\id}(t)}{t}$$ 

\item $S:=\fail$

Independently of $\Gamma$, for all $t\in \TF$ we have 
$$\stratappctx{\Gamma}{\fail}{t}\sred \failres$$
and 
$$\trsred{\{  \varphi_{\fail}(x \at \ap \dummy) \ra \bot(x) ~\} \cup \TRg{\Gamma}{\Gamma}}{\varphi_{\fail}(t)}{\bot(t)}$$

\item $S:=l \ra r$

  \begin{enumerate}
  \item $S \Rightarrow  \TRctx{\Gamma}{S} \cup \TRg{\Gamma}{\Gamma}$

    If $\stratappctx{\Gamma}{l \ra r}{t}\sred u$ then $\exists \sigma,
    \applysubs{\sigma}{\mathit{l}} = t$ and
    $\applysubs{\sigma}{\mathit{r}}=u$.  Then,
    $\applysubs{\sigma}{\mathit{\varphi_{l\ra r}(l)}} = \varphi_{l\ra r}(t)$ and thus
    $\trsred{\varphi_{l \ra r}(l) \ra r}{\varphi_{l\ra r}(t)}{u}$.

    If $\stratappctx{\Gamma}{l \ra r}{t}\sred \failres$ then
    $\not\exists \sigma, \applysubs{\sigma}{\mathit{l}} = t$ and thus
    $\not\exists \sigma, \applysubs{\sigma}{\mathit{\varphi_{l\ra
          r}(l)}} = \varphi_{l \ra r}(t)$. The rewrite rule $\varphi_{l\ra r}(l)
    \ra r$ cannot be applied to $\varphi_{l\ra r}(t)$ at the head
    position. Since $\varphi_{l\ra r}$ is a fresh symbol ($\not\exists
    p$ s.t. $t(p)=\varphi_{l\ra r}$) the rule cannot be applied to
    another position of $\varphi_{l\ra r}(t)$. Since $\not\exists
    \sigma, \applysubs{\sigma}{\mathit{\varphi_{l\ra r}(l)}} =
    \varphi_{l\ra r}(t)$ then $ \varphi_{l\ra r}(t)$ is in the
    semantics of $\varphi_{l\ra r}(\ap l)$ and thus the rule
    $\varphi_{l\ra r}(x \at \ap l) \ra \bot(x)$ can be applied and the
    result is $\bot(t)$.

  \item $\TRctx{\Gamma}{S} \cup \TRg{\Gamma}{\Gamma}  \Rightarrow  S$

    If the rule $\varphi_{l\ra r}(l) \ra r$ is used then, since
    $\varphi_{l\ra r}$ does not occur in $t$, it can be only applied
    at the head position of $\varphi_{l\ra r}(t)$ and thus
    $\exists\sigma,\applysubs{\sigma}{\mathit{\varphi_{l\ra
          r}(l)}}=\varphi_{l\ra r}(t)$ and
    $\applysubs{\sigma}{\mathit{r}}=u$. Since $u$ contains no
    $\varphi_{l\ra r}$ (the codomain of $\sigma$ contains no
    $\varphi_{l\ra r}$ because it does not occur in $t$)
    and the only rules in $\TRg{\Gamma}{\Gamma}$ concerning
    $\varphi_{l\ra r}$ could be at most the same as those in
    $\TRctx{\Gamma}{l \ra r}$
    then $u$ is in normal form w.r.t. $\TRctx{\Gamma}{l \ra r} \cup
    \TRg{\Gamma}{\Gamma}$. We have then that
    $\exists\sigma,\applysubs{\sigma}{\mathit{l}}=t$ and thus $\stratapp{l\ra
      r}{t}\sred u$. Since $\varphi_{l\ra r}$ is a fresh symbol then there is no
    other way to apply the rewrite rule $\varphi_{l\ra r}(l) \ra r$ to
    $\varphi_{l\ra r}(t)$.

    If the rule $\varphi_{l\ra r}(x \at \ap l) \ra \bot(x)$ is applied
    then it is applied at the top position ($\varphi_{l\ra r}$ fresh)
    and $\exists\sigma=\{t/x\},\applysubs{\sigma}{\varphi_{l\ra
        r}(x\at\ap l)}=\varphi_{l\ra r}(t)$ and
    $\not\!\!\exists\mu,\applysubs{\mu}{\mathit{l}}=\varphi_{l\ra
      r}(t)$. Consequently, $\stratappctx{\Gamma}{l \ra r}{t}\sred \failres$. We
    have $\trsred{\varphi_{l\ra r}(x\at\ap l)\ra\bot(x)}{t}{\bot(t)}$,
    and because $\varphi$ symbols do not occur in $t$, then $\bot(t)$
    is in normal form w.r.t. $\TRctx{\Gamma}{l \ra r} \cup \TRg{\Gamma}{\Gamma}$.
    Because $\bot$ is a fresh symbol then $\not\exists \sigma,
    \applysubs{\sigma}{\bot(x)} = t$ and the rule $\varphi_{l\ra
      r}(\bot(x)) \ra \bot(x)$ cannot be applied to $\varphi_{l\ra
      r}(t)$.

  \end{enumerate}
\end{enumerate}

\emph{Induction:} $S  :=   \seq{S_1}{S_2} \mid  \choice{S_1}{S_2} \mid \one(S') \mid \all(S') \mid \mu X.{S'} $

\begin{enumerate}
\item $S:=\seq{S_1}{S_2}$

  \begin{enumerate}
  \item $S \Rightarrow  \TRctx{\vpg}{S}$

    Since $\stratappctx{\Gamma}{\seq{S_1}{S_2}}{t}\sred u$ then
    $\stratappctx{\Gamma}{S_1}{t}\sred v$ and
    $\stratappctx{\Gamma}{S_2}{v}\sred u$ with a shorter derivation
    tree. By induction,
    $\trsred{\TRgGG\cup\TRctx{{\varphi_1}}{S_1}}{\vpo(t)}{v}$ and
    $\trsred{\TRgGG\cup\TRctx{{\varphi_2}}{S_2}}{\vpt(v)}{u}$. We have
    $\trsred{\varphi_{\seqsub{S_1}{S_2}}(x\at\ap\dummy)\ra\phiseq(\vpt(\vpo(x)),x)}{\varphi_{\seqsub{S_1
       }{S_2}}(t)}{\phiseq(\vpt(\vpo(t)),t)}$.
    By the induction hypothesis,
    $\trsred{\TRgGG\cup\TRctx{{\varphi_1}}{S_1}}{\phiseq(\vpt(\vpo(t)),t)}{\phiseq(\vpt(v),t)}$
    and
    $\trsred{\TRgGG\cup\TRctx{{\varphi_2}}{S_2}}{\phiseq(\vpt(v),t)}{\phiseq(u,t)}$. Finally,
    $\trsred{\phiseq(x\at\ap\bot(\any),\any)\ra x}{\phiseq(u,t)}{u}$.

    If $\stratappctx{\Gamma}{\seq{S_1}{S_2}}{t}\sred \failres$ then
    $\stratappctx{\Gamma}{S_1}{t}\sred\failres$, or
    $\stratappctx{\Gamma}{S_1}{t}\sred v$ and
    $\stratappctx{\Gamma}{S_2}{v}\sred\failres$.
    For the first case, by induction,
    $\trsred{\TRgGG\cup\TRctx{{\varphi_1}}{S_1}}{\vpo(t)}{\bot(t)}$.  We have
    $\trsred{\varphi_{\seqsub{S_1}{S_2}}(x\at\ap\dummy)\ra\phiseq(\vpt(\vpo(x)),x)}{\varphi_{\seqsub{S_1}{S_2}}(t)}{\phiseq(\vpt(\vpo(t)),t)}$
    and, by induction,
    $\trsred{\TRgGG\cup\TRctx{{\varphi_1}}{S_1}}{\phiseq(\vpt(\vpo(t)),t)}{\phiseq(\vpt(\bot(t)),t)}$.
    By Lemma~\ref{lemma:failurePropagation},
    $\trsred{\TRgGG\cup\TRctx{{\varphi_2}}{S_2}}{\phiseq(\vpt(\bot(t)),t)}{\phiseq(\bot(t),t)}$
    and
    $\trsred{\phiseq(\bot(\any),x)\ra\bot(x)}{\phiseq(\bot(t),t)}{\bot(t)}$.
    For the second case, by induction,
    $\trsred{\TRgGG\cup\TRctx{{\varphi_1}}{S_1}}{\vpo(t)}{v}$ and
    $\trsred{\TRgGG\cup\TRctx{{\varphi_2}}{S_2}}{\vpt(v)}{\bot(v)}$.  We have,
    $\trsred{\varphi_{\seqsub{S_1}{S_2}}(x\at\ap\dummy)\ra\phiseq(\vpt(\vpo(x)),x)}{\varphi_{\seqsub{S_1}{S_2}}(t)}{\phiseq(\vpt(\vpo(t)),t)}$.
    By induction,
    $\trsred{\TRgGG\cup\TRctx{{\varphi_1}}{S_1}}{\phiseq(\vpt(\vpo(t)),t)}{\phiseq(\vpt(v),t)}$
    and
    $\trsred{\TRgGG\cup\TRctx{{\varphi_2}}{S_2}}{\phiseq(\vpt(v),t)}{\phiseq(\bot(v),t)}$. Finally,
    $\trsred{\phiseq(\bot(\any),x)\ra\bot(x)}{\phiseq(\bot(v),t)}{\bot(t)}$.

  \item $\TRctx{\vpg}{S}  \Rightarrow  S$

    If
    $\trsred{\TRgGG\cup\TRctx{\varphi_{\seqsub{S_1}{S_2}}}{\seq{S_1}{S_2}}}{\varphi_{\seqsub{S_1}{S_2}}(t)}{u}$,
    since $t\in\TF$ contains no $\varphi$ symbols, the first reduction is
    necessarily
    $\trsred{\varphi_{\seqsub{S_1}{S_2}}(x\at\ap\dummy)\ra\phiseq(\vpt(\vpo(x)),x)}{\varphi_{\seqsub{S_1}{S_2}}(t)}{\phiseq(\vpt(\vpo(t)),t)}$. According
    to Lemma~\ref{lemma:rigidTF}, $\vpo(t)$ can only be reduced by $\TRgGG\cup
    \TRctx{}{S_1}$, which gives
    $\trsred{\TRgGG\cup\TRctx{{\varphi_1}}{S_1}}{\phiseq(\vpt(\vpo(t)),t)}{\phiseq(\vpt(v),t)}$
    for some $v$. If we had
    $\trsred{\TRgGG\cup\TRctx{{\varphi_1}}{S_1}}{\vpo(t)}{\bot(t)}$, the $\bot$
    would have been propagated (by Lemma~\ref{lemma:failurePropagation}) and the
    final result would have been $\bot(t)$ which contradicts the initial
    hypothesis. By induction (the latter reduction needs less steps than the
    initial one), $\stratappctx{\Gamma}{S_1}{t}\sred v$.  Next, the only
    possible reduction is
    $\trsred{\TRgGG\cup\TRctx{{\varphi_2}}{S_2}}{\phiseq(\vpt(v),t)}{\phiseq(u,t)}$
    and then $\trsred{\phiseq(x \at \ap \bot(\any), \any) \ra
      x}{\phiseq(u,t)}{u}$.  By induction, $\stratappctx{\Gamma}{S_2}{v}\sred u$
    and since $\stratappctx{\Gamma}{S_1}{t}\sred v$ we can conclude that
    $\stratappctx{\Gamma}{\seq{S_1}{S_2}}{t}\sred u$.

    If $\trsred{\TRgGG\cup\TRctx{\varphi_{\seqsub{S_1}{S_2}}}{\seq{S_1}{S_2}}}{\varphi_{\seqsub{S_1}{S_2}}(t)}{\bot(t)}$,
    since $\varphi_{\seqsub{S_1}{S_2}}(t)$ is in normal form
    w.r.t. $\TRgGG\cup\TRctx{{\vpo}}{S_1}$ and $\TRgGG\cup\TRctx{{\vpt}}{S_2}$
    (according to Lemma~\ref{lemma:rigidTF}), the first reduction is necessarily
    $\trsred{\varphi_{\seqsub{S_1}{S_2}}(x\at\ap\dummy)\ra\phiseq(\vpt(\vpo(x)),x)}{\varphi_{\seqsub{S_1}{S_2}}(t)}{\phiseq(\vpt(\vpo(t)),t)}$. The
    only way to obtain $\bot(t)$ as a result is to have a reduction
    $\trsred{\TRgGG\cup\TRctx{\varphi_{\seqsub{S_1}{S_2}}}{\seq{S_1}{S_2}}}{\vpt(\vpo(t))}{\bot(v)}$.
    Thus, either $\trsred{\TRgGG\cup\TRctx{\varphi_{\seqsub{S_1}{S_2}}}{\seq{S_1}{S_2}}}{\vpo(t)}{v'}$ and
    $\trsred{\TRgGG\cup\TRctx{\varphi_{\seqsub{S_1}{S_2}}}{\seq{S_1}{S_2}}}{\vpt(v')}{\bot(v)}$ or
    $\trsred{\TRgGG\cup\TRctx{\varphi_{\seqsub{S_1}{S_2}}}{\seq{S_1}{S_2}}}{\vpo(t)}{\bot(v)}$ and
    $\trsred{\TRgGG\cup\TRctx{\varphi_{\seqsub{S_1}{S_2}}}{\seq{S_1}{S_2}}}{\vpt(\bot(v))}{\bot(v)}$ (by
    Lemma~\ref{lemma:failurePropagation}).
    In the first case, we have
    $\trsred{\TRgGG\cup\TRctx{\vpo}{S_1}}{\vpo(t)}{v'}$ since only the rules in
    $\TRgGG \cup \TRctx{\varphi_{\seqsub{S_1}{S_2}}}{S_1}$ can reduce $\vpo(t)$
    and, by induction, $\stratappctx{\Gamma}{S_1}{t}\sred v'$. Similarly,
    $\trsred{\TRgGG\cup\TRctx{\vpt}{S_2}}{\vpt(v')}{\bot(v)}$ and, by
    Lemma~\ref{lemma:sameFailure}, $v'=v$. Thus, by induction,
    $\stratappctx{\Gamma}{S_2}{v'}\sred \failres$. Consequently,
    $\stratappctx{\Gamma}{\seq{S_1}{S_2}}{t}\sred \failres$.
    For the second case, as previously, we have that
    $\trsred{\TRgGG\cup\TRctx{\vpo}{S_1}}{\vpo(t)}{\bot(v)}$ and, by
    Lemma~\ref{lemma:sameFailure}, $v=t$. By induction, 
    $\stratappctx{\Gamma}{S_1}{t}\sred \failres$ and consequently,
    $\stratappctx{\Gamma}{\seq{S_1}{S_2}}{t}\sred \failres$.

  \end{enumerate}

\item $S:=\mu X.{S'}$

  \begin{enumerate}
  \item $S \Rightarrow  \TRctx{\vpg}{S}$

    Since $\stratappctx{\Gamma}{\mu X.{S'}}{t}\sred u$ then
    $\stratappctx{\Gamma;{X\colon S'}}{S'}{t}\sred u$
    with the latter having a shorter derivation tree and thus, by
    induction,
    $\trsred{\TRg{\Gamma;X\colon{S'}}{\Gamma;X\colon{S'}}\cup\TRctx{{\varphi_1}}{S'}}{\varphi_{S'}(t)}{u}$.
    The only possible reduction of $\varphi_{\mu X.S'}(t)$ \wrt\
    $\TRgGG\cup\TRctx{{\varphi_1}}{\mu X.S'}$ is obtained by
    applying $\varphi_{\mu X.S'}(Y \at \ap \bot(\any)) \ra \vpn{S'}(Y)$
    which results in the term $\vpn{S'}(t)$. Since, by induction,
    $\trsred{\TRg{\Gamma;X\colon{S'}}{\Gamma;X\colon{S'}}\cup\TRctx{{\varphi_1}}{S'}}{\varphi_{S'}(t)}{u}$
    and since
    $\TRg{\Gamma;X\colon{S'}}{\Gamma;X\colon{S'}}\cup\TRctx{{\varphi_1}}{S'}\subseteq\TRgGG\cup\TRctx{{\varphi_1}}{\mu 
      X.S'}$ we eventually have
    $\trsred{\TRgGG\cup\TRctx{{\varphi_1}}{\mu
        X.S'}}{\varphi_{\mu X.S'}(t)}{u}$.
    
    If $\stratappctx{\Gamma}{\mu X.S'}{t}\sred \failres$ then
    $\stratappctx{\Gamma;{X\colon
        S'}}{S'}{t}\sred\failres$.  Using the same
    reasoning as for the successful case we obtain
    $\trsred{\TRgGG\cup\TRctx{{\varphi_1}}{\mu
        X.S'}}{\varphi_{\mu X.S'}(t)}{\bot(t)}$.

  \item $\TRctx{\vpg}{S}  \Rightarrow  S$

    If $\trsred{\TRgGG\cup\TRctx{{\varphi_1}}{\mu
        X.S'}}{\varphi_{\mu X.S'}(t)}{u}$ since $t\in\TF$
    contains no $\varphi$ symbols and $\TRgGG$ contains no rules
    potentially rewriting $\varphi_{\mu X.S'}(t)$, the first reduction
    is necessarily $\trsred{\varphi_{\mu X.S'}(Y \at \ap \bot(\any))
      \ra \vpn{S'}(Y)}{\varphi_{\mu X.S'}(t)}{\vpn{S'}(t)}$ and thus
    $\trsred{\TRgGG\cup\TRctx{{\varphi_1}}{\mu
        X.S'}}{\vpn{S'}(t)}{u}$ in strictly less steps than the
    original reduction.  We also have that
    $\TRg{\Gamma;X\colon{S'}}{\Gamma;X\colon{S'}}\cup\TRctx{{\varphi_1}}{S'}=\TRgGG\cup\TRctx{{\varphi_1}}{\mu
      X.S'} \setminus \{\varphi_{\mu X.S'}(\bot(Y))\ra\bot(Y),\varphi_{\mu
      X.S'}(Y\at\ap\bot(\any))\ra\vpn{S'}(Y)\} $ and since $\vpn{S'}(t)$ and
    all its reducts \wrt\ $\TRgGG\cup\TRctx{{\varphi_1}}{\mu
      X.S'}$ contain no $\varphi_{\mu X.S'}$ then
    $\trsred{\TRg{\Gamma;X\colon{S'}}{\Gamma;X\colon{S'}}\cup\TRctx{{\varphi_1}}{S'}}{\varphi_{S'}(t)}{u}$
    (in strictly less steps than the original reduction).  By
    induction, $\stratappctx{\Gamma;{X\colon
        S'}}{S'}{t}\sred u$ and thus
    $\stratappctx{\Gamma}{\mu X.S'}{t}\sred u$.
    
    If $\trsred{\TRgGG\cup\TRctx{{\varphi_1}}{\mu
        X.S'}}{\varphi_{\mu X.S'}(t)}{\bot(t)}$ the first
    reduction is still $\trsred{\varphi_{\mu X.S'}(Y \at \ap
      \bot(\any)) \ra \vpn{S'}(Y)}{\varphi_{\mu X.S'}(t)}{\vpn{S'}(t)}$ and
    thus $\trsred{\TRgGG\cup\TRctx{{\varphi_1}}{\mu
        X.S'}}{\vpn{S'}(t)}{\bot(u)}$ in strictly less steps than
    the original reduction. With the same arguments as before we
    obtain $\stratappctx{\Gamma;{X\colon
        S'}}{S'}{t}\sred \failres$ and thus
    $\stratappctx{\Gamma}{\mu X.S'}{t}\sred \failres$.

  \end{enumerate}

\item $S:=X$

  \begin{enumerate}
  \item $S \Rightarrow  \TRctx{\vpg}{S}$

    Since $\stratappctx{\Gamma;{X\colon S'}}{X}{t}\sred u$ then
    $\stratappctx{\Gamma;{X\colon S'}}{S'}{t}\sred u$ with the latter
    having a shorter derivation tree and thus, by induction,
    $\trsred{\TRg{\Gamma;X\colon{S'}}{\Gamma;X\colon{S'}}\cup\TRctx{{\varphi_1}}{S'}}{\varphi_{S'}(t)}{u}$.
    Since we supposed all bound variables (by a $\mu$ operator or a
    context assignment) to have different names, the only rewrite
    rules in $\TRg{\Gamma;X\colon{S'}}{\Gamma;X\colon{S'}}$ involving
    $\varphi_{X}$ are the ones generated by the context assignment:
    $\varphi_{X}(\bot(Y)) \ra \bot(Y)$ and $\varphi_{X}(Y \at \ap
    \bot(\any)) \ra \vpn{S'}(Y)$.  Consequently, the only possible
    reduction of $\varphi_{X}(t)$ \wrt\
    $\TRg{\Gamma;X\colon{S'}}{\Gamma;X\colon{S'}}$ is obtained by
    applying the latter rule which results in the term $\vpn{S'}(t)$. It
    is easy to check that
    $\TRg{\Gamma;X\colon{S'}}{\Gamma;X\colon{S'}}\cup\TRctx{{\varphi_1}}{S'}=\TRg{\Gamma;X\colon{S'}}{\Gamma;X\colon{S'}}$.
    Consequently
    $\trsred{\TRg{\Gamma;X\colon{S'}}{\Gamma;X\colon{S'}}}{\varphi_{S'}(t)}{u}$
    and thus
    $\trsred{\TRg{\Gamma;X\colon{S'}}{\Gamma;X\colon{S'}}}{\varphi_{X}(t)}{u}$.
    If $\stratappctx{\Gamma;{X\colon S'}}{X}{t}\sred\failres$ then
    $\stratappctx{\Gamma;{X\colon S'}}{S'}{t}\sred\failres$.  Using the
    same reasoning as for the successful case we obtain
    $\trsred{\TRg{\Gamma;X\colon{S'}}{\Gamma;X\colon{S'}}}{\varphi_{X}(t)}{\bot(t)}$.

  \item $\TRctx{\vpg}{S} \Rightarrow S$ 

    If
    $\trsred{\TRg{\Gamma;X\colon{S'}}{\Gamma;X\colon{S'}}}{\varphi_{X}(t)}{u}$
    since $t\in\TF$ contains no $\varphi$ symbols and since, as
    explained before the only rewrite rules in
    $\TRg{\Gamma;X\colon{S'}}{\Gamma;X\colon{S'}}$ involving
    $\varphi_{X}$ are the ones generated by the context assignment,
    the first reduction is necessarily
    $\trsred{\varphi_{X}(Y\at\ap\bot(\any))\ra\vpn{S'}(Y)}{\varphi_{X}(t)}{\vpn{S'}(t)}$
    and thus
    $\trsred{\TRg{\Gamma;X\colon{S'}}{\Gamma;X\colon{S'}}}{\vpn{S'}(t)}{u}$
    in strictly less steps than the original reduction. We have
    $\TRg{\Gamma;X\colon{S'}}{\Gamma;X\colon{S'}}\cup\TRctx{{\varphi_1}}{S'}=\TRg{\Gamma;X\colon{S'}}{\Gamma;X\colon{S'}}$
    and thus
    $\trsred{\TRg{\Gamma;X\colon{S'}}{\Gamma;X\colon{S'}}\cup\TRctx{{\varphi_1}}{S'}}{\vpn{S'}(t)}{u}$ 
    (in strictly less steps than the original reduction).  By
    induction, $\stratappctx{\Gamma;{X\colon S'}}{S'}{t}\sred u$ and
    thus $\stratappctx{\Gamma;{X\colon S'}}{X}{t}\sred u$.

    If
    $\trsred{\TRg{\Gamma;X\colon{S'}}{\Gamma;X\colon{S'}}}{\varphi_{X}(t)}{\bot(t)}$
    the first reduction is still
    $\trsred{\varphi_{X}(Y\at\ap\bot(\any))\ra\vpn{S'}(Y)}{\varphi_{X}(t)}{\vpn{S'}(t)}$
    and thus
    $\trsred{\TRg{\Gamma;X\colon{S'}}{\Gamma;X\colon{S'}}}{\vpn{S'}(t)}{\bot(u)}$
    in strictly less steps than the original reduction.  With the same
    arguments as before we obtain $\stratappctx{\Gamma;{X\colon
        S'}}{S'}{t}\sred \failres$ and thus
    $\stratappctx{\Gamma;{X\colon S'}}{X}{t}\sred\failres$.

  \end{enumerate}

\item $S:=\choice{S_1}{S_2}$

  \begin{enumerate}
  \item $S \Rightarrow  \TRctx{\vpg}{S}$

    If $\stratapp{\choice{S_1}{S_2}}{t}\sred u$ then
    $\stratapp{S_1}{t}\sred{u}$ or, $\stratapp{S_1}{t}\sred{\failres}$ and  $\stratapp{S_2}{t}\sred{r}$.
    In the first case, by induction,  
    $\trsred{\TRgGG \cup \TRctx{{\varphi_1}}{S_1}}{\vpo(t)}{v}$. 
    We have,
    $\trsred{\vpn{\choicesub{S_1}{S_2}}(x\at\ap\dummy)\ra\phich(\vpo(x))}{\vpn{\choice
     {S_1}{S_2}}(t)}{\phich(\vpo(t))}$
    and
    $\trsred{\TRgGG \cup \TR{{\varphi_1}}{S_1}}{\phich(\vpo(t))}{\phich(u)}$.
    Finally, 
    $\trsred{\phich(x \at \ap \bot(\any)) \ra  x}{\phich(u)}{u}$.
    For the second case, by induction,  
    $\trsred{\TRgGG \cup \TRctx{{\varphi_1}}{S_1}}{\vpo(t)}{\bot(t)}$
    and
    $\trsred{\TRgGG \cup \TRctx{{\varphi_2}}{S_2}}{\vpt(t)}{u}$.
    We have, as before,
    $\trsred{\vpn{\choicesub{S_1}{S_2}}(x\at\ap\dummy)\ra\phich(\vpo(x))}{\vpn{\choicesub
     {S_1}{S_2}}(t)}{\phich(\vpo(t))}$ 
    and
    $\trsred{\TRgGG \cup \TRctx{{\varphi_1}}{S_1}}{\phich(\vpo(t))}{\phich(\bot(u))}$. 
    Then,
    $\trsred{\phich(\bot(x))\ra\vpt(x)}{\phich(\bot(u))}{\vpt(u)}$ and finally,
    $\trsred{\TRgGG \cup \TRctx{{\varphi_2}}{S_2}}{\varphi_2(t)}{u}$.

    If $\stratapp{\choice{S_1}{S_2}}{t}\sred \failres$ then
    $\stratapp{S_1}{t}\sred{\failres}$ and $\stratapp{S_2}{t}\sred{\failres}$.
    By induction,
    $\trsred{\TRgGG \cup \TRctx{{\varphi_1}}{S_1}}{\vpo(t)}{\bot(t)}$ and
    $\trsred{\TRgGG \cup \TRctx{{\varphi_2}}{S_2}}{\vpt(t)}{\bot(t)}$.  We have,
    $\trsred{\vpn{\choice{S_1}{S_2}}(x\at\ap\dummy)\ra\phich(\vpo(x))}{\vpn{\choice{S_1}{S_2}}(t)}{\phich(\vpo(t))}$ 
    and
    $\trsred{\TRgGG \cup \TRctx{{\varphi_1}}{S_1}}{\phich(\vpo(t))}{\phich(\bot(u))}$.
    Then,
    $\trsred{\phich(\bot(x))\ra\vpt(x)}{\phich(\bot(u))}{\vpt(u)}$ and
    finally, $\trsred{\TRgGG \cup \TRctx{{\varphi_2}}{S_2}}{\varphi_2(t)}{\bot(u)}$.

  \item $\TRctx{\vpg}{S}  \Rightarrow  S$

    If $\trsred{\TRgGG \cup
      \TRctx{\vpn{\choice{S_1}{S_2}}}{\choice{S_1}{S_2}}}{\vpn{\choicesub{S_1}{S_2}}(t)}{u}$,
    since $t\in\TF$ contains no $\varphi$ symbols, the first reduction is
    necessarily
    $\trsred{\vpn{\choicesub{S_1}{S_2}}(x\at\ap\dummy)\ra\phich(\vpo(x))}{\vpn{\choicesub
     {S_1}{S_2}}(t)}{\phich(\vpo(t))}$.
    No rule can be applied at the top position until $\vpo(t)$ is
    reduced to a term in $\TF$ or to a term of the form
    $\bot(\any)$. In the former case, we can only have
    $\trsred{\TRgGG \cup \TRctx{{\varphi_1}}{S_1}}{\vpo(t)}{u}$ in less steps
    than the original reduction, so by
    induction, $\stratapp{S_1}{t}\sred{u}$. Consequently,
    $\stratapp{\choice{S_1}{S_2}}{t}\sred u$. In the latter case, we
    necessarily have (Lemma~\ref{lemma:sameFailure})
    $\trsred{\TRgGG \cup \TRctx{{\varphi_1}}{S_1}}{\vpo(t)}{\bot(t)}$ and
    $\trsred{\TRgGG \cup
      \TRctx{{\varphi_1}}{S_1}}{\phich(\vpo(t))}{\phich(\bot(t))}$. Then 
    we have 
    $\trsred{\phich(\bot(x))\ra\vpt(x)}{\phich(\bot(t))}{\vpt(t)}$, which can
    only be reduced as
    $\trsred{\TRgGG \cup \TRctx{{\varphi_2}}{S_2}}{\vpt(t)}{u}$. By induction,
    $\stratapp{S_1}{t}\sred{\failres}$ and $\stratapp{S_2}{t}\sred{u}$  and consequently,
    $\stratapp{\choice{S_1}{S_2}}{t}\sred u$.

    If $\trsred{\TRgGG \cup
      \TRctx{\vpn{\choicesub{S_1}{S_2}}}{\choice{S_1}{S_2}}}{\vpn{\choicesub{S_1}{S_2}}(t)}{\bot(t)}$,
    since $\vpn{\choicesub{S_1}{S_2}}(t)$ is in normal form w.r.t. $\TRgGG \cup
    \TRctx{{\varphi_1}}{S_1}$ and $\TRgGG \cup \TRctx{{\varphi_2}}{S_2}$
    (Lemma~\ref{lemma:rigidTF}), the first reduction is necessarily
    $\trsred{\vpn{\choicesub{S_1}{S_2}}(x\at\ap\dummy)\ra\phich(\vpo(x))}{\vpn{\choicesub
        {S_1}{S_2}}(t)}{\phich(\vpo(t))}$.  No rule can be applied at the top
    position until $\vpo(t)$ is reduced to a term in $\TF$ or to a term of the
    form $\bot(\any)$. If it is a term in $\TF$ then the rule
    $\phich(x\at\ap\bot(\any))\ra{x}$ is the only one that can be applied and
    the final result is not $\bot(t)$. Thus, $\trsred{\TRgGG \cup
      \TRctx{{\varphi_1}}{S_1}}{\vpo(t)}{\bot(t)}$ (in less steps than the
    original reduction) and $\trsred{\TRgGG \cup
      \TRctx{{\varphi_1}}{S_1}}{\phich(\vpo(t))}{\phich(\bot(t))}$.  The only
    possible reduction is
    $\trsred{\phich(\bot(x))\ra\vpt(x)}{\phich(\bot(t))}{\vpt(t)}$ and then
    $\trsred{\TRgGG \cup \TRctx{{\varphi_2}}{S_2}}{\vpt(t)}{\bot(t)}$ in less
    steps than the original reduction.  By induction,
    $\stratapp{S_1}{t}\sred{\failres}$ and $\stratapp{S_2}{t}\sred{\failres}$
    and consequently, $\stratapp{\choice{S_1}{S_2}}{t}\sred \failres$.

  \end{enumerate}

\item $S:= \all(S')$
  \begin{enumerate}
  \item $S \Rightarrow  \TRctx{\vpg}{S}$

    If $\stratapp{ \all(S')}{t}\sred u$ with $t=f(t_1,\ldots,t_n)$ then
    $\forall i \in {[}1{,}n{]}, \stratapp{S'}{t_i}\sred{u_i}$ and
    $u=f(u_1,\ldots,u_n)$.  Then, by induction,
    $\trsred{\TRgGG \cup \TRctx{{\vpS}}{S'}}{\vpn{S'}(t_i)}{u_i}$.  When we apply the
    rules corresponding to the encoding of $\all$ we obtain
    \begin{align*}
    \trsrednl {\vpn{\all(S')}(f(x_1,\ldots,x_n)) \ra 
      \psiF(\vpn{S'}(x_1),\ldots,\vpn{S'}(x_n),f(x_1,\ldots,x_n))}{\vpn{\all(S')}(f(t_1,\ldots,t_n))}
    { \psiF(\vpn{S'}(t_1),\ldots,\vpn{S'}(t_n),f(t_1,\ldots,t_n))}
    \end{align*}
    and no other rule can be applied at the top position for this
    latter term. By the induction hypothesis we have
    $\trsredstar{\TRgGG \cup \TRctx{{\vpn{S'}}}{S'}^{*}}{\psiF(\vpn{S'}(t_1),\ldots,\vpn{S'}(t_n),f(t_1,\ldots,t_n))}{\psiF(u_1,\ldots,u_n,f(t_1,\ldots,t_n))}$
    and subsequently
    \begin{align*}
      \trsrednl{\psiF(x_1\at\ap\bot(\any),\ldots,x_n\at\ap\bot(\any),\any)\ra
      f(x_1,\ldots,x_n)}
    {\psiF(u_1,\ldots,u_n,f(t_1,\ldots,t_n))}{f(u_1,\ldots,u_n)}
    \end{align*}
    If $\stratapp{ \all(S')}{t}\sred \failres$ with $t=f(t_1,\ldots,t_n)$
    then $\exists i\in{[}1{,}n{]},\step{\stratapp{S'}{t_i}}{\failres}$.
    By induction, $\trsred{\TRgGG \cup \TRctx{{\vpS}}{S'}}{\vpn{S'}(t_i)}{\bot(t_i)}$.  As
    before, when we apply the rules corresponding to the encoding of
    $\all$ we obtain
    \begin{align*}
      \trsrednl{\vpn{\all(S')}(f(x_1,\ldots,x_n))\ra\psiF(\vpn{S'}(x_1),\ldots,\vpn{S'}(x_n),f(x_1,\ldots,x_n))}{\vpn{\all(S')}(f(t_1,\ldots,t_n))}{\psiF(\vpn{S'}(t_1),\ldots,\vpn{S'}(t_n),f(t_1,\ldots,t_n))}
    \end{align*}
    and using induction
\begin{align*}
    \qquad\qquad\qquad\ \,\psiF(\vpn{S'}(t_1),\ldots,\vpn{S'}(t_i),\ldots,&\vpn{S'}(t_n),f(t_1,\ldots,t_n))\cr
    &\trsredstar
   {\TRgGG \cup\TRctx{{\vpS}}{S'}}
   {}
   {\psiF(u_1,\ldots,\bot(t_i),\ldots,u_n,f(t_1,\ldots,t_n))\,.}
\end{align*}
    We have then
    \begin{align*}
    \trsrednl{\psiF(\bot(\any),\ldots,\any,z)\ra\bot(z)}{\psiF(u_1,\ldots,\bot(t_i),\ldots,u_n,f(t_1,\ldots,t_n))}{\bot(f(t_1,\ldots,t_n))} 
    \end{align*}
    as required. 

    If $t$ is a constant $c$ then $\stratapp{ \all(S')}{c}\sred c$.  We
    have
    $\trsred{\vpn{\all(S')}(c)\ra\varphi_{c}(c)}{\vpn{\all(S')}(c)}{\varphi_{c}(c)}$
    and $\trsred{\varphi_{c}(\any)\ra c}{\varphi_{c}(c)}{c}$.

  \item $\TRctx{\vpg}{S}  \Rightarrow  S$

    We consider $t=f(t_1,\ldots,t_n)$ and handle the case of a
    constant later on.
    If $\trsred{\TRgGG \cup \TRctx{\vpg}{\all(S')}}{\vpn{\all(S')}(t)}{u}$, since $t\in\TF$
    contains no $\varphi$ symbols, we have first
    $\trsred{\vpn{\all(S')}(f(x_1,\ldots,x_n))\ra\psiF(\vpn{S'}(x_1),\ldots,\vpn{S'}(x_n),f(x_1,\ldots,x_n))}{\vpn{\all(S')}(f(t_1,\ldots,t_n))}{\psiF(\vpn{S'}(t_1),\ldots,\vpn{S'}(t_n),f(t_1,\ldots,t_n))}$.
    No rule can be applied at the top position until the $\vpn{S'}(t_i)$
    are reduced to terms in $\TF$ or at least one of them is reduced
    to a term of the form $\bot(\any)$.
    The former case holds only if 
    $\trsred{\TRgGG \cup \TRctx{{\vpS}}{S'}}{\vpn{S'}(t_i)}{u_i}$ (in less steps
    than the original reduction) and 
    \[
      \hspace{6em}\trsredstar{\TRgGG \cup
      \TRctx{{\vpS}}{S'}^{*}}{\psiF(\vpn{S'}(t_1),\ldots,\vpn{S'}(t_n),f(t_1,\ldots,t_n))}{\psiF(u_1,\ldots,u_n,f(t_1,\ldots,t_n))}.
    \]
    This term can be further reduced only by the rule
    $\psiF(x_1\at\ap\bot(\any),\ldots,x_n\at\ap\bot(\any),\any) \ra
    f(x_1,\ldots,x_n)$ to $f(u_1,\ldots,u_n)=u$. By induction,
    $\stratapp{S'}{t_i}\sred{u_i}$.  Consequently, $\stratapp{\all(S')}{t}\sred
    u$.
    In the latter case, we necessarily have (Lemmas~\ref{lemma:sameFailure})
    $\trsred{\TRgGG \cup \TRctx{{\vpS}}{S'}}{\vpn{S'}(t_i)}{\bot(t_i)}$ for some
    $i$, and then
    \begin{align*}
      \trsrednl{\hspace{1em}\psiF(\ldots,\bot(\any),\ldots,\any,z)\ra\bot(z)}{\psiF(\vpn{S'}(t_1),\ldots,
        \bot(t_i), \ldots,
        \vpn{S'}(t_n),f(t_1,\ldots,t_n))}{\bot(f(t_1,\ldots,t_n))}
    \end{align*}
    We obtain a term not in $\TF$, hence a contradiction with the original
    hypothesis.

    If $\trsred{\TRgGG \cup \TRctx{\vpg}{\all(S')}}{\vpn{\all(S')}(t)}{\bot(t)}$, since
    $\vpn{\all(S')}(t)$ is in normal form w.r.t. $\TRgGG \cup \TRctx{{\vpS}}{S'}$
    (Lemma~\ref{lemma:rigidTF}), the first reduction can only be,
    as before,
    \begin{align*}
      \trsrednl{\vpn{\all(S')}(f(x_1,\ldots,x_n))\ra\psiF(\vpn{S'}(x_1),\ldots,\vpn{S'}(x_n),f(x_1,\ldots,x_n))}{\vpn{\all(S')}(f(t_1,\ldots,t_n))}{\psiF(\vpn{S'}(t_1),\ldots,\vpn{S'}(t_n),f(t_1,\ldots,t_n))}
    \end{align*}
    and no rule can be applied at the top position until the $\vpn{S'}(t_i)$ are
    reduced to terms in $\TF$ or at least one of them is reduced to a term of
    the form $\bot(\any)$. We already handled the former case just before. As we
    have seen, for the latter case we have $\trsred{\TRgGG \cup
      \TRctx{{\vpS}}{S'}}{\vpn{S'}(t_i)}{\bot(t_i)}$ for some $i$ in less steps
    than the original reduction, and eventually
    $\trsred{\psiF(\ldots,\bot(\any),\ldots,\any,z)\ra\bot(z)}{\psiF(\vpn{S'}(t_1),\ldots,\vpn{S'}(t_n),f(t_1,\ldots,t_n))}{\bot(f(t_1,\ldots,t_n))}$.
    By induction, $\stratapp{S'}{t_i}\sred{\failres}$ and thus
    $\stratapp{\all(S')}{f(t_1,\ldots,t_n)}\sred{\failres}$.

    If $t$ is a constant $c$ then
    $\trsred{\vpn{\all(S')}(c)\ra\varphi_{c}(c)}{\vpn{\all(S')}(c)}{\varphi_{c}(c)}$
    and $\trsred{\varphi_{c}(\any)\ra c}{\varphi_{c}(c)}{c}$.  On the
    other hand we have $\stratapp{\all(S')}{c}\sred c$ which confirms the
    property.

  \end{enumerate}

\item $S:= \one(S')$

  \begin{enumerate}
  \item $S \Rightarrow  \TRctx{\vpg}{S}$

    If $\stratapp{ \one(S')}{t}\sred u$ with $t=f(t_1,\ldots,t_n)$ then
    $\exists i \in {[}1{,}n{]}, \stratapp{S}{t_i}\sred{u_i}$ and
    $u=f(t_1,\ldots,u_i\ldots,t_n)$.  Then, by induction,
    $\trsred{\TRgGG \cup \TRctx{{\vpS}}{S'}}{\vpn{S'}(t_i)}{u_i}$.
    The encoding implements a leftmost behaviour for $\one$, \ie it
    supposes that $\forall j<i, \stratapp{S'}{t_j}\sred{\failres}$ and
    if we suppose that this assumption holds in our case then, by induction,
    $\trsred{\TRgGG \cup \TRctx{{\vpS}}{S'}}{\vpn{S'}(t_j)}{\bot(t_j)}$ for all $j<i$.
    When we apply the rules corresponding to the encoding of $\one$ we
    obtain
    $\trsred{\vpn{\one(S')}(f(x_1,\ldots,x_n)) \ra
      \psif{1}(\vpn{S'}(x_1),x_2,\ldots,x_n)}{\vpn{\one(S')}(f(t_1,\ldots,t_n))}{\psif{1}(\vpn{S'}(t_1),t_2,\ldots,t_n)}$.
    and no other rule can be applied at the top position for this
    latter term until $\vpn{S'}(t_1)$ is reduced to a term in $\TF$ or to a term of the form
    $\bot(\any)$. 
    We have
\[\trsredstar{\TRgGG \cup
    \TRctx{{\vpS}}{S'}}{\psif{1}(\vpn{S'}(t_1),\ldots,t_n)}{\psif{1}(\bot(t_1),\ldots,t_n)}
\]
    and then 
\begin{align*}
  \trsrednl
  {\psif{1}(\bot(x_1),x_{2},\ldots,x_n) \ra \psif{2}(\bot(x_1),\vpn{S'}(x_{2}),\ldots,x_n)}
  {\psif{1}(\bot(t_1),\ldots,t_n)}
  {\psif{2}(\bot(t_1),\vpn{S'}(t_2)\ldots,t_n)\,.}
\end{align*}
    Repeating the reasoning for all $j < i$, we eventually get
\[\trsredstar{\TRgGG \cup \TRctx{{\vpS}}{S'}^{*}}{\psif{1}(\vpn{S'}(t_1),\ldots,t_n)}{\psif{i}(\bot(t_1),\ldots,\bot(t_{i-1}),u_i\ldots,t_n)}\]
    and then  
    \begin{align*}
    \trsrednl
    {\psif{i}(\bot(x_1),\ldots,\bot(x_{i-1}),x_i\at\ap\bot(\any),x_{i+1},\ldots,x_n) \ra f(x_1,\ldots,x_n)}
    {\psif{i}(\bot(t_1),\ldots,\bot(t_{i-1}),u_i\ldots,t_n)}
    {\psif{i}(t_1,\ldots,u_i,\ldots,t_n)}
    \end{align*}
    as required.

    If $\stratapp{\one(S')}{t}\sred \failres$ with $t=f(t_1,\ldots,t_n)$
    then $\forall i\in{[}1{,}n{]},\step{\stratapp{S'}{t_i}}{\failres}$.
    By induction, $\trsred{\TRgGG \cup \TRctx{{\vpS}}{S'}}{\vpn{S'}(t_i)}{\bot(t_i)}$.  As
    before, applying the rules corresponding to the encoding of
    $\one$ gives
    $\trsredstar{\TRgGG \cup \TRctx{{\vpS}}{S'}^{*}}{\vpn{\one(S')}(f(x_1,\ldots,x_n))}{\psif{n}(\bot(t_1),\ldots,\bot(t_n))}$
    and then 
    \begin{align*}
      \trsrednl{\psif{n}(\bot(x_1),\ldots,\bot(x_n)) \ra
        \bot(f(x_1,\ldots,x_n))}{\psif{n}(\bot(t_1),\ldots,\bot(t_n))}{\bot(f(t_1,\ldots,t_n))}
    \end{align*}
    as wished. 

    If $t$ is a constant $c$ then $\stratapp{\one(S')}{c}\sred \failres$, and we
    have $\trsred{\vpn{\one(S')}(c) \ra \bot(c)}{\vpn{(\one(S'))}(c)}{\bot(c)}$
    as well.

  \item $\TRctx{\vpg}{S}  \Rightarrow  S$

    If $\trsred{\TRgGG \cup \TRctx{\vpg}{\one(S')}}{\vpn{\one(S')}(t)}{u}$, then
    $t=f(t_1,\ldots,t_n)$, because if $t$ is a constant $c$ then we would
    necessarily have $\trsred{\vpn{\one(S')}(c) \ra
      \bot(c)}{\vpn{\one(S')}(c)}{\bot(c)}$, which contradicts the hypothesis.
    Because $t\in\TF$ contains no $\varphi$ symbols, the first reduction is
    necessarily 
\[\trsrednl{\vpn{\one(S')}(f(x_1,\ldots,x_n)) \ra
      \psif{1}(\vpn{S'}(x_1),x_2,\ldots,x_n)}{\vpn{\one(S')}(f(t_1,\ldots,t_n))}{\psif{1}(\vpn{S'}(t_1),t_2,\ldots,t_n)\,.}
\]
    No rule can be applied at the top position until the $\vpn{S'}(t_1)$
    is reduced to a term in $\TF$ or to a term of the form
    $\bot(\any)$.
    In the former case, we have
    $\trsred{\TRgGG \cup \TRctx{{\vpS}}{S'}}{\vpn{S'}(t_1)}{u_1}$ in less steps
    than the original reduction, and also
    $\trsredstar{\TRgGG \cup
      \TRctx{{\vpS}}{S'}}{\psif{1}(\vpn{S'}(t_1),\ldots,t_n)}{\psif{1}(u_1,\ldots,t_n)}$ 
    which can be further reduced only by the rule
    $\psif{1}(x_1\at\ap\bot(\any),x_{2},\ldots,x_n) \ra
    f(x_1,\ldots,x_n)$ to $f(u_1,\ldots,t_n)=u$.  By induction,
    $\stratapp{S'}{t_1}\sred{u_1}$, and consequently, $\stratapp{\one(S')}{t}\sred
    u$.
    In the latter case, we necessarily have
    (Lemma~\ref{lemma:sameFailure})
    $\trsred{\TRgGG \cup \TRctx{{\vpS}}{S'}}{\vpn{S'}(t_1)}{\bot(t_1)}$ and
    $\trsredstar{\TRgGG \cup \TRctx{{\vpS}}{S'}}{\psif{1}(\vpn{S'}(t_1),\ldots,t_n)}{\psif{1}(\bot(t_1),\ldots,t_n)}$
    which can be further reduced only by
    $\psif{1}(\bot(x_1),x_{2},\ldots,x_n) \ra
    \psif{2}(\bot(x_1),\vpn{S'}(x_{2}),\ldots,x_n)$ to
    $f(\bot(t_1),\vpn{S'}(t_{2}),\ldots,t_n)$.  Once again, no rule can be
    applied at the top position until the $\vpn{S'}(t_2)$ is reduced to a
    term in $\TF$, in which case we conclude as before, or to a term
    of the form $\bot(\any)$, in which case we continue the same way
    and we eventually get a term of the form
    $\psif{i}(\bot(t_1),\ldots,\bot(t_{i-1}),u_i,t_{i+1},\ldots,t_n)$,
    then reduced by
\[\psif{i}(\bot(x_1),\ldots,\bot(x_{i-1}),x_i\at\ap\bot(\any),x_{i+1},\ldots,x_n)
    \ra f(x_1,\ldots,x_n)
\]
    to
    $f(t_1,\ldots,u_i,t_{i+1},\ldots,t_n)=u$.  In this case we have,
    by induction, $\stratapp{S'}{t_j}\sred{\bot(t_j)}$ for all $j<i$ and
    $\stratapp{S'}{t_i}\sred{u_i}$, and thus, $\stratapp{\one(S')}{t}\sred{u}$.
    Note that we can always get an $u_i\in\TF$ at some point since
    otherwise we would eventually obtain the term
    $\psif{n}(\bot(t_1),\ldots,\bot(t_n))$ which can be reduce only to
    $\bot(f(t_1,\ldots,t_n))$ which is not a term in $\TF$ and thus
    contradicts the original hypothesis.

    If $\trsred{\TRgGG \cup
      \TRctx{\vpg}{\one(S')}}{\vpn{\one(S')}(t)}{\bot(t)}$ and $t$  is not a
    constant, because
    $\vpn{\one(S')}(t)$ is in normal form w.r.t. $\TRgGG \cup
    \TRctx{{\vpS}}{S'}$ (Lemma~\ref{lemma:rigidTF}), the first reduction can
    only be 
\begin{align*}
\trsrednl{\vpn{\one(S')}(f(x_1,\ldots,x_n)) \ra
      \psif{1}(\vpn{S'}(x_1),x_2,\ldots,x_n)}{\vpn{\one(S')}(f(t_1,\ldots,t_n))}{\psif{1}(\vpn{S'}(t_1),t_2,\ldots,t_n)}
\end{align*}
    and no rule can be applied at the top position until the $\vpn{S'}(t_1)$ is
    reduced to a term in $\TF$ or to a term of the form $\bot(\any)$.  As we
    have seen just before, the former case leads to a term in $\TF$
    which does not correspond to our hypothesis. We have also already handled
    the second case but we supposed that at some point we have $\trsred{\TRgGG
      \cup \TRctx{{\vpS}}{S'}}{\vpn{S'}(t_i)}{u_i}$ which would lead to an
    eventual reduction to a term in $\TF$ which, once again, does not correspond
    to our hypothesis. The only remaining possibility is $\trsred{\TRgGG \cup
      \TRctx{{\vpS}}{S'}}{\vpn{S'}(t_i)}{\bot(\any)}$ for all $i\leq n$ and
    thus, by Lemma~\ref{lemma:sameFailure}, $\trsred{\TRgGG \cup
      \TRctx{{\vpS}}{S'}}{\vpn{S'}(t_i)}{\bot(t_i)}$ (in less steps than the
    original reduction) for all $i\leq n$. In this
    case, we obtain $\trsred{\psif{n}(\bot(x_1),\ldots,\bot(x_n)) \ra
      \bot(f(x_1,\ldots,x_n))}{\psif{n}(\bot(t_1),\ldots,\bot(t_n))}{\bot(f(t_1,\ldots,t_n))}$.
    By induction, $\stratapp{S'}{t_i}\sred{\failres}$ and thus
    $\stratapp{\one(S')}{f(t_1,\ldots,t_n)}\sred{\failres}$.

    If $t$ is a constant $c$ then $\trsred{\vpn{\one(S')}(c) \ra
      \bot(c)}{\vpn{\one(S')}(c)}{\bot(c)}$.  We also have
    $\stratapp{\one(S')}{c}\sred \failres$.\qedhere
  \end{enumerate}
\end{enumerate}
\end{proof}

\confluence*
\begin{proof}
  We consider two subsets of $\TRctx{\Gamma}{S}\cup\TRg{\Gamma}{\Gamma}$: a set $T_1$
  consisting of the rewrite rules in $\TRctx{\Gamma}{S}\cup\TRg{\Gamma}{\Gamma}$ obtained by
  expanding the rule schemas of the form
  $\varphi_{l\ra{r}}(\Xx\at\ap{l})\ra\bot(\Xx)$ and a set $T_2$
  consisting of the remaining rules in $\TRctx{\Gamma}{S}\cup\TRg{\Gamma}{\Gamma}$. Keeping in
  mind that all $\varphi$ symbols are freshly generated we can easily
  notice that $T_2$ is orthogonal and thus confluent. $T_1$ is linear
  but not orthogonal because of possible critical pairs between rules
  generated for a rewrite rule $l\ra{r}$ by
  $\varphi_{l\ra{r}}(\Xx\at\ap{l})\ra\bot(\Xx)$; there are no critical
  pairs between the rules generated for different rewrite rules
  because of the different $\varphi$ head symbol in the left-hand
  sides of the corresponding rules. All these critical pairs are
  trivially joinable and since $T_1$ is also obviously terminating
  then it is also confluent.
  Since $T_1$ and $T_2$ are orthogonal to each other, {\ie} there
  is no overlap between a rule from $T_1$ and one from $T_2$, then $T_1\cup
  T_2=\TRctx{\Gamma}{S}\cup\TRg{\Gamma}{\Gamma}$ is confluent~\cite{Ohlebusch2002}.  
\end{proof}

\section{Proofs for the Meta-Encoding}
\label{se:appendixMeta}

\begin{figure}[!tp]
\begin{mathpar}
\begin{array}{lr@{\hspace{5pt}}c@{\hspace{5pt}}l}
  \mathbf{(ME1)}&
  \TRm{\id} & =  & \{  ~ \varphi_{\id}(\Xx \at \appl{\any}{\any}) \ra \Xx, 
                                  \quad
                                \varphi_{\id}(\bot(\Xx)) \ra \bot(\Xx)~\} \\[+4pt]
\mathbf{(ME2)}&
  \TRm{\fail} & = & \{ ~ \varphi_{\fail}(\Xx \at \appl{\any}{\any}) \ra \bot(\Xx),
                                  \quad
                                \varphi_{\fail}(\bot(\Xx)) \ra \bot(\Xx) ~\} \\ [+4pt]
\mathbf{(ME3)}&
\TRm{l \ra r} & = & \{ ~   \varphi_{l\ra r}(\symb{l}) \ra \symb{r}, 
                                \quad
                                \varphi_{l\ra r}(\Xx \at \symb{\ap l}) \ra  \bot(\Xx),\\[+2pt]
                          &&&  ~~~~~ 
                              \varphi_{l\ra r}(\bot(\Xx)) \ra \bot(\Xx)~~\} 
                              \\[+4pt]
\mathbf{(ME4)}&
  \TRm{\seq{S_1}{S_2}} & = &  \TRm{S_1} ~\cup~ \TRm{S_2}  \\[+2pt]
                          &\bigcup & \{ &  \varphi_{\seqsub{S_1}{S_2}}(\Xx \at \appl{\any}{\any}) \ra \phiseq( \vpt(\vpo(\Xx)), \Xx) ,\\[+2pt] 
                          &&& \varphi_{\seqsub{S_1}{S_2}}(\bot(\Xx)) \ra \bot(\Xx) , \\[+2pt]
                          &&& \phiseq(\Xx \at \appl{\any}{\any}, \any) \ra  \Xx, 
                          \quad
                          \phiseq(\bot(\any), \Xx) \ra  \bot(\Xx) ~~\} 
                          \\[+4pt]
\mathbf{(ME5)}&
  \TRm{\choice{S_1}{S_2}} & =  & \TRm{S_1} ~\cup~ \TRm{S_2}  \\[+2pt]
                           &\bigcup & \{ &  \varphi_{\choicesub{S_1}{S_2}}(\Xx \at \appl{\any}{\any}) \ra  \phich(\vpo(\Xx)) , \\[+2pt]
                           &&&
                           \varphi_{\choicesub{S_1}{S_2}}(\bot(\Xx)) \ra \bot(\Xx) , \\[+2pt]
                           &&&  \phich(\bot(\Xx)) \ra  \vpt(\Xx) , 
                           \quad
                           \phich(\Xx \at \appl{\any}{\any}) \ra  \Xx ~~\}
                           \\[+4pt]
\mathbf{(ME6)}&
  \TRm{\mu X\dotsym S} & =  & \TRm{S} \\[+2pt]
                 &\bigcup & \{ &  \varphi_{\mu X\dotsym S}(\Xx \at \appl{\any}{\any}) \ra  \vpS(\Xx) ,
                 \quad
                 \varphi_{\mu X\dotsym S}(\bot(\Xx)) \ra \bot(\Xx),  \\[+2pt]
                 &&&  \varphi_{X}(\Xx \at \appl{\any}{\any}) \ra  \vpS(\Xx) , 
                 \quad
                 \varphi_{X}(\bot(\Xx)) \ra \bot(\Xx)  ~~\} 
                 \\[+4pt]
\mathbf{(ME7)}&
  \TRm{X} & =  & \emptyset
  \\[+4pt]
\\[+1pt]
\mathbf{(ME8)}&
\TRm{\all(S)} & =  & \TRm{S} ~\cup~ \TRlists  \\[+2pt]
&\bigcup & \{ &  \varphi_{\all(S)}(\bot(\Xx)) \ra \bot(\Xx) , \\[+2pt]
&&& \varphi_{\all(S)}(\appl{\f}{\args}) \ra \propag{\appl{\f}{\varphi^{list}_{\all(S)}(\args)}} , \\[+2pt]
&&& \varphi^{list}_{\all(S)}(\cons{\head}{\tail}) \ra \varphi'_{\all(S)}(\vpS(\head),\tail,\cons{\head}{\nil},\nil), \\[+2pt]
&&& \varphi^{list}_{\all(S)}(\nil) \ra \nil, \\[+2pt]
&&& \varphi'_{\all(S)}(\bot(\any),\todo,\revtried,\any) \ra \botlist{\rconcat{\revtried}{\todo}}, \\[+2pt]
&&& \varphi'_{\all(S)}(\Xx \at \appl{\any}{\any},\nil,\any,\revdone) \ra \reverse{\cons{\Xx}{\revdone}}, \\[+2pt]
&&& \varphi'_{\all(S)}(\Xx \at \appl{\any}{\any},\cons{\head}{\tail},\revtried,\revdone) \ra \\[+2pt]
&&& \multicolumn{1}{r}{\varphi'_{\all(S)}(\vpS(\head), \tail, \cons{\head}{\revtried}, \cons{\Xx}{\revdone}) ~~\}} \\[+2pt]
    \\[+4pt]
\mathbf{(ME9)}&
  \TRm{\one(S)} & = & \TRm{S}~\cup~ \TRlists \\[+2pt]
&\bigcup & \{ &  \varphi_{\one(S)}(\bot(\Xx)) \ra \bot(\Xx) , \\[+2pt]
&&& \varphi_{\one(S)}(\appl{\f}{\args}) \ra \propag{\appl{\f}{\varphi^{list}_{\one(S)}(\args)}} , \\[+2pt]
&&& \varphi^{list}_{\one(S)}(\cons{\head}{\tail}) \ra  \varphi'_{\one(S)}(\vpS(\head),\tail,\cons{\head}{\nil}), \\[+2pt]
&&& \varphi^{list}_{\one(S)}(\nil) \ra \botlist{\nil}, \\[+2pt]
&&& \varphi'_{\one(S)}(\bot(\any),\nil, \revtried) \ra \botlist{\reverse{\revtried}}, \\[+2pt]
&&& \varphi'_{\one(S)}(\bot(\any),\cons{\head}{\tail}, \revtried) \ra \varphi'_{\one(S)}(\vpS(\head),\tail,\cons{\head}{\revtried}), \\[+2pt]
&&& \varphi'_{\one(S)}(\Xx \at \appl{\any}{\any},\todo,\cons{\any}{\revtried})
\ra \\[+2pt]
&&& \multicolumn{1}{r}{\rconcat{\revtried}{\cons{\Xx}{\todo}}}
~~\} \\[+2pt]
  \\[+4pt]
  \mathbf{(ME10)}&
  \TRmg{\Gamma; {X\colon S}} & =  &  \TRmg{\Gamma} ~\cup~ \TRm{S}             \\[+2pt]
                           &\bigcup &  \{  & \varphi_{X}(\Xx \at \appl{\any}{\any}) \ra  \vpS(\Xx) , 
                           \quad
                           \varphi_{X}(\bot(\Xx)) \ra \bot(\Xx)   ~~\}
  \\[+4pt]
  \mathbf{(ME11)}&
  \TRmg{\emptyctx} & =  &  \emptyset
\end{array}
\end{mathpar}
\caption{\label{fig:metaencoding-complete}Strategy translation for meta-encoded terms;
  $\TRlists$ is defined in Figure~\ref{fig:aux-symbols}.}
\end{figure}

\begin{lem}
  \label{lemma:all-success}
  Given terms $t_1 \ldots t_n \in \TF$, $S$, and $\Gamma$
  such that $\varS S \subseteq \domS \Gamma$, $\trsred{\TRm{S}
    \cup \TRmg{\Gamma}}{\vpn{S}(\symb {t_i})}{\symb{t'_i}}$ with $\symb{t'_i}
  \neq \bot(\symb{t_i})$ for all $i$ iff
  \begin{multline*}
    \trsrednl{\TRm{\all(S)}
      \cup \TRmg{\Gamma}}{\varphi'_{\all(S)}(\varphi_{S}(\symb{t_1}), \symb{t_2} \conss \ldots \conss
      \symb{t_n} \conss \nil, \symb{t_1} \conss \nil, \nil)}{\symb{t'_1} \conss \ldots \conss
      \symb{t'_n} \conss \nil}.
  \end{multline*}
\end{lem}

\begin{proof}
  Because $\TRm{\all(S)}$ contains $\TRm{S}$, and by definition of
  $\TRm{\all(S)}$, $\trsred{\TRm{S} \cup \TRmg{\Gamma}}{\vpn{S'}(\symb
    {t_i})}{\symb{t'_i}}$ with $\symb{t'_i} \neq \bot(\symb{t_i})$ for all $i$
  iff
  \begin{multline*}
    \varphi'_{\all(S)}(\varphi_{S}(\symb{t_i}), \symb{t_{i+1}} \conss \ldots
    \conss \symb{t_n} \conss \nil, \symb{t_i} \conss \ldots \symb{t_1} \conss
    \nil, \symb{t'_{i-1}} \conss \ldots \conss
    \symb{t'_1} \conss \nil) \\
    \multieval{\TRm{\all(S)} \cup \TRmg{\Gamma}} \\
    \varphi'_{\all(S)}(\varphi_{S}(\symb{t_{i+1}}), \symb{t_{i+2}} \conss \ldots
    \conss \symb{t_n} \conss \nil, \symb{t_{i+1}} \conss \ldots \symb{t_1} \conss
    \nil, \symb{t'_{i}} \conss \ldots \conss
    \symb{t'_1} \conss \nil)
  \end{multline*}
  for all $i$. As a result, we deduce
  \begin{align*}  
    & \varphi'_{\all(S)}(\varphi_{S}(\symb{t_1}), \symb{t_2} \conss \ldots
    \conss
    \symb{t_n} \conss \nil, \symb{t_1} \conss \nil, \nil) \\
    \multieval{\TRm{\all(S)} \cup \TRmg{\Gamma}}~&
    \varphi'_{\all(S)}(\varphi_{S}(\symb{t_{n}}), \nil, \symb{t_{n}} \conss
    \ldots \symb{t_1} \conss \nil, \symb{t'_{n-1}} \conss \ldots \conss
    \symb{t'_1} \conss \nil) \\
    \multieval{\TRm{\all(S)} \cup \TRmg{\Gamma}}~&
    \varphi'_{\all(S)}(\symb{t'_{n}}, \nil, \symb{t_{n}} \conss \ldots
    \symb{t_1} \conss \nil, \symb{t'_{n-1}} \conss \ldots \conss
    \symb{t'_1} \conss \nil)\\
    \trs_{\TRm{\all(S)} \cup \TRmg{\Gamma}}~& \reverse{\symb{t'_{n}} \conss
      \symb{t'_{n-1}} \conss \ldots \conss \symb{t'_1} \conss \nil}\\
    \multieval{\TRm{\all(S)} \cup \TRmg{\Gamma}}~& \symb{t'_1} \conss \ldots
    \conss \symb{t'_n} \conss \nil
  \end{align*}
  as wished.
\end{proof}

\begin{lem}
  \label{lemma:all-failure}
  Given terms $t_1 \ldots t_n \in \TF$, $S$, and $\Gamma$
  such that $\varS S \subseteq \domS \Gamma$, $\trsred{\TRm{S}
    \cup \TRmg{\Gamma}}{\vpn{S}(\symb {t_j})}{\symb{t'_j}}$ with $\symb{t'_j}
  \neq \bot(\symb{t_j})$ for all $j < i_0$ and $\trsred{\TRm{S}
    \cup \TRmg{\Gamma}}{\vpn{S}(\symb {t_{i_0}})}{\bot(\symb{t_{i_0}})}$ iff 
  \begin{multline*}
    \trsrednl{\TRm{\all(S)} \cup
      \TRmg{\Gamma}}{\varphi'_{\all(S)}(\varphi_{S}(\symb{t_1}), \symb{t_2}
      \conss \ldots \conss \symb{t_n} \conss \nil, \symb{t_1} \conss \nil,
      \nil)}{\botlist{\symb{t_1} \conss \ldots \conss \symb{t_n} \conss \nil}}.
  \end{multline*}
\end{lem}

\begin{proof}
  As in the proof of Lemma~\ref{lemma:all-success}, $\trsred{\TRm{S} \cup
    \TRmg{\Gamma}}{\vpn{S}(\symb {t_j})}{\symb{t'_j}}$ with $\symb{t'_j} \neq
  \bot(\symb{t_j})$ for all $j < i_0$ iff 
  \begin{multline*}
    \varphi'_{\all(S)}(\varphi_{S}(\symb{t_1}), \symb{t_2} \conss \ldots
    \conss
    \symb{t_n} \conss \nil, \symb{t_1} \conss \nil, \nil) \\
    \multieval{\TRm{\all(S)} \cup \TRmg{\Gamma}} \\
    \varphi'_{\all(S)}(\varphi_{S}(\symb{t_{i_0}}), \symb{t_{i_0+1}} \conss \ldots
    \conss \symb{t_n} \conss \nil, \symb{t_{i_0}} \conss \ldots \symb{t_1} \conss
    \nil, \symb{t'_{i_0-1}} \conss \ldots \conss
    \symb{t'_1} \conss \nil)
  \end{multline*}
  This latter term can only reduce to 
  \[
  \varphi'_{\all(S)}(\bot{\symb{t_{i_0}}}, \symb{t_{i_0+1}} \conss \ldots \conss
  \symb{t_n} \conss \nil, \symb{t_{i_0}} \conss \ldots \symb{t_1} \conss \nil,
  \symb{t'_{i_0-1}} \conss \ldots \conss \symb{t'_1} \conss \nil)
  \]
  which in turn reduces to $\botlist{\rconcat{\symb{t_{i_0}} \conss \ldots
      \symb{t_1} \conss \nil}{\symb{t_{i_0+1}} \conss \ldots \conss \symb{t_n}
      \conss \nil}}$, and then to $\botlist{\symb{t_1} \conss \ldots \conss
    \symb{t_n} \conss \nil}$, as wished.
\end{proof}

\begin{lem}
  \label{lemma:one-failure}
  Given terms $t_1 \ldots t_n \in \TF$, $S$, and $\Gamma$
  such that $\varS S \subseteq \domS \Gamma$, $\trsred{\TRm{S}
    \cup \TRmg{\Gamma}}{\vpn{S}(\symb {t_i})}{\bot(\symb{t_i})}$ for all $i$ iff 
  \begin{multline*}
    \trsrednl{\TRm{\one(S)} \cup
      \TRmg{\Gamma}}{\varphi'_{\all(S)}(\varphi_{S}(\symb{t_1}), \symb{t_2}
      \conss \ldots \conss \symb{t_n} \conss \nil, \symb{t_1} \conss
      \nil)}{\botlist{\symb{t_1} \conss \ldots \conss \symb{t_n} \conss \nil}}.
  \end{multline*}
\end{lem}

\begin{proof}
  Because $\TRm{\one(S)}$ contains $\TRm{S}$, and by definition of
  $\TRm{\one(S)}$, $\trsred{\TRm{S} \cup \TRmg{\Gamma}}{\vpn{S'}(\symb
    {t_i})}{\bot(\symb{t_i})}$ for all $i$
  iff
  \begin{multline*}
    \varphi'_{\one(S)}(\varphi_{S}(\symb{t_i}), \symb{t_{i+1}} \conss \ldots
    \conss \symb{t_n} \conss \nil, \symb{t_i} \conss \ldots \symb{t_1} \conss
    \nil) \\
    \multieval{\TRm{\one(S)} \cup \TRmg{\Gamma}} 
    \varphi'_{\one(S)}(\varphi_{S}(\symb{t_{i+1}}), \symb{t_{i+2}} \conss \ldots
    \conss \symb{t_n} \conss \nil, \symb{t_{i+1}} \conss \ldots \symb{t_1} \conss
    \nil)
  \end{multline*}
  for all $i$. As a result, we deduce
  \begin{align*}  
    & \varphi'_{\one(S)}(\varphi_{S}(\symb{t_1}), \symb{t_2} \conss \ldots
    \conss
    \symb{t_n} \conss \nil, \symb{t_1} \conss \nil) \\
    \multieval{\TRm{\one(S)} \cup \TRmg{\Gamma}}~&
    \varphi'_{\one(S)}(\varphi_{S}(\symb{t_{n}}), \nil, \symb{t_{n}} \conss
    \ldots \symb{t_1} \conss \nil) \\
    \multieval{\TRm{\one(S)} \cup \TRmg{\Gamma}}~&
    \varphi'_{\one(S)}(\bot(\symb{t_{n}}), \nil, \symb{t_{n}} \conss \ldots
    \symb{t_1} \conss \nil)\\
    \trs_{\TRm{\one(S)} \cup \TRmg{\Gamma}}~& \botlist{\reverse{\symb{t_{n}} \conss
      \ldots \conss \symb{t_1} \conss \nil}}\\
    \multieval{\TRm{\one(S)} \cup \TRmg{\Gamma}}~& \botlist{\symb{t_1} \conss \ldots
    \conss \symb{t_n} \conss \nil}
  \end{align*}
  as wished.
\end{proof}

\begin{lem}
  \label{lemma:one-success}
  Given terms $t_1 \ldots t_n \in \TF$, $S$, and $\Gamma$
  such that $\varS S \subseteq \domS \Gamma$, $\trsred{\TRm{S}
    \cup \TRmg{\Gamma}}{\vpn{S}(\symb {t_j})}{\bot(\symb{t_j})}$ for all $j <
  i_0$ and $\trsred{\TRm{S} \cup \TRmg{\Gamma}}{\vpn{S}(\symb
    {t_{i_0}})}{\symb{t'_{i_0}}}$ with $\symb{t'_{i_0}} \neq
  \bot(\symb{t_{i_0}})$ iff 
  \begin{multline*}
    \trsrednl{\TRm{\one(S)} \cup
    \TRmg{\Gamma}}{\varphi'_{\one(S)}(\varphi_{S}(\symb{t_1}), \symb{t_2} \conss
    \ldots \conss \symb{t_n} \conss \nil, \symb{t_1} \conss \nil)}{\symb{t_1}
    \conss \ldots \conss \symb{t'_{i_0}} \conss \ldots \conss \symb{t_n} \conss
    \nil}.
  \end{multline*}
\end{lem}

\begin{proof}
  As in the proof of Lemma~\ref{lemma:one-failure}, $\trsred{\TRm{S} \cup
    \TRmg{\Gamma}}{\vpn{S}(\symb {t_j})}{\bot(\symb{t_j})}$ for all $j < i_0$ iff 
  \begin{multline*}
    \varphi'_{\one(S)}(\varphi_{S}(\symb{t_1}), \symb{t_2} \conss \ldots
    \conss
    \symb{t_n} \conss \nil, \symb{t_1} \conss \nil) \\
    \multieval{\TRm{\one(S)} \cup \TRmg{\Gamma}} 
    \varphi'_{\one(S)}(\varphi_{S}(\symb{t_{i_0}}), \symb{t_{i_0+1}} \conss \ldots
    \conss \symb{t_n} \conss \nil, \symb{t_{i_0}} \conss \ldots \symb{t_1} \conss
    \nil)
  \end{multline*}
  This latter term can only reduce to 
  \[
  \varphi'_{\one(S)}(\symb{t'_{i_0}}, \symb{t_{i_0+1}} \conss \ldots \conss
  \symb{t_n} \conss \nil, \symb{t_{i_0}} \conss \ldots \symb{t_1} \conss \nil)
  \]
  which in turn reduces to $\rconcat{\symb{t_{i_0}-1} \conss \ldots \symb{t_1}
    \conss \nil}{\symb{t'_{i_0}} \conss \symb{t_{i_0+1}} \conss \ldots \conss
    \symb{t_n} \conss \nil}$, and then to $\symb{t_1} \conss \ldots \conss
  \symb{t'_{i_0}} \conss \ldots \conss \symb{t_n} \conss \nil$, as wished.
\end{proof}

\meta*

\begin{proof}
  The proof is by induction on $S$. Most cases are straightforward, as the
  meta-encoding is the same as the encoding in these cases. We only detail the
  ones with significant changes. In what follows, we write $\RRM$ for $\TRm{S}
  \cup \TRmg{\Gamma}$.
  \begin{enumerate}
  \item $S = \all(S')$. 
    \begin{enumerate}
    \item $\TRr{S} \Rightarrow \TRm{S}$.  
      
      Suppose $t$ is a constant $c$; then $\trsred{\TRctx{\Gamma}{S} \cup
        \TRg{\Gamma}{\Gamma}}{\vpS(c)}{c}$. With the meta-encoding, we have
      \begin{align*}
        \vpn{\all(S')}(\appl{c}{\nil}) & \trs_\RRM
        \propag{\appl{c}{\varphi^{list}_{S}(\nil)}} \\
        & \trs_\RRM \propag{\appl{c}{\nil}} \\
        & \trs_\RRM \appl{c} \nil = \symb c
      \end{align*}
      The result therefore holds.
      
      We now suppose that $t = f(t_1, \ldots, t_n)$. First, we assume $t' \neq
      \bot(t)$. By the definition of $\TRr{S}$, it is possible iff the rule
      \[
    \varphi_{\all(S')}(f(\Xx_1,\ldots,\Xx_n)) \ra
      \psiF(\vpn{S'}(\Xx_1),\ldots,\vpn{S'}(\Xx_n),f(\Xx_1,\ldots,\Xx_n))
      \]
      has been applied, meaning that $\trsred{\TRr{S'} \cup
        \TRgg{\Gamma}}{\vpn{S'}(t_i)}{t'_i}$ with $t'_i \neq \bot(t_i)$ for
      all~$i$. By induction, we have $\trsred{\TRm{S'} \cup
        \TRmg{\Gamma}}{\vpn{S'}(\symb {t_i})}{\symb{t'_i}}$, and $\symb{t'_i}
      \neq \bot(\symb{t_i})$. As a result, we have
      \begin{align*}
        & \vpn{\all(S')}(\appl{f}{\symb {t_1} \conss \ldots \conss \symb{t_n}
          \conss \nil}) \\
        \trs_\RRM~ & 
        \propag{\appl{f}{\varphi^{list}_{S}(\symb{t_1} \conss \ldots \conss
            \symb{t_n} \conss \nil)}}
        \\
        \trs_\RRM~ &  \propag{\appl{f}{\varphi'_{S}(\varphi_{S'}(\symb{t_1}), \symb{t_2} \conss
            \ldots \conss \symb{t_n}, \symb{t_1} \conss \nil, \nil)}}
      \end{align*}
      By Lemma~\ref{lemma:all-success}, we have 
      \begin{align*}
        & \propag{\appl{f}{\varphi'_{S}(\varphi_{S'}(\symb{t_1}), \symb{t_2}
          \conss
          \ldots \conss \symb{t_n}, \symb{t_1} \conss \nil, \nil)}} \\
      \multieval{\RRM}~ & \propag{\appl{f}{\symb{t'_1} \conss \ldots \conss
          \symb{t'_n}}} \\
      \trs_\RRM~ & \appl{f}{\symb{t'_1} \conss \ldots \conss \symb{t'_n}} =
      \symb{t'}
    \end{align*}

    Assume now $t' = \bot(t)$. By the definition of $\TRr S$, it is possible iff
    there exists~$i$ such that $\trsred{\TRr{S'} \cup
      \TRgg{\Gamma}}{\vpn{S'}(t_i)}{\bot{(t_i)}}$. Let $i_0$ be the smallest of
    such $i$; then for all $j < i_0$, we have $\trsred{\TRr{S'} \cup
      \TRgg{\Gamma}}{\vpn{S'}(t_j)}{t'_j}$ with $t'_j \neq \bot(t_j)$. By
    induction, we have $\trsred{\TRm{S'} \cup \TRmg{\Gamma}}{\vpn{S'}(\symb
      {t_j})}{\symb{t'_j}}$ with $\symb{t'_j} \neq \bot(\symb{t_j})$ for $j <
    i_0$, and $\trsred{\TRm{S'} \cup
      \TRmg{\Gamma}}{\vpn{S'}(\symb{t_{i_0}})}{\bot{(\symb{t_{i_0}})}}$. Using
    Lemma~\ref{lemma:all-failure}, we have
    \begin{align*}
      & \vpn{\all(S')}(\appl{f}{\symb {t_1} \conss \ldots \conss \symb{t_n}
        \conss \nil}) \\
      \multieval{\RRM}~ & \propag{\appl{f}{\varphi'_{S}(\varphi_{S'}(\symb{t_1}), \symb{t_2} \conss
          \ldots \conss \symb{t_n}, \symb{t_1} \conss \nil, \nil)}} \\
      \multieval{\RRM}~ & \propag{\appl{f}{\botlist{\symb{t_1} \conss \ldots
            \conss \symb{t_n} \conss \nil}}} \\
      \trs_{\RRM}~ & \bot(\appl f {\symb{t_1} \conss \ldots \conss \symb{t_n}
        \conss \nil}) = \symb{\bot(t)}
    \end{align*}
    \item $\TRr{S} \Rightarrow \TRm{S}$

      If $t$ is a constant $c$, then it is easy to check that $t'' = \trmeta
      c$. Suppose $t = f(t_1, \ldots, t_n)$. First, we assume $t'' \in
      \TFmeta$. It is possible only if the rule 
      \[
      \varphi'_{\all(S')}(\Xx \at \appl{\any}{\any},\nil,\any,\revdone) \ra
      \reverse{\cons{\Xx}{\revdone}}
      \]
      has been applied. It is possible only if rewriting
      $\varphi_{S'}(\trmeta{t_i})$ produces a term $t''_i$ of the form $\appl
      \any \any$ for all $i$, so that $t'' = \appl {\trmeta f}{t''_1 \conss
        t''_2 \conss \ldots \conss t''_n \conss \nil}$. Only the rules of $\TRm
      {S'}$ may apply to $\varphi_{S'}(\trmeta{t_i})$, so we have in fact
      $\trsred{\TRm{S'} \cup \TRmg{\Gamma}}{\varphi_{S'}(\symb {t_i})}{t''_i}$
      with $t''_i \in \TFmeta$. By induction, there exists $t'_i$ such that
      $\symb{t'_i} = t''_i$ and $\trsred{\TRctx{\Gamma}{S'} \cup
        \TRg{\Gamma}{\Gamma}}{\varphi_{S'}(t_i)}{t'_i}$ for all $i$. From that,
      we can deduce that $\trsred{\TRctx{\Gamma}{\all(S')} \cup
        \TRg{\Gamma}{\Gamma}}{\varphi_{\all(S')}(f(t_1, \ldots, t_n))}{f(t'_1, \ldots,
        t'_n)}$, as wished.

      Now assume we obtain $\bot(t'')$ with $t'' \in \TFmeta$. Then the rule
      \[
      \varphi'_{\all(S')}(\bot(\any),\todo,\revtried,\any) \ra \botlist{\rconcat{\revtried}{\todo}}
      \]
      has been applied, meaning that there exists $i$ such that
      $\trsred{\TRm{\all(S')} \cup \TRmg{\Gamma}}{\varphi_{S'}(\symb
        {t_i})}{\bot(t''_i)}$ for some $t''_i$. Only the rules of $\TRm {S'}$
      may apply to $\varphi_{S'}(\trmeta{t_i})$, so we have in fact
      $\trsred{\TRm{S'} \cup \TRmg{\Gamma}}{\varphi_{S'}(\symb
        {t_i})}{\bot(t''_i)}$. By induction, $t''_i = \symb{t_i}$ and
      $\trsred{\TRctx{\Gamma}{S'} \cup
        \TRg{\Gamma}{\Gamma}}{\varphi_{S'}(t_i)}{\bot(t_i)}$. From there, we can
      show that $\trsred{\TRctx{\Gamma}{\all(S')} \cup
        \TRg{\Gamma}{\Gamma}}{\varphi_{\all(S')}(f(t_1, \ldots, t_n))}{\bot(f(t_1,
        \ldots, t_n))}$, as wished.
    \end{enumerate}
  \item $S = \one(S')$. 
        \begin{enumerate}
        \item $\TRr{S} \Rightarrow \TRm{S}$.  
          
          Suppose $t$ is a constant $c$; then
          $\trsred{\TRctx{\Gamma}{S} \cup \TRg{\Gamma}{\Gamma}}{\vpS(c)}{\bot(c)}$. With the
          meta-encoding, we have 
          \begin{align*}
            \vpn{\one(S')}(\appl{c}{\nil}) & \trs_\RRM
            \propag{\appl{c}{\varphi^{list}_{S}(\nil)}} \\
            & \trs_\RRM \propag{\appl{c}{\botlist{\nil}}} \\
            & \trs_\RRM \bot(\appl{c} \nil) = \symb{\bot(c)}
          \end{align*}
          The result therefore holds.
          
          We now suppose that $t = f(t_1, \ldots, t_n)$. First, we assume $t' =
          \bot(t)$. By the definition of $\TRr{S}$, it is possible if
          $\trsred{\TRr{S'} \cup \TRgg{\Gamma}}{\vpn{S'}(t_i)}{\bot(t_i)}$ for
          all~$i$. By induction, we have $\trsred{\TRm{S'} \cup
            \TRmg{\Gamma}}{\vpn{S'}(\symb {t_i})}{\symb{\bot(t_i)}}$, which in
          turn implies 
          \begin{align*}
            & \vpn{\one(S')}(\appl{f}{\symb {t_1} \conss \ldots \conss \symb{t_n}
              \conss \nil}) \\
            \trs_\RRM~ & 
            \propag{\appl{f}{\varphi^{list}_{S}(\symb{t_1} \conss \ldots \conss
                \symb{t_n} \conss \nil)}}
            \\
            \trs_\RRM~ &  \propag{\appl{f}{\varphi'_{S}(\varphi_{S'}(\symb{t_1}), \symb{t_2} \conss
                \ldots \conss \symb{t_n}, \symb{t_1} \conss \nil)}}
          \end{align*}
    By Lemma~\ref{lemma:one-failure}, we have 
    \begin{align*}
      & \propag{\appl{f}{\varphi'_{S}(\varphi_{S'}(\symb{t_1}), \symb{t_2}
          \conss
          \ldots \conss \symb{t_n}, \symb{t_1} \conss \nil)}} \\
      \multieval{\RRM}~ & \propag{\appl{f}{\botlist{\symb{t_1} \conss \ldots \conss
          \symb{t_n}}}} \\
      \trs_\RRM~ & \bot(\appl{f}{\symb{t_1} \conss \ldots \conss \symb{t_n}}) =
      \symb{\bot(t)}
    \end{align*}

    Assume now $t' \neq \bot(t)$. By the definition of $\TRr S$, it is possible
    if there exists~$i_0$ such that $\trsred{\TRr{S'} \cup
      \TRgg{\Gamma}}{\vpn{S'}(t_i)}{t'_{i_0}}$ with $t'_{i_0} \neq
    \bot(t_{i_0})$, and for all $j < i_0$, we have $\trsred{\TRr{S'} \cup
      \TRgg{\Gamma}}{\vpn{S'}(t_j)}{\bot(t_j)}$. By induction, it means that
    $\trsred{\TRm{S'} \cup \TRmg{\Gamma}}{\vpn{S'}(\symb
      {t_j})}{\bot(\symb{t_j})}$ and $\trsred{\TRm{S'} \cup
      \TRmg{\Gamma}}{\vpn{S'}(\symb{t_{i_0}})}{t'_{i_0}}$ with $\symb{t'_{i_0}}
    \neq \bot(\symb{t_{i_0}})$. Using Lemma~\ref{lemma:one-success}, we have
    \begin{align*}
      & \vpn{\one(S')}(\appl{f}{\symb {t_1} \conss \ldots \conss \symb{t_n}
        \conss \nil}) \\
      \multieval{\RRM}~ &
      \propag{\appl{f}{\varphi'_{S}(\varphi_{S'}(\symb{t_1}), \symb{t_2} \conss
          \ldots \conss \symb{t_n}, \symb{t_1} \conss \nil)}} \\
      \multieval{\RRM}~ & \propag{\appl{f}{\symb{t_1} \conss \ldots \conss
          \symb{t'_{i_0}} \conss \ldots    \conss \symb{t_n} \conss \nil}} \\
    \trs_{\RRM}~ & \appl{f}{\symb{t_1} \conss \ldots \conss \symb{t'_{i_0}}
      \conss \ldots \conss \symb{t_n} \conss \nil} = \symb{t'}
    \end{align*}
    \item $\TRr{S} \Rightarrow \TRm{S}$

      If $t$ is a constant $c$, then it is easy to check that the meta-encoding
      fails with $\bot(\trmeta c)$. Suppose $t = f(t_1, \ldots, t_n)$.  First,
      we assume we obtain $\bot(t'')$ with $t'' \in \TFmeta$. Then the rule
      \[
      \varphi'_{\one(S')}(\bot(\any),\nil, \revtried) \ra \botlist{\reverse{\revtried}}
      \]
      has been applied, meaning that for all $i$, we have
      $\trsred{\TRm{\one(S')} \cup \TRmg{\Gamma}}{\varphi_{S'}(\symb
        {t_i})}{\bot(t''_i)}$ for some $t''_i$. Only the rules of $\TRm {S'}$
      may apply to $\varphi_{S'}(\trmeta{t_i})$, so we have in fact
      $\trsred{\TRm{S'} \cup \TRmg{\Gamma}}{\varphi_{S'}(\symb
        {t_i})}{\bot(t''_i)}$ for all $i$. By induction, $t''_i = \symb{t_i}$ and
      $\trsred{\TRctx{\Gamma}{S'} \cup
        \TRg{\Gamma}{\Gamma}}{\varphi_{S'}(t_i)}{\bot(t_i)}$. From there, we can
      show that $\trsred{\TRctx{\Gamma}{\one(S')} \cup
        \TRg{\Gamma}{\Gamma}}{\varphi_{\one(S')}(f(t_1, \ldots, t_n))}{\bot(f(t_1,
        \ldots, t_n))}$, as wished.

      Now we assume $t'' \in \TFmeta$. It is possible only if the rule
      \[
      \varphi'_{\one(S')}(\Xx \at
      \appl{\any}{\any},\todo,\cons{\any}{\revtried}) \ra
      \rconcat{\revtried}{\cons{\Xx}{\todo}}
      \]
      has been applied. It is possible only if there exists $i$ such that
      $\varphi_{S'}(\trmeta{t_i})$ produces a term $t''_i$ of the form $\appl
      \any \any$, and for all $j < i$, we obtain $\bot(t''_j)$. As a result, we
      have $t'' = \appl {\trmeta f}{\trmeta{t_1} \conss \trmeta{t_2} \conss
        \ldots \conss t''_i \conss \ldots \conss \trmeta{t_n} \conss \nil}$. Only the
      rules of $\TRm {S'}$ may apply to $\varphi_{S'}(\trmeta{t_i})$, so we have
      in fact $\trsred{\TRm{S'} \cup \TRmg{\Gamma}}{\varphi_{S'}(\symb
        {t_i})}{t''_i}$ with $t''_i \in \TFmeta$ and $\trsred{\TRm{S'} \cup
        \TRmg{\Gamma}}{\varphi_{S'}(\symb {t_j})}{\bot(t''_j)}$ for all
      $j<i$. By induction, there exists $t'_i$ such that $\symb{t'_i} = t''_i$
      and $\trsred{\TRctx{\Gamma}{S'} \cup
        \TRg{\Gamma}{\Gamma}}{\varphi_{S'}(t_i)}{t'_i}$, and also $t''_j =
      \bot(\trmeta{t_j})$, $\trsred{\TRctx{\Gamma}{S'} \cup
        \TRg{\Gamma}{\Gamma}}{\varphi_{S'}(t_j)}{\bot(t_j)}$ for all $j<i$. From
      that, we can deduce that $\trsred{\TRctx{\Gamma}{\one(S')} \cup
        \TRg{\Gamma}{\Gamma}}{\varphi_{\one(S')}(f(t_1, \ldots, t_n))}{f(t_1,
        \ldots, t'_i, \ldots, t'_n)}$, as wished.\qedhere
    \end{enumerate}
  \end{enumerate}
\end{proof}

\section{Properties of the many-sorted encoding}
\label{se:propertiesSortedAppendix}

\begin{figure}[!tp]
\begin{mathpar}
\begin{array}{lr@{\hspace{5pt}}c@{\hspace{5pt}}l}
\mathbf{(SE1)}&
  \TRctxS{\Gamma}{\id} & =  & \bigcup\limits_{\sst\in\SO} \{  ~  \varphi_{\id}(\Xx \at \apt{\sst} \bott{\sst}(\any)  ) \ra \Xx,
  \quad
  \varphi_{\id}(\bot(\Xxs{}{\sst})) \ra \bot(\Xxs{}{\sst})~\} \\[+4pt]
\mathbf{(SE2)}&
  \TRctxS{\Gamma}{\fail} & = &   \bigcup\limits_{\sst\in\SO} 
  \{ ~ \varphi_{\fail}(\Xx \at \apt{\sst} \bott{\sst}(\any)  ) \ra \bot(\Xx),
  \quad
  \varphi_{\fail}(\bot(\Xxs{}{\sst})) \ra \bot(\Xxs{}{\sst}) ~\} \\ [+4pt]
\mathbf{(SE3)}&
  \TRctxS{\Gamma}{l \ra r} & = & \{ ~ \varphi_{l\ra r}(l) \ra r ~\} \\[+2pt]
                            &\bigcup\limits_{\sst\in\SO} & \{&    ~
                            \varphi_{l\ra r}(\Xx \at \apt{\sst} l) \ra \bot(\Xx),
                            \quad 
                            \varphi_{l\ra r}(\bot(\Xxs{}{\sst})) \ra \bot(\Xxs{}{\sst}) ~\} 
                            \\[+4pt]
\mathbf{(SE4)}&
  \TRctxS{\Gamma}{\seq{S_1}{S_2}} & = &  \TRctxS{\Gamma}{S_1} ~\cup~ \TRctxS{\Gamma}{S_2}  \\[+2pt]
                          & \bigcup\limits_{\sst\in\SO} & \{ &  
                          \varphi_{\seqsub{S_1}{S_2}}(\Xx \at \apt{\sst}
                          \bott{\sst}(\any)  ) \ra \phiseq( \vpt(\vpo(\Xx)),
                          \Xx),
                          \\[+2pt]
                          & & &                                                    
                          \varphi_{\seqsub{S_1}{S_2}}(\bot(\Xxs{}{\sst})) \ra \bot(\Xxs{}{\sst}),
                          \\[+2pt]
                          &&&\phiseq(\Xx \at \apt{\sst} \bott{\sst}(\any), \Yys{}{\sst}) \ra  \Xx, 
                           \quad
                          \phiseq(\bot(\Yys{}{\sst}), \Xxs{}{\sst}) \ra  \bot(\Xxs{}{\sst}) ~\} 
                          \\[+4pt]
\mathbf{(SE5)}&
  \TRctxS{\Gamma}{\choice{S_1}{S_2}} & =  & \TRctxS{\Gamma}{S_1} ~\cup~ \TRctxS{\Gamma}{S_2}  \\[+2pt]
                           &\bigcup\limits_{\sst\in\SO}  & \{ &  
                           \varphi_{\choicesub{S_1}{S_2}}(\Xx \at \apt{\sst} \bott{\sst}(\any)  ) \ra  \phich(\vpo(\Xx)),
                           \quad
                           \varphi_{\choicesub{S_1}{S_2}}(\bot(\Xxs{}{\sst})) \ra \bot(\Xxs{}{\sst}),
                           \\[+2pt]
                           &&&
                           \phich(\bot(\Xxs{}{\sst})) \ra  \vpt(\Xxs{}{\sst}),
                           \quad
                           \phich(\Xx \at \apt{\sst} \bott{\sst}(\any)) \ra  \Xx ~\}
                           \\[+4pt]
\mathbf{(SE6)}&
  \TRctxS{\Gamma}{\mu X\dotsym S} & =  & \TRctxS{\Gamma; {\Xx\colon S}}{S} \\[+2pt]
                 &\bigcup\limits_{\sst\in\SO}  & \{ &  
                 \varphi_{\mu X\dotsym S}(\Xx \at \apt{\sst} \bott{\sst}(\any)) \ra  \vpS(\Xx),
                 \quad
                 \varphi_{\mu X\dotsym S}(\bot(\Xxs{}{\sst})) \ra \bot(\Xxs{}{\sst}),
                 \\[+2pt]
                 & & &  
                 \varphi_{X}(\Xx \at \apt{\sst} \bott{\sst}(\any)) \ra  \vpS(\Xx),
                 \quad
                 \varphi_{X}(\bot(\Xxs{}{\sst})) \ra \bot(\Xxs{}{\sst})  ~\} 
                 \\[+4pt]
\mathbf{(SE7)}&
  \TRctxS{\Gamma; {\Xx\colon S}}{X} & =  & \emptyset
  \\[+4pt]
\mathbf{(SE8)}&
  \TRctxS{\Gamma}{\all(S)} & =  & \TRctxS{\Gamma}{S}  \\[+2pt]
                 &\bigcup\limits_{\sst\in\SO} & \{ &  
                  \varphi_{\all(S)}(\bot(\Xxs{}{\sst})) \ra \bot(\Xxs{}{\sst})  ~~\} 
                 \bigcup\limits_{c\in \FF^{0}}  \{ \varphi_{\all(S)}(c) \ra  c  ~~\}
                  \\[+2pt]
                  \multicolumn{2}{r}{\bigcup\limits_{f\in\FF^{\supzero}}} &\{ &  
                  \varphi_{\all(S)}(f(\Xxs{1}{\sst_1},\ldots,\Xxs{n}{\sst_n})) \ra  \\[-7pt]
                  &&&
                  \multicolumn{1}{r}{\psiF(\vpS(\Xxs{1}{\sst_1}),\ldots,\vpS(\Xxs{n}{\sst_n}),f(\Xxs{1}{\sst_1},\ldots,\Xxs{n}{\sst_n})),}
                  \\[+2pt]
                  &&&  \psiF(\Xx_1\at\apt{\sst_1}\bott{\sst_1}(\any),\ldots,\Xx_n\at\apt{\sst_n}\bott{\sst_n}(\any),\Yys{}{\sst}) \ra  f(\Xx_1,\ldots,\Xx_n), \\[+2pt]
                  &&&  \psiF(\bot(\Xxs{1}{\sst_1}),\Xxs{2}{\sst_2},\ldots,\Xxs{n}{\sst_n},\Xxs{}{\sst}) \ra  \bot(\Xxs{}{\sst}) , \\[+2pt]
                  &&&  \vdots  \\[+2pt]
                  &&&  \psiF(\Xxs{1}{\sst_1},\Xxs{2}{\sst_2},\ldots,\bot(\Xxs{n}{\sst_n}),\Xxs{}{\sst}) \ra  \bot(\Xxs{}{\sst})  ~~\}
                  \\[+4pt]
\mathbf{(SE9)}&
  \TRctxS{\Gamma}{\one(S)} & = & \TRctxS{\Gamma}{S} \\[+2pt]
                &\bigcup\limits_{\sst\in\SO} & \{ &  
                \varphi_{\one(S)}(\bot(\Xxs{}{\sst})) \ra \bot(\Xxs{}{\sst}) ~~\}
                                                    \bigcup\limits_{c\in \FF^{0}}  \{  
                                                    \varphi_{\one(S)}(c) \ra  \bot(c)  ~~\}\\[+2pt]
                &\bigcup\limits_{f\in\FF^{\supzero}} & \{ &  
                \varphi_{\one(S)}(f(\Xxs{1}{\sst_1},\ldots,\Xxs{n}{\sst_n})) \ra  \psif{1}(\vpS(\Xxs{1}{\sst_1}),\Xxs{2}{\sst_2},\ldots,\Xxs{n}{\sst_n}) ~~\} \\[+2pt]
                \multicolumn{2}{r}{\bigcup\limits_{f\in\FF^{\supzero}}~\bigcup\limits_{1\leq i\leq n}} & \{ &  
                \psif{i}(\bot(\Xxs{1}{\sst_1}),\ldots,\bot(\Xxs{i-1}{\sst_{i-1}}),\Xx_i\at\apt{\sst_i}\bott{\sst_i}(\any),\Xxs{i+1}{\sst_{i+1}},\ldots,\Xxs{n}{\sst_n}) \ra\\[-7pt] 
                &&&
                \multicolumn{1}{r}{f(\Xxs{1}{\sst_1},\ldots,\Xxs{n}{\sst_n}) ~~\} } \\[+2pt]
                \multicolumn{2}{r}{\bigcup\limits_{f\in\FF^{\supzero}}~\bigcup\limits_{1\leq i<n}} &  \{ &  
                \psif{i}(\bot(\Xxs{1}{\sst_1}),\ldots,\bot(\Xxs{i}{\sst_{i}}),\Xxs{i+1}{\sst_{i+1}},\ldots,\Xxs{n}{\sst_{n}}) \ra\\[-7pt] 
                &&&\multicolumn{1}{r}{\psif{i+1}(\bot(\Xxs{1}{\sst_1}),\ldots,\bot(\Xxs{i}{\sst_i}),\vpS(\Xxs{i+1}{\sst_{i+1}}),\Xxs{i+2}{\sst_{i+2}},\ldots,\Xxs{n}{\sst_n})  ~~\}} \\[+2pt]
                \multicolumn{2}{r}{\bigcup\limits_{f\in\FF^{\supzero}}} & \{ &  \psif{n}(\bot(\Xxs{1}{\sst_1}),\ldots,\bot(\Xxs{n}{\sst_n})) \ra  \bot(f(\Xxs{1}{\sst_1},\ldots,\Xxs{n}{\sst_n})) ~~\}
                \\[+4pt]
                 \\[+4pt]
  \mathbf{(SE10)}&
  \TRgS{\Gamma'}{\Gamma; {X\colon S}} & =  &  \TRgS{\Gamma'}{\Gamma} ~\cup~ \TRctxS{\Gamma'}{S}             \\[+2pt]
                           &\bigcup\limits_{\sst\in\SO}  &  \{  & \varphi_{X}(\Xx \at \apt{\sst} \bott{\sst}(\any)) \ra  \vpS(\Xx),
                           \quad
                           \varphi_{X}(\bot(\Xxs{}{\sst})) \ra \bot(\Xxs{}{\sst})   ~\}
  \\[+4pt]
  \mathbf{(SE11)}&
  \TRgS{\Gamma'}{\emptyctx} & =  &  \emptyset
\end{array}
\end{mathpar}
\caption{\label{fig:encodingTyped}Sorted strategy translation. We
  consider that
  $f{\sigd}\sst_1\prosep\ldots\prosep\sst_n{\sarrow}\sst\in\FFs{\sst}^{\supzero}$.
The original signature is enriched with
$\bot,$
$\varphi_{\id}, \varphi_{\fail},$
$\varphi_{l\ra r},$
$\varphi_{\seqsub{S_1}{S_2}},$
$\phiseq,$
$\varphi_{\choicesub{S_1}{S_2}},$
$\phich,$
$\varphi_{\mu X\dotsym S},$
$\varphi_{X},$
$\varphi_{\all(S)},$
$\varphi_{\one(S)}{\sigd}\sst{\sarrow}\sst$, 
and 
$\psiF{\sigd}\sst_1\prosep\ldots\prosep\sst_n\prosep\sst {\sarrow} \sst$, and 
$\psif{i+1}{\sigd}\sst_1\prosep\ldots\prosep\sst_n {\sarrow} \sst$.
}
\end{figure}

\subRed*
\begin{proof}
  By induction on the structure of the strategy $S$ (and the context
  $\Gamma$) and by cases on the rewrite rule applied in the reduction.
  We have that $t \longrightarrow_{\RR} t'$ iff there exist a rule
  $l\ra r\in\RR$, $\omega\in\PPos(t)$, and a well-sorted substitution
  $\sigma$ such that $\stt t \omega = \sigma(l)$ and $t' = \rmp t
  \omega {\sigma(r)}$; we write $t \longrightarrow_{l\ra r} t'$ to
  explicit the applied rule. We first show by cases on the rewrite
  rule applied in the reduction that the property holds for
  $\omega=\varepsilon$. It is then easy to conclude by observing that
  for any term $\rmp{t}{\omega}{u}{\sigd}s$ with $u{\sigd}s'$ we also have
  $\rmp{t}{\omega}{u'}{\sigd}s$ whenever $u'{\sigd}s'$.

  ~\\
  \emph{Base case:} the applied rewrite rule is one the  rules
  in $\TRctxS{\Gamma}{\id}$, $\TRctxS{\Gamma}{\id}$ or
  $\TRctxS{\Gamma}{l \ra r}$.

\begin{enumerate}
\item \label{case:id}
  $t\longrightarrow_{\rho}t'$  with $\rho\in\TRctxS{\Gamma}{\id}$
\begin{enumerate} 
  \item \label{case:apBot}
    $\rho$ corresponds to a rule schema of the form $\varphi_{\id}(\Xx \at \apt{\sst} \bott{\sst}(\any)  ) \ra \Xx$.
    In this case $t=\varphi_{\id}(u)$ with $u{\sigd}\sst$ and there exist
    $v\in\apt{\sst} \bott{\sst}(\any)$ and $\sigma$ s.t. $\sigma(v)=u$. Then
    $t'=u$ and thus $t'{\sigd}\sst$.
  \item  \label{case:bot}
    $\rho$ is a rule of the form  $\varphi_{\id}(\bot(\Xxs{}{\sst})) \ra \bot(\Xxs{}{\sst})$.
    In this case $t=\varphi_{\id}(\bot(u))$ with $u{\sigd}\sst$ and there
    exists $\sigma$ s.t. $\sigma(\Xxs{}{\sst})=u$.  We have thus
    $t'=\bot(u)$ and $t'{\sigd}\sst$.
\end{enumerate}

\item $t\longrightarrow_{\rho}t'$  with $\rho\in\TRctxS{\Gamma}{\fail}$
\begin{enumerate} 
  \item $\rho$ corresponds to a rule schema of the form $\varphi_{\fail}(\Xx \at \apt{\sst} \bott{\sst}(\any)  ) \ra \bot(\Xx)$.
    In this case $t=\varphi_{\fail}(u)$ with $u{\sigd}\sst$ and there exist
    $v\in\apt{\sst} \bott{\sst}(\any)$ and $\sigma$ s.t. $\sigma(v)=u$. Then
    $t'=\bot(u)$ and $t'{\sigd}\sst$.
  \item $\rho$ is a rule of the form  $\varphi_{\fail}(\bot(\Xxs{}{\sst})) \ra \bot(\Xxs{}{\sst})$.
    Similar to case~\ref{case:bot}.
\end{enumerate}

\item $t\longrightarrow_{\rho}t'$  with $\rho\in\TRctxS{\Gamma}{l \ra r}$
\begin{enumerate} 
  \item $\rho$ is a rule  of the form $\varphi_{l\ra r}(l) \ra r$.
    In this case $t=\varphi_{l\ra r}(u)$ and $u{\sigd}\sst$. Then $l,r{\sigd}\sst$ and
    there exists $\sigma$ s.t. $\sigma(l)=u$. Then $t'=\sigma(r)$ and
    $t'{\sigd}\sst$.
  \item $\rho$ corresponds to a rule schema of the form  $\varphi_{l\ra r}(\Xx \at \apt{\sst} l) \ra  \bot(\Xx)$.
    In this case $t=\varphi_{l\ra r}(u)$ with $u{\sigd}\sst$ and there exist
    $v\in\apt{\sst} l$, $v{\sigd}\sst$ and $\sigma$ s.t. $\sigma(v)=u$. Then
    $t'=\bot(u)$ and thus $t'{\sigd}\sst$.

  \item $\rho$ is a rule of the form  $\varphi_{l\ra r}(\bot(\Xxs{}{\sst})) \ra \bot(\Xxs{}{\sst})$.
    Similar to case~\ref{case:bot}.
\end{enumerate}

\end{enumerate}
~\\
\emph{Induction case:} the applied rewrite rule is one the rules in
the encoding of a strategy other than $\id$, $\fail$ or rewrite rule.

\begin{enumerate}\addtocounter{enumi}{3}
\item  \label{case:seq}
$t\longrightarrow_{\rho}t'$  with $\rho\in\TRctxS{\Gamma}{\seq{S_1}{S_2}}$
\begin{enumerate} 
  \item \label{case:seqOK}
    $\rho$ corresponds to a rule schema of the form $\varphi_{\seqsub{S_1}{S_2}}(\Xx \at \apt{\sst} \bott{\sst}(\any)  ) \ra \phiseq( \vpt(\vpo(\Xx)), \Xx)$.
    In this case $t=\varphi_{\seqsub{S_1}{S_2}}(u)$ with $u{\sigd}\sst$ and
    there exist $v\in\apt{\sst} \bott{\sst}(\any)$ and $\sigma$
    s.t. $\sigma(v)=u$. Then $t'=\phiseq(\vpt(\vpo(u)),u)$ and since
    $\vpt,\vpo{\sigd}\sst{\sarrow}\sst$ and $\phiseq{\sigd}\sst,\sst{\sarrow}\sst$ we
    have $t'{\sigd}\sst$.
  \item $\rho$ corresponds to a rule schema of the form  $\phiseq(\Xx \at \apt{\sst} \bott{\sst}(\any), \Yys{}{\sst}) \ra  \Xx$.
    In this case $t=\phiseq(u',u'')$ with $u',u''{\sigd}\sst$. Then $t'=u'$
    and thus $t'{\sigd}\sst$ using the same reasoning as in
    case~\ref{case:apBot}.
  \item $\rho$ is a rule of the form
    $\varphi_{\seqsub{S_1}{S_2}}(\bot(\Xxs{}{\sst})) \ra
    \bot(\Xxs{}{\sst})$.
    Similar to case~\ref{case:bot}.

  \item $\rho$ is a rule of the form or $\phiseq(\bot(\Yys{}{\sst}),
    \Xxs{}{\sst}) \ra \bot(\Xxs{}{\sst})$.
    In this case $t=\phiseq(\bot(u'),u'')$ with $u',u''{\sigd}\sst$. Then
    $t'=\bot(u'')$ and thus $t'{\sigd}\sst$ using the same reasoning as in
    case~\ref{case:apBot}.

  \item $\rho$ is one of the rules in $\TRctxS{\Gamma}{S_1} ~\cup~
    \TRctxS{\Gamma}{S_2}$. Apply th induction hypothesis.
\end{enumerate}

\item $t\longrightarrow_{\rho}t'$  with $\rho\in\TRctxS{\Gamma}{\choice{S_1}{S_2}}$.
  Each of the possible cases is similar to one of the cases for~\ref{case:seq}.

\item   \label{case:mu}
  $t\longrightarrow_{\rho}t'$  with $\rho\in\TRctxS{\Gamma}{\mu X\dotsym S}$.
  Each of the possible cases is similar to one of the cases for~\ref{case:seq}.

\item  \label{case:all}
$t\longrightarrow_{\rho}t'$  with $\rho\in\TRctxS{\Gamma}{\all(S)}$
\begin{enumerate} 
  \item $\rho$ is a rule of the form $\varphi_{\all(S)}(\bot(\Xxs{}{\sst})) \ra \bot(\Xxs{}{\sst})$.
    Similar to case~\ref{case:bot}.
  \item $\rho$ is a rule of the form $\varphi_{\all(S)}(c) \ra  c$.
    In this case $t=\varphi_{\all(S)}(c)$ with $c{\sigd}\sst$ and $t'=c$
    and $t'{\sigd}\sst$.
  \item $\rho$ has the form  $\varphi_{\all(S)}(f(\Xxs{1}{\sst_1},\ldots,\Xxs{n}{\sst_n})) \ra  \psiF(\vpS(\Xxs{1}{\sst_1}),\ldots,\vpS(\Xxs{n}{\sst_n}),f(\Xxs{1}{\sst_1},\ldots,\Xxs{n}{\sst_n}))$.
    In this case $t=\varphi_{\all(S)}(f(u_1,\ldots,u_n))$ with
    $u_1{\sigd}\sst_1,\ldots,u_n{\sigd}\sst_n$ if
    $f{\sigd}\sst_1\prosep\ldots\prosep\sst_n{\sarrow}\sst$.  Since
    $\vpS{\sigd}\sst'{\sarrow}\sst'$ for any $\sst'\in\SO$ and $\psiF {\sigd}
    \sst_1\prosep\ldots\prosep\sst_n\prosep\sst {\sarrow} \sst$ we have $t'{\sigd}\sst$.
  \item $\rho$ corresponds to a rule schema of the form
    $\psiF(\Xx_1\at\apt{\sst_1}\bott{\sst_1}(\any),\ldots,\Xx_n\at\apt{\sst_n}\bott{\sst_n}(\any),\Yys{}{\sst}) \ra  f(\Xx_1,\ldots,\Xx_n)$.
    In this case $t=\psiF(u_1,\ldots,u_n,u)$ with
    $u_1{\sigd}\sst_1,\ldots,u_n{\sigd}\sst_n,u{\sigd}\sst$ if
    $f{\sigd}\sst_1\prosep\ldots\prosep\sst_n{\sarrow}\sst$ and for all $i=1\ldots n$
    there exist $v_i\in\apt{\sst_i} \bott{\sst_i}(\any)$ and $\sigma$
    s.t. $\sigma(v_i)=u_i$.  Then $t'=f(u_1,\ldots,u_n)$ and thus
    $t'{\sigd}\sst$.
  \item $\rho$ is a rule  of the form or $\psiF(\Xxs{1}{\sst_1},\ldots,\bot(\Xxs{i}{\sst_i}),\ldots,\Xxs{n}{\sst_n},\Xxs{}{\sst}) \ra  \bot(\Xxs{}{\sst})$.
    In this case $t=\psiF(u_1,\ldots,\bot(u_i),\ldots,u_n,u)$ with
    $u_1{\sigd}\sst_1,\ldots,u_n{\sigd}\sst_n,u{\sigd}\sst$ if
    $f{\sigd}\sst_1\prosep\ldots\prosep\sst_n{\sarrow}\sst$.
    Then $t'=\bot(u)$ and thus $t'{\sigd}\sst$.

  \item $\rho$ is one of the rules in $\TRctxS{\Gamma}{S}$. Apply th induction hypothesis.
\end{enumerate}

\item $t\longrightarrow_{\rho}t'$  with $\rho\in\TRctxS{\Gamma}{\one(S)}$.
  Each of the  possible cases is similar to one of the cases for~\ref{case:all}.

\item $t\longrightarrow_{\rho}t'$  with $\rho\in\TRgS{\Gamma'}{\Gamma; {X\colon S}}$.
  Each of the possible cases is similar to one of the cases
  for~\ref{case:mu}, for example.\qedhere
\end{enumerate}
\end{proof}

\equivalenceSorted*
\begin{proof}
  Every rewrite rule in $\RR_{\SO}$ is also included in $\RR$ and
  thus, if $t\longrightarrow_{\RR_{\SO}}t'$ then
  $t\longrightarrow_{\RR}t'$.

  For the other direction we proceed by induction on the structure of
  the strategy $S$ (and the context $\Gamma$) and by cases on the
  rewrite rule applied in the reduction.  For simplicity, we consider
  reduction at the top position but as shown in the proof of
  Lemma~\ref{th:subRed} this generalises immediately for any
  application position. We consider in what follows that
  $t\in\TFsf{\sst}{{\FFgen{}}}$ for some $\sst\in\SO$.

  ~\\
  \emph{Base case:} the applied rewrite rule is one the  rules
  in $\TRctx{\Gamma}{\id}$, $\TRctx{\Gamma}{\id}$ or
  $\TRctx{\Gamma}{l \ra r}$.

  \begin{enumerate}
  \item \label{caseEq:id}
    $t\longrightarrow_{\rho}t'$  with $\rho\in\TRctx{\Gamma}{\id}$
    \begin{enumerate} 
    \item $\rho$ is one of the rules corresponding to the rule schema
      $\varphi_{\id}(\Xx \at \ap \bot(\any) ) \ra \Xx$ and thus of the
      form $\varphi_{\id}(f(\Xx_1,\ldots,\Xx_n)) \ra
      f(\Xx_1,\ldots,\Xx_n)$ with $f\in\FF$.  In this case
      $t=\varphi_{\id}(f(t_1,\ldots,t_n))$ and thus $f\in\FFs{\sst}$.
      This rule also corresponds to the rule schema $\varphi_{\id}(\Xx
      \at \apt{\sst} \bott{\sst}(\any) ) \ra \Xx$ from the sorted translation
      and thus $t\longrightarrow_{\RR_{\SO}}t'$.
    \item \label{caseEq:idBot} $\rho$ is the rule
      $\varphi_{\id}(\bot(\Xx)) \ra \bot(\Xx)$. The rule
      $\varphi_{\id}(\bot(\Xxs{}{\sst})) \ra \bot(\Xxs{}{\sst})$ from
      the sorted translation applies also to $t$ which has necessarily
      the form $\varphi_{\id}(\bot(u))$ with $u{\sigd}\sst$ and thus
      $t\longrightarrow_{\RR_{\SO}}t'$.
    \end{enumerate}

  \item $t\longrightarrow_{\rho}t'$ with
    $\rho\in\TRctx{\Gamma}{\fail}$. We proceed similarly as for the
    case~\ref{caseEq:id}.

  \item $t\longrightarrow_{\rho}t'$  with $\rho\in\TRctx{\Gamma}{l \ra r}$
    \begin{enumerate} 
    \item $\rho$ is one of the rules corresponding to the rule schema
      $\varphi_{l\ra r}(\Xx \at \ap l) \ra \bot(\Xx)$ and thus of the
      form $\varphi_{l\ra r}(f(u_1,\ldots,u_n)) \ra
      \bot(f(u_1,\ldots,u_n))$ with $f\in\FF$.  In this case
      $t=\varphi_{l\ra r}(f(t_1,\ldots,t_n))$ and thus
      $f\in\FFs{\sst}$. Consequently, the rule $\rho$ corresponds also
      to the rule schema $\varphi_{l\ra r}(\Xx \at \apt{\sst} l) \ra
      \bot(\Xx)$ from the sorted translation and thus
      $t\longrightarrow_{\RR_{\SO}}t'$.
    \item $\rho$ is the rule $\varphi_{l\ra r}(l) \ra r$.  Since the
      rule is also present in the sorted translation then
      $t\longrightarrow_{\RR_{\SO}}t'$.
    \item $\rho$ is the rule $\varphi_{l\ra r}(\bot(\Xx)) \ra
      \bot(\Xx)$. Similar to the case~\ref{caseEq:idBot}.
    \end{enumerate}

\end{enumerate}
~\\
\emph{Induction case:} the applied rewrite rule is one the rules in
the encoding of a strategy other than $\id$, $\fail$ or rewrite rule.
When the applied rule is not in $\TRctx{\Gamma}{\all(S)}$ or
$\TRctx{\Gamma}{\one(S)}$
we can conclude either by induction or similarly to one of the above
cases.

\begin{enumerate}\addtocounter{enumi}{3}
\item \label{caseEq:all} $t\longrightarrow_{\rho}t'$ with
  $\rho\in\TRctx{\Gamma}{\all(S)}$. Once again we can proceed either
  by induction or similarly to one of the above cases except for one
  case:
  \begin{enumerate} 
  \item 
    $\rho$ corresponds to a rule schema
    $\psiF(\Xx_1\at\ap\bot(\any),\ldots,\Xx_n\at\ap\bot(\any),\any)
    \ra f(\Xx_1,\ldots,\Xx_n)$ and thus is a rewrite rule of the form
    $\psiF(f_1(\Xx_1^1,\ldots,\Xx_m^1),\ldots,f_n(\Xx_1^n,\ldots,\Xx_p^n),\Xx)\ra$
    $f(f_1(\Xx_1^1,\ldots,\Xx_m^1),\ldots,f_n(\Xx_1^n,\ldots,\Xx_p^n))$
    with $f_1,\ldots,f_n,f\in\FF$. Since the rewrite rule can be applied to
    the well-sorted term $t$ then
    $f_1\in\FFs{\sst_1},\ldots,f_n\in\FFs{\sst_n},$ $f\in\FFs{\sst}$ and
    consequently, this rewrite rule also corresponds to the rule schema
    $\psiF(\Xx_1\at\apt{\sst_1}\bott{\sst_1}(\any),\ldots,\Xx_n\at\apt{\sst_n}\bott{\sst_n}(\any),\Yys{}{\sst})
    \ra f(\Xx_1,\ldots,\Xx_n)$ from the sorted translation. Thus
    $t\longrightarrow_{\RR_{\SO}}t'$.
  \end{enumerate}
\item $t\longrightarrow_{\rho}t'$ with
  $\rho\in\TRctxS{\Gamma}{\one(S)}$. Similar to the cases for~\ref{caseEq:all}.\qedhere
\end{enumerate}
\end{proof}

\simulationSorted*
\begin{proof}
Follows immediately from Lemma~\ref{th:equivalenceSorted} and Theorem~\ref{th:simulation}.
\end{proof}

\end{document}